\documentclass[aps,pre,reprint,nobibnotes,amsmath,amssymb,longbibliography,floatfix]{revtex4-1}
\usepackage{graphicx,color}
\usepackage{dcolumn}
\usepackage{bm}
\usepackage{array,url}
\usepackage{algorithm,algcompatible}
\usepackage{hyperref,xcolor}
\usepackage[mathlines]{lineno}
\usepackage{caption}
\DeclareCaptionLabelFormat{continued}{#1~#2 (Continued)}
\newtheorem{prop}{Proposition}
\newenvironment{proof}{\paragraph*{Proof:}}{\hfill$\square$\\}

\newcommand\rar{\rightarrow}
\newcommand\vor{\zeta}

\hypersetup{
    colorlinks,
    linkcolor={red},
    citecolor={blue},
    urlcolor={blue!80!black}
}

\begin{document}

\preprint{APS/123-QED}

\title{Study of dynamics in post-transient flows\\using Koopman mode decomposition}

\author{Hassan Arbabi}
 \email{harbabi@engineering.ucsb.edu}
\author{Igor Mezi\'{c}}%
 \email{mezic@engineering.ucsb.edu}
\affiliation{%
 University of California, Santa Barbara\\
 Santa Barbara, CA, 93106, USA
}%

\date{\today}

\begin{abstract}
The Koopman Mode Decomposition (KMD) is a data-analysis technique which is often used to extract the spatio-temporal patterns of complex flows. In this paper, we use KMD to study the dynamics of the lid-driven flow in a two-dimensional square cavity based on theorems related to the spectral theory of the Koopman operator. We adapt two algorithms, from the classical Fourier and power spectral analysis, to compute the discrete and continuous spectrum of the Koopman operator for the post-transient flows.
Properties of the Koopman operator spectrum are linked to the sequence of flow regimes occurring between $Re=10000$ and $Re=30000$, and changing the flow nature from steady to aperiodic. The Koopman eigenfunctions for different flow regimes, including flows with mixed spectra, are constructed using the assumption of ergodicity in the state space.
The associated Koopman modes show remarkable robustness even as the temporal nature of the flow is changing substantially.
We observe that KMD outperforms the proper orthogonal decomposition in reconstruction of the flows with strong quasi-periodic components.
\end{abstract}

\pacs{Valid PACS appear here}
\maketitle

\section{Introduction}

In 1931, Bernard Koopman offered a new formulation of Hamiltonian mechanics based on the theory of Hilbert spaces \cite{koopman1931}. In his formulation, the central object was a linear transformation, called the Koopman operator, which described the time evolution of observations on a Hamiltonian system.
The potential applications of Koopman's work to general theory of dynamical systems went mostly unrecognized for a long time, until it was brought to attention, and further developed in the context of spectral analysis and dissipative systems \cite{mezic2004comparison,mezic2005}.
In particular, Ref. \onlinecite{mezic2005} presented a linear expansion to describe the temporal evolution of observables on a nonlinear dynamical system in terms of the eigenfunctions and eigenvalues of the Koopman operator. This expansion, known as the Koopman Mode Decomposition (KMD), introduced the concept of Koopman modes, which are projection of observable evolution onto eigenfunctions of the Koopman operator, and give spatial ``shapes" that evolve exponentially in time - with exponents which could be complex numbers. Since then, KMD has been used as a tool of data-driven analysis with applications ranging from power network stability analysis to pattern detection in neural networks  \cite{susuki2014,giannakis2015spatiotemporal,georgescu2015building,erichson2015compressed,brunton2016extracting,mann2016dynamic}

In high-dimensional systems like fluid flows, direct analysis of the state space is computationally prohibitive and one often relies on the data-driven methods to extract the underlying spatio-temporal features from data obtained by numerical simulation or experiment.
For this purpose, the KMD was introduced to fluid mechanics by \citet{rowley2009spectral}, with an application to the problem of jet in cross flow. Rowley \textit{et al.} showed how KMD can extract various oscillation frequencies and their associated spatial structures from the data. They also discovered the connection between the KMD and the numerical algorithm called Dynamic Mode Decomposition (DMD), which had been proposed by P.J. Schmid \cite{schmid2008}.
Since then, KMD and its numerical counterpart DMD, have become popular tools for extraction of physically-relevant time scales and their associated spatial structures in complex flows \citep{schmid2010,schmid2011applications,pan2011dynamical,seena2011dynamic,muld2012flow,hua2016dynamic}.
Furthermore, the dynamical-systems origin of KMD has motivated a number of its applications in study of dynamical phenomena such as stability, bifurcation and  transition in flows. Some examples are study of bifurcation and transition in flow past a cylinder \citep{chen2012variants,bagheri2013koopman}, detection of significant structures in boundary layer transition \citep{sayadi2014reduced,subbareddy2014direct}, bifurcation analysis using parametric DMD \citep{sayadi2015parametrized} and identification of flow regimes in thermo-fluid systems using sparse sensing \citep{kramer2015sparse}.

In this work, we discuss KMD as a versatile tool to identify and analyze the state space dynamics of post-transient flows (i.e. state space trajectory evolving on the attractor).
The traditional approach to determine the dynamic regime of fluid systems is to look at the Fourier or power spectrum  of time series . The existence of sharp peaks in those spectra is deemed to indicate periodic motion while broadband spectrum is often interpreted as a sign of chaos \cite[e.g.][]{swinney1978transition,farmer1980power,karniadakis1992three,tomboulides2000numerical,peng2003transition,basley2011experimental}.
We point out the relationship between this classical viewpoint and the Koopman spectral analysis, and describe how the Koopman spectrum of data can be used to determine the geometry of the attractor, using the results in \cite{mezic2017koopman}.
The Koopman viewpoint generalizes the classical spectral analysis through the notion of Koopman eigenfunctions and modes. In particular, the Koopman eigenfunctions provide linearly evolving coordinates in the state space, and enable the use of state-space analysis and control techniques which are not reachable by classical spectral analysis. For example, the knowledge of Koopman eigenfunctions can be used to
construct nonparametric predictors \cite{giannakis2017data,korda2016linear}, state estimators \cite{surana2016linear,surana2016koopman} and  nonlinear controllers \cite{korda2016linear,kaiser2017data} using linear system strategies. The utility of Koopman eigenfunctions in flow prediction and control is further discussed in \cite{taira2017modal}. In this paper, we present a new way to construct (and visualize) the Koopman eigenfunctions in post-transient flows which are ergodic in the state space.

We apply the Koopman analysis to the two dimensional lid-driven cavity flow with regularized lid velocity. This flow provides a good benchmark for our analysis, since it shows a wide range of dynamic behavior over various Reynolds numbers. The dynamics of each flow regime is discussed in terms of the Koopman spectral properties: the Koopman spectrum determines the type of the attractor, the Koopman eigenfunctions indicate the oscillatory directions of motion in the state space, and the Koopman modes describe the evolution of velocity field in the flow domain. In particular, we use the Koopman modes to study the traveling waves that appear as a result of bifurcation from steady solution to periodic and quasi-periodic flow.

The lid-driven cavity flow becomes fully chaotic at ultimately high Reynolds numbers. In such flows, the Koopman spectrum is continuous and does not contain any (non-trivial) eigenvalues.
Using the properties of the Koopman operator and plausible assumptions on the post-transient dynamics, one can show how the measurements of observable on the chaotic flow can be interpreted as a realization of a \emph{wide-sense stationary stochastic process}.
This observation allows us to use the techniques from random signal processing to compute the continuous spectrum of the Koopman operator. We also study the flow regimes with mixed spectra, i.e., flows that have both discrete and continuous spectrum. In those flows, the evolution of flow observables is a mixture of quasi-periodic and chaotic motion, and the Koopman eigenfunctions help us distinguish and extract the quasi-periodic components of motion in the state space.

As of now, DMD-type algorithms are the most popular methods for computation of Koopman modes and eigenvalues. A notable extension of the original algorithms in \cite{schmid2010,rowley2009spectral}, is the so-called Extended DMD \cite{williams2015data} which approximates the Koopman operator as a matrix using different trajectories and a dictionary of observables. Other early works have discussed the connection of DMD with other data analysis tools such as Discrete Fourier Transform (DFT) \cite{chen2012variants} and linear system identification methods \cite{tu2014dynamic}. More recently, other variants of DMD have been proposed to study problems that exhibit a large range of time scales \cite{kutz2016multiresolution} and systems with external input \cite{proctor2016dynamic}.
Given the increasingly large and complex data sets that are generated by simulations and experiments, DMD has also been extended to handle larger data sets \citep{hemati2014dynamic,gueniat2015dynamic}, different sampling techniques \citep{brunton2013compressive,tu2014dynamic,gueniat2015dynamic} and noise \citep{dawson2016characterizing,hemati2017biasing}.

In this paper, we use a different approach for computation of Koopman spectral properties.
For post-transient flows, the spectrum of the Koopman operator (including both continuous and discrete components) lies on the imaginary axis, and the problem of estimating the Koopman spectrum reduces to the classical spectral estimation of signals. This problem is challenging for flows with mixed spectra where there is no a priori model for the continuous spectrum. Our methodology for Koopman spectral estimation consists of three steps: first, we apply a high-resolution algorithm - adapted from Laskar \cite{laskar1990chaotic,laskar1992measure}- to detect the candidate discrete Koopman frequencies and modes.
The Laskar algorithm provides a controllable balance between accuracy and computational efficiency which makes it suitable for large data sets like high-resolution flow snapshots. Moreover, it makes direct use of the harmonic averaging \cite{mezic2004comparison,mezic2005} which has proven convergence properties for computation of Koopman modes. In the second step, we use the ergodic properties of the attractor to discard the spurious frequencies that are artifacts of the continuous spectra. After extracting the periodic components of the flow, we estimate the continuous Koopman spectrum by applying the Welch method \cite{welch1967use} to the chaotic residual. Our computational approach is advantageous over DMD-type algorithms since it can handle the continuous spectrum, and it is related to the well-studied techniques and notions in spectral analysis of signals.

A key objective of modal decomposition techniques is to obtain low-dimensional representation of the data from experiments or numerical simulations. Therefore, an important question regarding any decomposition is how efficiently it can capture the flow evolution.
Several authors have already proposed variations of DMD algorithm to obtain low-dimensional description of the the flow features in an optimal manner
\cite{chen2012variants,wynn2013optimal,jovanovic2014sparsity}.
DMD is also used in a data assimilation approach to obtain a low-dimensional dynamic model of the cylinder wake flow \cite{tissot2014model}. Here, we study the efficiency of the Koopman modes by considering the error in the low-dimensional truncations of KMD in representing the cavity flow dynamics. We also compare the performance of Koopman modes with the modes obtained by Proper Orthogonal Decomposition (POD).

The outline of this paper is as follows. In \S\ref{sec:KoopmanTheory}, we review the basics of the Koopman operator theory and describe how the Koopman spectrum is related to the geometry of the attractor. We also point out the connection between the Koopman mode decomposition of different observables such as stream function, velocity field and vorticity.
In \S\ref{sec:CFD}, the flow settings for the lid-driven cavity and its numerical solution are described. Section \ref{sec:numericalKoopman} discusses the problem of estimating the Koopman spectrum and modes.
In \S\ref{sec:Koopmanresults}, we present the results of Koopman spectral analysis for the cavity flow. The Koopman modes and eigenfunctions are discussed in \S\ref{sec:Kefun} and \S\ref{sec:modes}, and the comparison with POD is presented in \S\ref{sec:ProjectedModels}. We summarize the results and conclude in \S\ref{sec:discussion}.

\section{Koopman Operator Theory \label{sec:KoopmanTheory}}
The Koopman operator theory is a mathematical formalism that relates the observations on a system to its underlying state-space dynamics. For a viscous incompressible flow, the state space is infinite-dimensional, i.e., it consists of all divergence-free smooth vector-fields defined on the flow domain. Studying the trajectories in the state space of such system via classical tools, like Poincar\'e maps, is nonviable, since they involve computation or visualization in an appropriate truncation of an infinite-dimensional space.
The Koopman operator viewpoint circumvents this problem by focusing on the time evolution of observables rather than state variables.
For example, the pressure magnitude at a certain point in the flow domain or the total kinetic energy of the fluid are are two observables on a flow that can be analyzed through the Koopman operator framework. Observables could be multiple-valued as well, like a vector containing the values of velocity at multiple grid points, or even a field of observables such as the vorticity field. One prominent outcome of the Koopman operator theory is the Koopman mode decomposition \citep{mezic2005}, which describes the evolution of such observables as a linear combination of Koopman modes, Koopman eigenfunctions and Koopman frequencies, which are all explained below.
We briefly review the basic formalism of the theory but the interested reader is referred to the review article \cite{mezic2013analysis} and the references therein for a more detailed exposition.

Consider the state space of a flow including all the smooth divergence-free velocity fields defined on the flow domain. The state of the flow, realized by the velocity field $\mathbf u$, evolves in time according to the Navier-Stokes equation written as
\begin{equation}
\partial_t \mathbf u(\mathbf x,t)=\mathbf F(\mathbf u(\mathbf x,t)).
\end{equation}
We let $g$ be a complex-valued function on the state space of the flow, i.e., for every state $\mathbf{u}$, the observable returns the complex value $g(\mathbf{u})$.
We call $g$ an \emph{observable} on the flow.
The Koopman operator describes how this observable changes with time.
More precisely, if we assume that the solution to above equation exists and it is unique, then the Koopman operator at time $\tau \in[0,\infty)$, denoted by $U^\tau$, maps the function $g$ to a new function $g^\tau$ such that
\begin{equation}\label{eq:KoopmanDef}
g^\tau(\mathbf{u}(\mathbf{x},t)):=U^\tau g(\mathbf{u}(\mathbf{x},t)) = g(\mathbf{u}(\mathbf{x},t+\tau)),
\end{equation}
The Koopman operator is a linear operator by definition ($U(\alpha_1 g + \alpha_2 h)= \alpha_1 U^\tau g +\alpha_2 U^\tau h$ for scalars $\alpha_{1,2}$) and therefore analyzing its spectrum and eigenfunctions gives a comprehensive understanding of its action on observables.
An eigenfunction of the Koopman operator is a function on the state space of the flow, similar to $g$, which evolves linearly with time.
Let us denote by $\phi_j$ the Koopman eigenfunction associated with the Koopman eigenvalue $\lambda_j$. Then
\begin{equation}
\phi^\tau(\mathbf{u}) := U^\tau \phi_j(\mathbf{u})= e^{\lambda_j \tau} \phi_j(\mathbf{u}) \label{eq:efun}
\end{equation}
In this work, we consider the post-transient flow dynamics, for which, the Koopman eigenvalues are known to lie on the imaginary axis \citep{mezic2005}. Therefore, we will be interested in Koopman frequencies, $\omega_j$, related to Koopman eigenvalues through the following,
\begin{equation}
\lambda_j=i\omega_j,\quad \omega_j\in \mathbb{R}.
\end{equation}
Let us temporarily assume that all the observables lie in the linear span of the Koopman eigenfunctions. Then any observable like $g$ can be expanded in the Koopman eigenfunctions,
\begin{equation}\label{eq:linexpansion}
g(\mathbf{u})= \sum_{j=1}^\infty g_j\phi_j(\mathbf{u}),
\end{equation}
where the scalar coefficient $g_j$ is given by the projection of observable $g$ onto the Koopman eigenfunction $\phi_j$. Since the Koopman operator is linear, we can use Eq. \eqref{eq:efun} and find a new expression for evolution of observable $g$ in terms of the Koopman eigenfunctions,
\begin{equation}
g^\tau(\mathbf{u})=U^\tau g(\mathbf{u})= \sum_{j=1}^\infty g_j\phi_j(\mathbf{u})e^{i\omega_j\tau}.
\end{equation}
If we replace the single-valued observable $g$ with a vector-valued observable such as $\mathbf{g}$, and follow the above procedure, the coefficient $g_j$ turns into the vector of coefficients $\mathbf{g}_j$, and we obtain a similar expansion in the vector form,
\begin{equation}
\mathbf{g}^\tau(\mathbf{u}):=U^\tau \mathbf{g}(\mathbf{u})= \sum_{j=1}^\infty \mathbf{g}_j\phi_j(\mathbf{u})e^{i\omega_j \tau}, \label{eq:KMD}
\end{equation}
This expansion of observables in terms of Koopman eigenfunctions is the so-called Koopman Mode Decomposition (KMD). The vector $\mathbf{g}_j$, called the Koopman mode associated with the Koopman frequency $\omega_j$, describes the components of the observable $\mathbf{g}$ obtained by projection of the observable onto the Koopman eigenfunction $\phi_j$.  As a result, the evolution of $\mathbf{g}$  in time could be described as a linear combination of Koopman modes with oscillating coefficients.
We will further explain the nature of this expansion in \S\ref{subsec:bifurcationKMD} and  Koopman modes in \S\ref{sec:Koopmanresults}.

The above expansion can be applied to \emph{fields of observables} as well, in which case, the Koopman modes become fields of coefficients. For example, projecting the \emph{velocity field observable} onto a Koopman eigenfunction returns a field of coefficients which can be thought of as a steady velocity field. Note that the velocity field undergoes nonlinear time evolution described by Navier-Stokes equations, but at the same time, the Koopman mode decomposition of the velocity field as an observable, offers a linear expansion in Koopman modes. This seeming paradox between the nonlinear evolution and linear expansion of KMD is resolved once we recall that the expansion in \eqref{eq:KMD} is essentially infinite dimensional and therefore it can describe the nonlinear time evolution.
In the following sections, we explain how the expansion above is related to the asymptotic dynamics of the trajectories in the state space and also remark on the choice of observable for study of the cavity flow.

\subsection{Flow bifurcations and Koopman mode decomposition}\label{subsec:bifurcationKMD}
It is interesting to see how the expansion in Eq. \eqref{eq:KMD} changes as the flow undergoes bifurcation.
We will be interested in detecting the post-transient flow dynamics which is directly related to the type of attractor on which the flow trajectory is evolving.
The bifurcations affect the Koopman eigenvalues, eigenfunctions and modes in the KMD, but here, we only discuss how the change in the distribution of Koopman eigenvalues can be traced back to qualitative changes in the attractor of the flow.

First, we recall some standard notions from dynamical systems theory. We call a compact invariant subset of the state space, denoted by $A$, an attractor of the dynamical system, if for many initial conditions the systems evolves toward $A$. Moreover $A$ is a minimal set in the sense that it cannot be split into smaller attractors (see e.g. \cite{wiggins2003introduction}). Simple examples of attractors in the state space of flows include stable fixed points and periodic orbits which correspond to steady and time-periodic flows, respectively. The attractor could be more complicated and exhibit chaos such as the butterfly-shaped attractor of the chaotic Lorenz system. We assume that the dynamics on the attractor preserves a physical measure (i.e. a distribution), which we denote by $\mu$. Roughly speaking, this implies that the time-averages (and therefore statistical properties) of continuous observables on the flow are well-defined. Now we let $\mathcal{H}:=L^2(A,\mu)$ be the Hilbert space of square-integrable observables defined on the attractor ($A$) with respect to measure $\mu$. In this work, we are interested in observables that belong to $\mathcal{H}$ (which includes continuous observables as well). It turns out that Koopman operator defined in \eqref{eq:KoopmanDef} is a unitary operator in $\mathcal{H}$ (i.e. its adjoint and inverse are the same), which implies that its spectrum lies on the unit circle. In the following, we use the symbol $\overset{\mu}{=}$ to describe the functional equalities, i.e., the functions on different sides of $\overset{\mu}{=}$ are equal everywhere on $A$ except on a set with zero $\mu$-measure. We also use $<f,g>_\mathcal{H}$ to denote the inner products in $\mathcal{H}$, i.e.,
\begin{equation}\label{eq:innerproduct}
  <f,g>_\mathcal{H}=\int_A fg^*d\mu.
\end{equation}

Let us revisit the key assumption that led to derivation of  \eqref{eq:KMD}, that is, the Koopman eigenfunctions span the space of observables, in this case, $\mathcal{H}$. For simple attractors like limit cycles and torus, this assumption holds and the expansion in \eqref{eq:KMD} can be used to explain the behavior of all observables. We first describe this case in more detail and then turn to the more general form of Koopman expansion for more complex dynamics.

When the trajectory in the state space of the flow evolves on a limit cycle or a torus, the post-transient flow shows (quasi-)periodic time dependence.
Let $\boldsymbol{\Omega}=[\omega_1,\omega_2,\ldots,\omega_m]^T$ denote the vector of basic frequencies for the for the motion of state variable on an $m-$dimensional torus (for limit cycles $m=1$). The Koopman spectral expansion for the flow is given by (\cite{mezic2017koopman})
\begin{equation}
U^\tau \mathbf{g}\overset{\mu}{=} \sum_{\mathbf{k}\in\mathbb{Z}^m} \mathbf{g}_{\mathbf{k}}\phi_\mathbf{k}e^{i\mathbf{k}\cdot\boldsymbol{\Omega} \tau}. \label{eq:KMDqperiodic}
\end{equation}
We have dropped the dependence of $\mathbf{g}$ and $\phi_k$ on the state $\mathbf{u}$ to simplify the notation. The above equation is a functional equality which holds almost everywhere on the attractor. We can evaluate it for a single trajectory starting  from the initial state $\mathbf{u}_0$ to obtain the vector expansion
\begin{equation}
U^\tau \mathbf{g}(\mathbf{u}_0)= \sum_{\mathbf{k}\in\mathbb{Z}^m} \mathbf{g}_{\mathbf{k}}\phi_\mathbf{k}(\mathbf{u}_0)e^{i\mathbf{k}\cdot\boldsymbol{\Omega} \tau}. \label{eq:KMDqperiodic_u0}
\end{equation}
The term $U^\tau \mathbf{g}(\mathbf{u}_0)$ is the signal generated by observing $\mathbf g$ over the trajectory starting at $\mathbf{u}_0$.
If the attractor is a limit cycle, this signal is time-periodic and \eqref{eq:KMDqperiodic_u0} is simply the Fourier series expansion in time. If the attractor is a torus, this expansion is a generalized Fourier expansion for the quasi-periodic signal that is generated by measuring $\mathbf{g}$. We observe that the Koopman frequencies in the above expansion form a \emph{lattice} on the frequency axis. For limit cycling systems, this lattice consists of multiples of the basic frequency $\omega_1$, while for the torus attractors, it is given linear combinations of the basic frequencies in $\boldsymbol{\Omega}$ with integer coefficients. Hence, a bifurcation from a limit cycle to a torus can be easily detected by counting the number of basic frequencies in the lattice of Koopman frequencies obtained from the data.

The point-evaluated expansion in  \eqref{eq:KMDqperiodic_u0} is more suitable for the study of fluid flows than the function expansion in  \eqref{eq:KMDqperiodic}. This is due to the fact that each flow simulation or experiment provides us with only a single trajectory in the state space and direct evaluation of the Koopman eigenfunctions on arbitrary regions of state space is not practical.
In case of post-transient flows, however, the ergodicity condition - which is discussed later - allows us to construct and visualize the Koopman eigenfunctions on the attractor using the signals coming from as few as one trajectory. We will use this fact to construct and visualize the eigenfunctions in \S\ref{sec:Koopmanresults}.

The converse of the above statements is also true, that is, if the Koopman spectrum of observables has only a countable number of frequencies, then the flow trajectory must be evolving on a torus-shaped attractor in the state space.  In fact, the  \emph{representation theorem} from the ergodic theory states that if the post-transient flow dynamics is ergodic and smooth, the Koopman operator having only discrete spectrum implies that the motion in the state space is topologically equivalent to rotation on a torus \cite{neumann1932operatorenmethode}. This classic result combined with numerical KMD algorithm gives a practical framework for detecting motion on tori in high-dimensional systems.

For post-transient flows with chaotic behavior, the Koopman eigenfunctions do not span $\mathcal{H}$ and evolution of observables cannot be described based on them. In fact, the Koopman operator spectra, in addition to eigenvalues, includes continuous spectrum which is related to the chaotic component of the flow. The spectral expansion for the Koopman operator takes a more general form (see e.g. \cite{maccluer2008elementary}), however,
as first stated in \citep{mezic2005}, we can still represent it in a way that distinguishes the quasi-periodic and chaotic components of the evolution. For the scalar observable $g$, it can be written as
\begin{equation}
U^\tau g\overset{\mu}{=} \sum_{\mathbf{k}\in\mathbb{Z}^m} g_{\mathbf{k}}\phi_\mathbf{k}e^{i\mathbf{k}\cdot\boldsymbol{\Omega} \tau} + \int_{-\infty}^{\infty}e^{ i \alpha \tau} dE_\alpha (g). \label{eq:KMDgeneral}
\end{equation}
The first term on the right-hand-side is the contribution of discrete spectrum and describes the quasi-periodic part of the flow (similar to \eqref{eq:KMDqperiodic}).
The second term is the contribution of the continuous spectrum. Informally speaking,  $i\alpha$ with $\alpha\in(-\infty,\infty)$, denotes a continuum  of eigenvalues  distributed along the imaginary axis.
The term $dE_\alpha(\cdot)$ is the spectral measure of the Koopman operator, that is, for each interval of frequencies such as $I=[\alpha_1,\alpha_2]$, $\int_{\alpha\in I}dE_{\alpha}(g)$ is the projection of the observable $g$ onto the eigen-subspace of $\mathcal{H}$ associated with $I$. The above expansion in the functional form is not suitable for flow applications, and it can be converted to a scalar equality by taking its inner product with the same observable $g$. That is
\begin{equation}
<g,U^\tau g>_\mathcal{H} = \sum_{k=1}^\infty |g_k|^2e^{i\omega_k \tau} + \int_{-\infty}^{\infty}e^{ i \alpha \tau} \rho_g(\alpha) d\alpha. \label{eq:KMDgeneral2}
\end{equation}
where we have assumed that eigenfunctions are normalized, i.e., $\|\phi_k\|_\mathcal{H}=1$. In passing from \eqref{eq:KMDgeneral} to \eqref{eq:KMDgeneral2}, we have made a technical assumption that the spectral measure of the Koopman operator for the chaotic part is absolutely continuous. The Koopman spectral density $\rho_g$ denotes the contribution of the continuous spectrum, such that the contribution of the frequency interval $I$ to the evolution of $g$ is given by
\begin{equation}
\mu_g(I)= \int_{I} \rho_g(\alpha) d\alpha.
\end{equation}

In order to compute the spectrum of the Koopman operator from the flow data, we need to assume that \emph{the  post-transient dynamics is ergodic}.
This implies that the statistics of the flow is independent of the initial condition, and the trajectories starting almost everywhere provide a perfect sampling of observables (in the sense defined in \eqref{eq:autocor} below). The ergodicity assumption holds for post-transient evolution of typical dynamical systems, including systems with periodic and quasi-periodic attractors and many chaotic systems like Lorenz \cite{luzzatto2005lorenz}.
Under this condition, we can use the pointwise ergodic theorem \cite{petersen1989ergodic} to approximate the inner product in  \eqref{eq:KMDgeneral2} from the data,
\begin{equation}
<g,U^\tau g>_\mathcal{H} = r_g(\tau):=\lim_{T\rar\infty} \frac{1}{T}\int_{0}^{T}g(t)g^*(t+\tau)dt. \label{eq:autocor}
\end{equation}
where $r_g(\tau)$ is the autocovariance function of $g$ at time $\tau$.
Therefore, we can approximate the spectral density of the Koopman operator by first extracting the chaotic component of $g$, then approximating $r_g$ using finite-time observations (i.e. finite $T$ in \eqref{eq:autocor}), and finally applying inverse Fourier transform to $r_g$.
We will discuss the practical aspects of this computation  in \S\ref{sec:numericalKoopman}.


\subsection{Stochastic processes and Koopman representation of deterministic chaos }\label{sec:stochasticprocess}
 In analyzing the chaotic data from experiments and simulations, it is customary to use the tools from applied probability theory  even in the case that underlying dynamical systems is fully deterministic. The reasoning behind this approach is the duality between the post-transient evolution of dynamical systems which is measure-preserving and the stationary stochastic processes. A classic formalism of this duality can be found e.g. in \cite{doob1953stochastic}. In this section, we reiterate this connection in the framework of the Koopman operator theory with an emphasis on the spectral expansion of observables.

 Recall that a continuous-time stochastic process is a collection of real random variables that are indexed by time, and denoted by
\begin{equation}\label{eq:stochastic_process}
  \{X_t\}_{t\in\mathbb{R}}.
\end{equation}
where $X_t$ is the random variable at time $t$ with a specified distribution over real line. A stochastic process is \emph{wide-sense stationary} if it satisfies two conditions. First, its expected value should not change with time, i.e.,
\begin{eqnarray}
  \mathbb{E}(X_t)=\mathbb{E}(X_{t+\tau})=m,\quad \text{for all } \tau \in \mathbb{R},
\end{eqnarray}
and second, its autocovariance function only depends on the lag time, i.e.,
 \begin{eqnarray}
  cov(X_t,X_{t+\tau})= \mathbb{E}((X_t-m)(X_{t+\tau}-m))=cov(\tau).
\end{eqnarray}
Now we consider the deterministic flow evolving on the attractor $A$ which preserves the normalized measure $\mu$ ($\mu(A)=1$). We see that the collection of observables
\begin{equation}\label{eq:gprocess}
  \{U^tg\}_{t\in\mathbb{R}},
\end{equation}
is a stochastic process defined on the probability space $(A,\mu)$. Each observable $U^tg$ is a random variable and it assigns a probability distribution on the real line which is given by
\begin{equation}\label{eq:distribution}
  \mathbf{P}(B)= \mu \big( (U^tg)^{-1}(B)\big),\quad B\subset\mathbb{R}.
\end{equation}
where $\mathbf{P}(B)$ is the probability of the interval $B$. Because of the measure-preserving property of the dynamics, this probability is independent of $t$, and the stochastic process in \eqref{eq:gprocess} is identically distributed (but not independent). Moreover, it is a wide-sense stationary process; in view of \eqref{eq:innerproduct}, we can write
\begin{equation}
  \mathbb{E}(U^tg)= <U^tg,1>_\mathcal{H}=<g,U^{-t} 1>_\mathcal{H}=<g,1>_\mathcal{H}=\mathbb{E}(g),
\end{equation}
and
\begin{eqnarray*}
  cov(U^tg,U^{t+\tau}g) &=& <U^tg,U^{t+\tau}g>_\mathcal{H},\\
  &=&<U^{-t}U^tg,U^\tau{g}>_\mathcal{H},\\
  &=&<g,U^\tau g>_\mathcal{H}=cov(\tau).
\end{eqnarray*}
where we have used the unitary property of the Koopman operator, i.e., $(U^t)^*=U^{-t}$. Using the measure-preserving property, one can show \eqref{eq:gprocess} is strictly stationary as well \cite{doob1953stochastic}, but that is not required for the spectral expansion.

According to the Wiener-Khintchine theorem (e.g. \cite{peebles1980probability}), the covariance of any wide-sense stationary process, such as \eqref{eq:gprocess}, has a spectral expansion in the following form,
\begin{eqnarray*}
  cov(g,U^{\tau}g) &=&\int_{-\infty}^{\infty}e^{i\alpha\tau}dF(\alpha)
\end{eqnarray*}
where $F$ is the power spectral distribution of the process. Note that this expansion holds for the general post-transient dynamics including both chaotic and quasi-periodic behavior. In case that there are no quasi-periodic components in the flow, and $F$ is absolutely continuous similar to \eqref{eq:KMDgeneral2}, we can rewrite the above expansion as
 \begin{eqnarray*}
  cov(g,U^{\tau}g) &=&\int_{-\infty}^{\infty}e^{i\alpha\tau}\rho(\alpha)d\alpha
\end{eqnarray*}
where $\rho$ is called the Power Spectral Density (PSD) of the stochastic process. Despite the deterministic nature of  our system, we observe that we can treat the chaotic component of the data as a a realization of a stationary process, and consequently, the notion of the Koopman spectral density coincides with that of PSD for random signals.
This observation enables us to use the spectral estimation techniques of stochastic signals for computation of Koopman continuous spectrum.

\subsection{Choice of observables and the relationship between their Koopman modes}\label{sec:KoopmanRemarks}
In this section, we consider the choice of observables for KMD and the relation between their modal decomposition.
This question is important since applying KMD to an observable reveals only the Koopman eigenvalues that are present in the expansion of that observable. Furthermore, one can use the relationship between the Koopman modes of different observables to reduce the computational cost of the analysis. The propositions in Appendix \ref{app:KMD} assert that if two observables are related through a linear operator, then their modes are also related via the same linear operator. For example, consider the field of stream function $\psi$ and the velocity field $\mathbf{u}$ in an incompressible 2D flow. These two observables are related thorough the linear operator $\nabla^\perp:=[\partial/\partial y,-\partial/\partial x]^T$, that is, $\mathbf{u}=\nabla^\perp \psi$. Let $\psi_j$ and $\mathbf{u}_j$ denote the Koopman modes of these two observable fields associated with Koopman eigenvalue $\lambda_j$, then
\begin{equation}
\mathbf{u}_j= \nabla^\perp \psi_j,\quad j=1,2,3,\ldots.	\label{eq:psitov}
\end{equation}
A similar relationship could be established between the Koopman modes of the vorticity field, denoted by $\boldsymbol{\vor}_j$, and those of the velocity field,
\begin{equation}
\boldsymbol{\vor}_j= \mathbf\nabla \times \mathbf{u}_j,\quad j=1,2,3,\ldots.\label{eq:vtow}
\end{equation}
This further implies that applying KMD to either of the these observable fields yields the same Koopman eigenvalues as long as none of the modes lie in the null space of the linear operator.

The knowledge of any of the above observable fields, i.e., stream function, velocity field or vorticity, uniquely determines the state of the system and therefore it can be used to elicit the Koopman spectrum of all other observables of interest.
Thus, we conclude that applying KMD to any of these fields would give us the information which is sufficient to detect the flow bifurcations.
In the dynamical analysis of the cavity flow, we choose the stream function as the primary observable for the application of KMD since its Koopman modes and eigenvalues are least expensive to compute. The Koopman modes of velocity and vorticity can be computed using (\ref{eq:psitov}) and (\ref{eq:vtow}).

\section{The lid-driven cavity flow}\label{sec:CFD}
The 2D lid-driven cavity flow is a simple model of an incompressible viscous fluid confined to a rectangular box with a moving lid.
This flow is usually used as a benchmark for numerical simulations, and represents a simplified model of geophysical flows driven by shear \cite{tseng2001mixing,gildor2010gulf},  and the flow inside a common type of mixer in polymer engineering \cite{chella1985fluid}. 
The 2D cavity flow is also realized in experiments using soap films  \cite{gharib1989liquid} (for experiments on 3D flow see \cite{koseff1984}).
This flow is particularly interesting for the Koopman analysis because it shows a wide range of dynamic behavior depending on the increase of the top lid velocity \cite{Shen1991,poliashenko1995direct,cazemier1998,auteri2002numerical,peng2003transition,balajewicz2013low}.


Our computational model of the flow consists of a square domain $[-1,1]^2$, with solid stationary boundaries, except the top lid (at $y=1$) which moves with a regularized velocity profile,
\begin{equation}
u_{lid} = (1-x^2)^2,\quad x\in[-1,1].
\end{equation}
This boundary condition has a low-order polynomial form which satisfies the continuity and incompressibility in the top corners (as opposed to the uniform velocity profile), and it is frequently used in numerical studies on cavity flow \cite[see e.g.][]{Shen1991,botella1997solution,balajewicz2013low}.

The incompressibility of the flow allows us to use the stream function formulation of the Navier-Stokes equation,
\begin{equation}
\frac{\partial }{\partial t} \nabla^2 \psi+ \frac{\partial \psi}{\partial y}\frac{\partial}{\partial x}\nabla^2 \psi - \frac{\partial  \psi}{\partial x}\frac{\partial}{\partial y}\nabla^2 \psi = \frac{1}{Re}\nabla^4 \psi,
\label{eq:si}
\end{equation}
subject to two types of boundary condition on the stream function,
\begin{equation}
\psi\bigg|_{\partial\Omega} =  0 \quad and \quad \frac{\partial \psi}{\partial n}\bigg|_{\partial\Omega}= u_w, \label{eq:BC}
\end{equation}
where the wall velocity $u_w$ is zero everywhere except at the top wall, where $u_w(y=1)=u_{lid}$. The solution of the cavity flow as described above is known to exist and be unique, and moreover, the flow trajectory asymptotically converges onto a universal attractor in  the state space \citep{temam1988infinite}.

For numerical solution, we have used the Chebyshev-spectral collocation method described in Ref. \onlinecite{trefethen2000spectral}. The stream function is approximated by a polynomial of order $N$ in spatial directions. This polynomial is determined by its values at the Chebyshev points,
\begin{eqnarray}
(x_i,y_j) = \left( \cos(\frac{i\pi} {N}),\cos(\frac{j\pi}{M})\right) \\ i=0,1,\ldots,N,\quad j=0,1,\ldots,M. \label{eq:Chebyshevgrid} \nonumber
\end{eqnarray}
Given the polynomial approximation and the prescribed boundary condition in (\ref{eq:BC}), we use the transformed variable $q(x,y)$ defined by
\begin{equation}
\psi(x,y)=(1-x^2)(1-y^2)q(x,y) \label{eq:si_q}.
\end{equation}
which satisfies the Dirichlet boundary condition identically, and turns the Neumann boundary condition  into Dirichlet boundary condition, i.e.,
\begin{eqnarray}
q(\pm 1,y) = q(x,-1) & = & 0, \\
q(x,+1)& = & -\frac{1}{2}u_{top}(x).
\end{eqnarray}
For the temporal discretization of the ordinary differential equations on $q(x_i,y_i)$, we have used the second-order Crank-Nicholson scheme for the diffusion terms and second-order Adams-Bashforth discretization for the convection terms. The flow solutions studied in this work are computed using zero initial velocity.
The numerical solutions of the steady flow obtained by our method agree with the results reported in Ref. \onlinecite{Shen1991}. There is also agreement on the time periods of the periodic flows between the two studies.  To the best of our knowledge, however, there are no reported benchmark solutions for quasi-periodic or aperiodic flow.

\section[KMD Computation]{Numerical computation of  Koopman spectrum and modes}\label{sec:numericalKoopman}

As discussed in section \ref{sec:KoopmanTheory}, the Koopman spectrum of post-transient flows lies on the imaginary axis, and its estimation reduces to the classical spectral analysis of flow signals. In this work, we are specially interested in flows that possess a continuous spectrum in addition to discrete frequencies.  Reliable estimation of each of these two components from data has a rich history in the context of signal processing and is still a subject of ongoing research. The DFT algorithm, by itself, gives a good approximation for the location of the discrete frequencies and there are a large number of the so-called \emph{high- or super-resolution} algorithms, based on DFT or otherwise, that improve the accuracy of such estimation.
For continuous spectra, however, DFT  is a poor estimator. Application of DFT to the autocovariance function in \eqref{eq:autocor} produces an estimate of spectral density with high fluctuations that \emph{do not} diminish with the increase of data samples \cite{stoica2005spectral}. 
Therefore, the algorithms developed to resolve continuous spectrum use some type of local averaging over frequency domain to reduce this variance. Conversely, this averaging process reduces the frequency resolution and makes these algorithms ill-suited  for detection of discrete spectra \cite{ghil2002advanced}. As a result a judicious combination of these methods should be used for computation of mixed spectrum (i.e. including both continuous and discrete parts).

Our strategy for computing the Koopman spectrum is to first detect and extract the discrete frequencies using a high-resolution algorithm, and then apply a continuous spectra estimator to the remainder. Note that most of the developed methods for accurate estimation of mixed spectrum are parametric, in the sense that they are based on specific models for the continuous spectrum such as colored or auto-regressive noise \cite{ghil2002advanced,li1996efficient,stoica1997cisoid}, which are not valid for typical chaotic dynamical systems.
Our methodology here is non-parametric,  and besides absolute continuity (discussed in \S\ref{subsec:bifurcationKMD}), we don't make any assumptions on the shape of the continuous spectrum. Instead, we connect our analysis to the theory of dynamical systems through the ergodicity assumption. Namely, given that the dynamics on the attractor is ergodic, we use the fact that the Koopman modes are unique (i.e. depend only on the observable and the flow parameters) which allows us to identify and discard the spurious discrete frequencies that are not robust with respect to the choice of initial condition or the time interval of integration.

The succession of ideas in this section are as follows: first, we describe the idea of harmonic averaging from classical ergodic theory which has proven convergence properties for computation of Koopman modes given the knowledge of Koopman frequencies. Then, we discuss the Laskar algorithm for computation of the discrete spectrum and benchmark its numerical performance against other high-resolution algorithms. In the last subsection, we discuss our procedure for approximation of Koopman continuous spectrum from the chaotic component of the data, and test its performance for two well-known chaotic dynamical systems.

\subsection{Harmonic averaging and DFT}
For post-transient flows, the Koopman eigenfunctions are orthogonal \cite{mezic2005} and the Koopman modes can be computed via direct projection of the observables onto the Koopman eigenfunctions. Let $\phi_j$ be the normalized Koopman eigenfunctions ($\|\phi_j\|=1$) associated with the frequency $\omega_j$. We observe that the Koopman eigenfunction evolves as $\phi^\tau(\mathbf{u}_0)= e^{i\omega \tau}$ over a single trajectory of the system.
Using the pointwsie ergodic theorem, we can compute the Koopman modes using the \emph{harmonic average},
\begin{equation}
\mathbf{g}_j:=<\mathbf{g},\phi_j>_H=\lim_{T\rightarrow \infty}\frac{1}{T}\int_{0}^{T}\mathbf{g}(\tau) e^{-i\omega_j \tau}d\tau. \label{eq:harmonicavg}
\end{equation}

The above limit is known to exist for almost every initial condition under the assumption that the dynamics on the attractor is preserving a measure \citep{wiener1941harmonic} - which is less restrictive than ergdocity.
The time series obtained by experiments and simulations consists of time-discrete samples over finite intervals. Assuming uniform sampling at time instants $\{\tau_0=0,\tau_1,\ldots,\tau_{N-1}=T\}$, we can approximate the harmonic average as
\begin{equation}
\mathbf{g}^N_j=\frac{1}{N}\sum_{k=0}^{N-1}\mathbf{g}(k) e^{-i\omega_j \tau_k}, \label{eq:harmonicavg2}
\end{equation}
where $\mathbf{g}(k)$ is the value of observable at the sampling time $\tau_k$. For any $\omega_j$ that is a Koopman frequency, we have $\mathbf{g}^N_j \rar \mathbf{g}_j$ as $N\rar\infty$, and otherwise $\mathbf{g}^N_j\rar 0$. For periodic and quasi-periodic attractors, the rate of convergence is proportional to $N$ \cite{mezic2002ergodic}, but for typical chaotic systems it scales with $\sqrt{N}$ \cite{krengel1985ergodic}.

Given a uniform sampling in time, we can use DFT frequencies as a rough approximation of the Koopman frequencies. Let the number of samples $N$ be even,  and denote the sampling interval by $\Delta\tau:=T/(N-1)$. The DFT grid of frequencies is
\begin{equation}
\Omega_j=\frac{2\pi j}{N\Delta \tau},\quad j=-\frac{N}{2},-\frac{N}{2}+1,\ldots,0,\ldots,\frac{N}{2}-1, \label{eq:Fourierfreq}
\end{equation}
Accordingly, computing the harmonic average in (\ref{eq:harmonicavg2}) reduces to computing the DFT amplitude of the observations,
\begin{equation}
\hat{\mathbf{g}}_j=\frac{1}{N}\sum_{k=0}^{N-1}\mathbf{g}(k) e^{-i\Omega_j k\Delta \tau}.
\end{equation}

DFT is already shown to be equivalent to DMD when applied to a linearly independent sequence of snapshots with zero mean \citep{chen2012variants}.
The advantage of using DFT to find the Koopman modes lies in its relative simplicity and the availability of Fast Fourier Transform (FFT) algorithms for its implementation. On the other hand, it suffers from two basic shortcomings.
The first one is the \emph{picket fencing}, i.e., the Koopman frequencies depend on the dynamics and may occupy arbitrary locations on the real interval, whereas DFT frequencies are determined by the sampling rate and observation interval. The second phenomenon, known as \emph{spectral leakage}, refers to the spillage of energy from a frequency to its neighborhood, and it is due to the finite length of observation interval which leads to errors in approximation of the modes \cite{oppenheim1999discrete,stoica2005spectral}. In the following, we discuss some of the methods developed to overcome these problems.

\subsection{Estimation of Koopman frequencies}
The problem of detecting discrete frequencies from noisy signals is often called \emph{line spectral estimation}.
The general goal of line spectral estimation methods is to compute estimates of frequencies with errors smaller than the DFT frequency resolution.
Many of such methods use DFT as a preliminary step because of its computational efficiency, and often utilize a combination of windowing and interpolation to reduce the leakage and fencing problem (see e.g. \cite{jain1979high,andria1989windows,agrez2002weighted}). Some other techniques, including Prony analysis \cite{de1795essai} and Nonlinear Last-Squares (NLS) method \cite{stoica2005spectral}, treat the line spectral estimation as a data fitting problem to find the frequencies and associated amplitudes that represent the time series with least error. These methods do not face the the DFT shortcomings, but they are more costly for computation and suffer vulnerability to noise (e.g. Prony analysis) or the choice of initial guess for the optimal values of fitting (e.g. NLS) \cite{stoica2005spectral}. We note that using Prony analysis to compute the Koopman modes is explored in \cite{susuki2015prony}.

There are also the so-called \emph{subspace} techniques which exploit the linear algebraic properties of matrices that embed the signal information. The two most popular algorithm in this class are the MUSIC \cite{schmidt1986multiple} and ESPRIT \cite{roy1989esprit} which use eigen-decomposition of the data covariance matrix. These methods circumvent the obstacles of the DFT-based methods by posing the frequency estimation as an eigenvalue problem, which leads to accurate estimates at a higher computational complexity due to the embedding of time-series in large matrices. Moreover, these methods are parametric and their good performance is only guaranteed when the noise follows a pre-determined model (which is usually white noise) \cite{stoica2005spectral}.
A more recent class of super-resolution algorithms recast the line spectral estimation as a convex optimization of measures on the frequency domain \cite{candes2014towards,fernandez2016super}. Under the two conditions of spectral sparsity and minimum separation between the frequencies, this framework recovers the exact values of frequencies from a finite number of time samples. Unfortunately, this framework has a high computational complexity and it is only suitable for discrete spectrum identification in presence of little noise.

In this work, we adapt the algorithm suggested by Laskar \cite{laskar1990chaotic,laskar1992measure} to compute the discrete Koopman frequencies and the associated modes. This algorithm is attractive for two reasons: first, it makes explicit use of harmonic averaging which allows us to assess its convergence based on the theory. In fact, this algorithm was invented to detect chaotic motion from data in Hamiltonian systems with a moderate number of degrees of freedom, like the solar system (see \cite{laskar1990chaotic}). Secondly, this algorithm is related to a popular sparse approximation technique known as Orthogonal Matching Pursuit (OMP)\cite{tropp2007signal}. OMP efficiently approximates a sparse vector (i.e. vector with few non-zero elements) given a relatively small number of linear measurements on the sparse vector through an iterative greedy algorithm. Different variants of this algorithm are frequently used in decomposition of signals and images into sinusoids, wavelets etc. (see e.g. \cite{mallat1993matching,wang2006seismic,fannjiang2012coherence,mamandipoor2016newtonized}).

The main idea in Laskar algorithm is to discretize the frequency domain and use OMP (implemented as FFT and harmonic averaging) to find the frequencies and associated amplitudes that best explain the time-sampled values of the observables (which are the linear measurements in the sense of OMP). The algorithm also uses windowing and adaptive refinement of the initial grid to diminish the effect of spectral leakage and picket fencing. The computational steps are outlined in algorithm \ref{alg:Laskar}, and below we describe the structure of data matrix used as the input.

Let $\{\mathbf{g}(0),\mathbf{g}(1),\ldots,\mathbf{g}({N-1})\}$ be the set of observations on the vector-valued observable $\mathbf{g}$, made on uniformly-spaced time instants $\{\tau_0=0,\tau_1,\ldots,\tau_{N-1}=T\}$. The snapshot data matrix $G$ is defined as
\begin{equation} \label{eq:snapshots}
G=[\mathbf{g}(0)|\mathbf{g}(1)|\ldots|\mathbf{g}(N-1)].
\end{equation}
Also let $\|\cdot\|$ denote an appropriate vector norm on $\mathbf g (\cdot)$, and $\sigma$ denote the expected $\|\cdot\|$-norm of measurement or computation noise in the data. We denote by $S(\omega)$ the sinusoid associated with frequency $\omega$, that is
\begin{equation}\label{eq:sin}
  S(\omega)=[1,~e^{i\omega \tau_1},~e^{i\omega \tau_2},\ldots, e^{i\omega \tau_{N-1}} ]^T.
\end{equation}
We also make use of windowing functions in the general form of a weight vector:
\begin{equation}\label{eq:window}
  W=[w_0,~w_1,~w_2,\ldots,w_{N-1}]^T.
\end{equation}

\begin{algorithm}[H]
\caption{ (Adapted) Laskar algorithm for estimation of Koopman frequencies} \label{alg:Laskar}
\begin{algorithmic}[1]
\REQUIRE{Snapshot matrix $G_{M\times N}$}.
\ENSURE{Set of Koopman frequencies $\Omega=\{\omega_1,\ldots,\omega_m\}$ and matrix of Koopman modes $V=[\mathbf{g}_1,\ldots,\mathbf{g}_m]$.  }
\STATE{Let $R=G$  and initialize the dictionaries $D=[~]$ and $\Omega=\{~\}$.}
\STATE{Apply row-wise FFT to $R$. Pick the DFT frequency $\hat{\omega}:=\omega_j$ which yields the complex amplitude $\mathbf{g}_j$ with highest $\|\cdot\|$-norm and satisfying $\|\mathbf{g}_j\|>\sigma$. If there is no such frequency proceed to step \ref{alg:Laskar7}.  }\label{alg:Laskar2}
\STATE{Compute the windowed harmonic average
\begin{equation}\label{eq:windowedharmonicavg}
  V_\omega = S^*(\omega) diag(W) R^T
\end{equation}
over a refined grid of frequencies centered around $\hat \omega$.
 Pick the frequency $\omega_k$ that yields $V_{\omega_k}$ with the highest $\|\cdot\|$-norm.}
\STATE{ Add $\omega_k$ and its sinusoid to the dictionary:
\begin{eqnarray}\label{eq:augmentDic}
  \Omega &\leftarrow& \Omega \cup \{\omega_k\}, \\
   D&\leftarrow& [D~S(\omega_k)].
\end{eqnarray}
}
\STATE{ Solve the least-square problem
\begin{equation}\label{eq:leastsqr}
   W = {\arg\min}_{\hat{W}}\|G^T-D\hat{W}\|_{fro}.
\end{equation}}
\STATE{Compute the new remainder $R$ by subtracting the contribution of the frequencies in the dictionary
\begin{equation}\label{eq:remainder}
   R = G^T-DW.
\end{equation}}
\STATE{Go to step \ref{alg:Laskar2}.}
\STATE{Return $\Omega$ and $V=W^T$. \label{alg:Laskar7}}
\end{algorithmic}
\end{algorithm}

\textbf{Choice of data matrix $G$ and appropriate norm:}
We have applied the above algorithm to the vector of the stream function ($\psi$) values at the computational grid points given in \eqref{eq:Chebyshevgrid}. The Koopman modes of velocity and vorticity are subsequently computed using \eqref{eq:psitov} and \eqref{eq:vtow}. The results reported in this paper are computed using the sampling rate $\omega_s=10~sec^{-1}$, where $sec$ is the unit of the the characteristic time given by
\begin{equation}
1~sec:=\frac{L_R}{U_R} \label{eq:stu_def}
\end{equation}
and $L_R$ and $U_R$ are the half of the cavity side length and maximum velocity on the top lid, respectively. Our numerical experiments show that the computed frequencies are independent of the sampling frequency $\omega_s\in[10,200]$.
We have chosen the vector-norm in the above algorithm such that it reflects the kinetic-energy norm of the Koopman modes, that is,
\begin{equation}
\|\mathbf{u}_j\|:=\|\mathbf{u}_j\|_{KE}=\left(\frac{\int_\Omega |\nabla^\perp \psi_j|^2 ds}{U_R^2}\right)^{1/2}. \label{eq:KEnorm}
\end{equation}
with $U_R$ denoting the maximum velocity on the top lid. We choose $\sigma^2=10^{-6}\|\mathbf{u}_0\|^2$. This is a heuristic choice and reflects how strong a periodic component we want to resolve.  Other factors that might be considered are the accuracy of the numerical simulation and the computational cost.

\textbf{Dictionary of frequencies for real-valued data and choice of the window function:}
Given that the spectrum is symmetric for real-valued data, we can effectively reduce the computational cost by doing the search and refinement (step 1 and 2) for a non-zero frequency $\omega_k$ and then add the pair $(-\omega_k,\omega_k)$ to the dictionary in step 3.
To evaluate the filtered harmonic average in step 2, we use the Hann window given by
\begin{equation}\label{eq:Hann}
  w(k)= \frac{1}{2} + \frac{1}{2}\cos(\pi k /N).
\end{equation}
Using the window function is not necessary but improves the detection of frequencies that are close to each other - in the case of quasi-periodic flow - as it reduces the local spectral leakage.
Alternative window functions can be chosen based on the proximity and relative strength of the frequencies (see e.g. \cite{stoica2005spectral}).

\textbf{Least-square projection and harmonic average:}
The least-square problem in step 4 is equivalent to orthogonal projection of observables onto the Koopman eigenfunctions. In fact, in the limit of $N\rar\infty$, the computation of Koopman modes in step 4 reduces to the harmonic average in \eqref{eq:harmonicavg}. To see this, note that the solution to \eqref{eq:leastsqr} is given by $W=D^\dagger G^T$, however as $N\rar \infty$, the columns of $D$ become orthogonal and $D^\dagger\rar (1/N) D^*$. It is easy to check that $(1/N)D^*G^T$ yields the harmonic average of columns of $G^T$, i.e., the Koopman modes.

 \begin{figure*}
	\centerline{\includegraphics[width=1\textwidth]{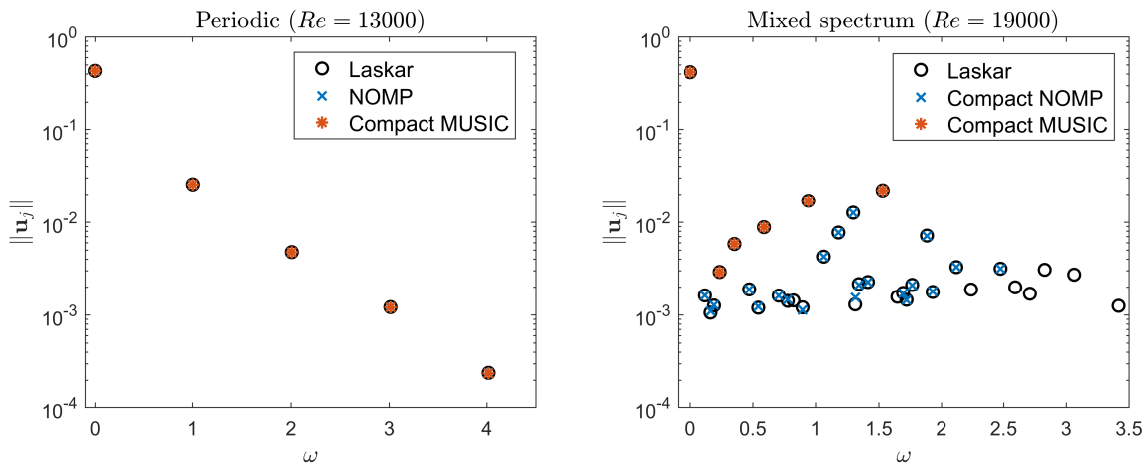}}
	\caption{Comparison of Laskar Algorithm with MUSIC \cite{schmidt1986multiple} and NOMP\cite{mamandipoor2016newtonized}. In the right panel, the continuous spectrum and frequencies with $\|\mathbf{u}_j\|<10^{-3}$ are omitted to avoid clutter. }
	\label{fig:LaskarBenchmark}
\end{figure*}

\textbf{Benchmark:} We compare the performance of Laskar algorithm to two other high-resolution algorithms. The first one is the Newtonized Orthogonal Matching Pursuit (NOMP) \cite{mamandipoor2016newtonized}. The benchmark study in \cite{mamandipoor2016newtonized} suggests that NOMP is a near-optimal algorithm in the sense that its accuracy is close to the theoretical limit. This algorithm is similar to Laskar, but one of its distinctive features is the refinement of \emph{all} frequency estimates after the detection of each new frequency. The computational run-time of NOMP is approximately $O(K^2)$ longer than Laskar, where $K$ is the number of detected frequencies.
We also implement the MUSIC algorithm using the \texttt{rootmusic()} function in MATLAB. This algorithm estimates the frequencies based on the eigen-decomposition of the data covariance matrix and has a high computational complexity which does not allow its implementation on large number of observables. In case of large data sets where these two algorithms cannot be applied to the all data, we implement them in the compact form, i.e., we choose the sampling of the stream function at 10 random points, and apply MUSIC/NOMP to compute the Koopman frequencies and then compute the Koopman modes using an orthogonal projection onto the dictionary of sinusoids.

Spectral estimation by MUSIC and NOMP do not suffer from the fencing problem since they don't use the discretization of the frequency domain, and therefore provide accurate estimates of Koopman frequencies for (purely) periodic and quasi-periodic flows. The accuracy of Laskar algorithm, on the other hand, depends on the choice of adaptive frequency grid, and higher resolutions requires evaluating \eqref{eq:windowedharmonicavg} over finer grids.
However, comparison (in FIG. \ref{fig:LaskarBenchmark}) shows that Laskar algorithm is better suited for computing the Koopman frequencies and modes in case of large flow data with mixed spectrum. The reason for this is two-fold: Many super-resolution algorithms are designed based on a special model for the noise spectrum (e.g., the MUSIC algorithm relies on the white noise model), which does not hold for the continuous spectrum of dynamical systems. (In the context of detecting discrete frequencies, we treat the continuous spectrum as noise.)
The second reason is related to the computational cost of the algorithm. The super-resolution algorithms that are not based on FFT are computationally expensive, and they can only be applied in the compact form, (i.e. applying to one or few observables simultaneously). In such cases, they may fail to capture many of low-energy frequencies in the presence of noise, as in the case of mixed spectra shown in FIG. \ref{fig:LaskarBenchmark}. (We have verified those frequencies by the criteria introduced in the next section and the fact that they lie on the lattice of frequencies described in \eqref{eq:KMDqperiodic}).
In contrast, the computational parsimony of the Laskar algorithm allows estimation of the frequencies using a larger set of (and possibly all) the observables. This, in turn, increases the effective signal-to-noise ratio for the low-energy periodic components which leads to detection of their associated frequencies.

\subsection{Detection of the spurious frequencies in flows with mixed spectra}\label{sec:VariabilityTest}
In the case of mixed spectra, application of high-resolution methods to the data might produce peaks that are not genuine Koopman frequencies, but artifacts of the continuous spectrum. In such flows, we need a criteria to distinguish such peaks from the actual frequencies.
On the other hand, the assumption of the ergodicity implies that the Koopman modes are unique (i.e. don't depend on the initial condition), and therefore Koopman modes computed over different (and sufficiently long) intervals should be the same. To use this notion we run Laskar'a algorithm on different (overlapping or non-overlapping) chunks of the  snapshot matrix.  We discard the frequencies whose associated modes show too much variability depending on the time interval of computation.
In the results to be presented, we have discarded the modes that show more than \%5 variability in the kinetic energy norm, while the modes are computed over intervals of 1000 $sec$ and longer.

\subsection{Estimation of Koopman continuous spectrum}
Recall from \S\ref{sec:stochasticprocess}, that the spectral density of the Koopman operator appearing in \eqref{eq:KMDgeneral2} coincides with the PSD of the chaotic component of signals generated by measuring observables. If $g$ is a real-valued observable with purely continuous spectrum, then we can approximate the autocovariance function in \eqref{eq:autocor} using the time series data
\begin{equation}\label{eq:autocarr_approx}
  r(\pm\tau)=\frac{1}{N}\sum_{k=0}^{N-1}g(k)g(k+\tau), \quad 0\leq\tau\leq N-1,
\end{equation}
and compute the \emph{correlogram} approximation of PSD as
\begin{equation}\label{eq:correlogram}
  \rho_g(\alpha)= \sum_{\tau=-(N-1)}^{N-1}r(\tau)e^{-i\alpha\tau}, \quad \alpha\in[0,2\pi).
\end{equation}
Such a direct evaluation produces a highly fluctuating estimate of $\rho_g$, and the computation must be modified to get a more reliable estimate \cite{stoica1989statistical}.

We use the Welch method \cite{welch1967use} to approximate the continuous part of the Koopman spectrum. The idea behind this algorithm is simple: it approximates the spectral density of the signal over small (and possibly overlapping) subsamples of the data using FFT, and averages the computed densities over all those subsamples. The averaging process reduces the variance of PSD estimation by a factor that is equal to the number of subsamples \cite{welch1967use}.  This reduction in the variance comes at the price of low spectral resolution (which increases with the length of subsamples), and using too many subsamples may result in over-averaging and getting a flat spectrum. Therefore, the number and length of windows should be chosen carefully to maintain an accurate estimate while resolving the distribution of energy over frequencies.

The Welch method reduces to the Bartlett method \cite{bartlett1950periodogram} in case of non-overlapping subsamples. Bartlett was among the first to realize that different subsamples of the data can be averaged to find a better estimate of the spectral density given that the autocovariance decays rapidly enough and the subsamples are sufficiently large. This methodology can be interpreted as special case in the well-known class of Balckman-Tukey estimators \cite{blackman1958measurement} and the general class of filter bank approaches. We refer the reader to Ref. \onlinecite{stoica2005spectral} for a discussion of connections between these methods.

We apply the Welch method, outlined in algorithm \ref{alg:Welch}, to the chaotic component of the velocity field. This component is computed by extracting the contribution of Koopman modes from the original data matrix. To measure the contribution of the continuous spectrum to the whole flow field, we compute the kinetic energy density of the continuous spectrum given by
\begin{equation}
p(\omega):=\frac{1}{U_R^2}\int_\Omega \rho_\mathbf{u}(\omega) ds. \label{eq:KEnormPSD}
\end{equation}
This definition would allows us to compute the kinetic energy content of each frequency interval via integrating $p(\omega)$ over that interval, i.e.,
\begin{equation}
P(I):=\frac{1}{2\pi}\int_I p(\omega)d\omega.
\end{equation}
and we will recover the average kinetic energy of chaotic fluctuations by calculating $P([0,2\pi))$.


 \begin{algorithm}[H]
\caption{ Welch method for estimation of Koopman continuous spectrum} \label{alg:Welch}
\begin{algorithmic}[1]
\REQUIRE{Snapshot matrix $G_{M\times N}$, length of subsamples $L$ and overlapping length $K$. }
\ENSURE{Matrix of spectral densities $R_{M\times L}$. }
\STATE{ Let $S=\text{FLOOR}[(M-L)/{K}]-1$ be the number of subsamples. }
 \FOR{the $i$-th row of $G$ denoted by $r$}
 \STATE {Divide $r$ into $S$ subsamples given by
 \begin{equation*}
   r_j(m)=r((j-1)K+m),~m=1,2,\ldots,L,\quad j=1,2,\ldots,S.
 \end{equation*}
}
 \FOR{each subsample $r_j$}
 \STATE {Use FFT to compute the PSD of $r_j$ \eqref{eq:correlogram} and denote it by $\phi_j(\omega_k)$ where $\omega_k$ with $k=,1,\ldots,L$ are the $L$-point FFT frequencies.}
 \ENDFOR
 \STATE{Let
 \begin{equation*}
   R_{ik} = \frac{1}{S}\sum_{j=1}^{S}\phi_j(\omega_k)
 \end{equation*}
 }
 \ENDFOR
 \STATE{Return $R$ and $\Omega=\{\omega_0,\omega_1,\ldots,\omega_L\}$.}
\end{algorithmic}
\end{algorithm}

We test the Welch algorithm using two well-known chaotic dynamical systems. The first one is a discrete-time map on a periodic 2D domain, known as Arnold's cat map, and given by
\begin{eqnarray*}
  x(t+1) &=& 2x(t)+y(t) \mod 1, \\
  y(t+1) &=& x(t)+y(t) \mod 1.
\end{eqnarray*}
For the choice of observables
\begin{eqnarray*}
  g_1(x,y) &=& e^{2\pi i (2x+y)} + \frac{1}{2}e^{2\pi i (5x+3y)}, \\
  g_2(x,y) &=& g_1(x,y) + \frac{1}{4}e^{2\pi i(13x+8y)}.
\end{eqnarray*}
the Koopman spectral density is known in analytical form \cite{Nithin2017convergent},
\begin{eqnarray*}
  \rho(g_1;\theta) &=& \frac{1}{2\pi}\bigg(\frac{5}{4}+\cos\theta\bigg), \\
  \rho(g_2;\theta) &=& \frac{1}{2\pi}\bigg(\frac{21}{16}+\frac{5}{4}\cos\theta + \frac{1}{2}\cos 2\theta\bigg).
\end{eqnarray*}
where $\theta\in[0,2\pi)$ is the discrete-time frequency.

The second system that we consider is the chaotic Lorenz system:
 \begin{eqnarray*}
  \dot x &=& 10(y-x), \\
  \dot y &=& x(28-z)-y,\\
  \dot z &=& xy-\frac{8}{3}z.
\end{eqnarray*}
This system is known to  have only continuous spectrum (except the zero frequency) \cite{luzzatto2005lorenz}, but no analytical expression exists for the spectral density of non-trivial observables. We compare the Welch estimation of the Lorenz spectrum with the recent results in \cite{korda2017data} which is based on the approximation of Fourier moments of the spectral measure and the Christoffel-Darboux kernel. In particular, we consider the observable
\begin{eqnarray*}
  g_3(x,y,z) &=& x.
\end{eqnarray*}
and compute its spectral density $\rho(g_3;\omega)$ where $\omega\in [0,\omega_s/2)$ is the continuous time frequency and sampling frequency $\omega_s$ is $10\pi$.
The comparison in FIG. \ref{fig:PSD} shows great agreement between the results of Welch method, analytic densities of the cat map and the numerical results of \cite{korda2017data} on Lorenz system.

 \begin{figure*}
	\centerline{\includegraphics[width=1\textwidth]{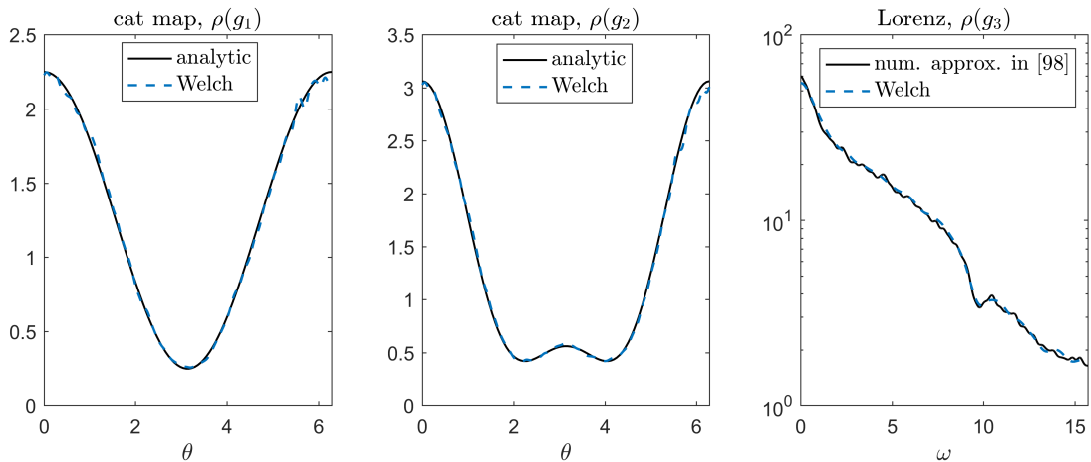}}
	\caption{Validation of Welch method for computing Koopman spectral density in Arnold's cat map (the left two panels) and chaotic  Lorenz system. The densities for cat map are known in analytic form \cite{Nithin2017convergent}, and the numerical benchmark results for Lorenz system are from \cite{korda2017data}. }
	\label{fig:PSD}
\end{figure*}

\section{Cavity flow dynamics and KMD}  \label{sec:Koopmanresults}

\subsection{ Koopman spectrum}
Figure \ref{fig:Spectrum} shows the distribution of kinetic energy in the discrete and continuous spectrum of the Koopman operator. The energy contained in Koopman modes (black bars in the figure) is simply the kinetic energy contained in each mode, but the representation of energy over the continuous spectrum is slightly different: the amount of energy contained at each frequency interval is the integral of the kinetic energy density of continuous spectrum (defined in \eqref{eq:KEnormPSD} and shown as the dashed blue curve) over that interval.

The evolution of the Koopman spectrum in FIG. \ref{fig:Spectrum} indicates the sequence of the bifurcations as follows: for $Re\leq10000$, the cavity flow induced by regularized lid velocity converges to a steady laminar solution which corresponds to a fixed point in the state space of the flow. The Koopman mode expansion for steady flow is trivial (hence not shown) and consists of zero frequency with an associated mode which is the steady flow. At a Reynolds number slightly above 10000, the steady solution becomes unstable and the numerical solution converges to a time-periodic flow which maintains stability up to $Re=15000$. The kinetic energy in this range is fully distributed in the Koopman modes. The basic frequency of periodic flow decreases with the Reynolds number, until at $Re\geq 15000$, another bifurcation occurs and the solution converges to a quasi-periodic flow.
The basic frequencies of the quasi-periodic flow also decrease with the $Re$, but around $Re=18000$ another bifurcation occurs and the level of kinetic energy lying in continuous spectrum quickly rises to a few percent. This kinetic energy of continuous spectra keeps rising such that at $Re\geq 22000$ we cannot detect any robust Koopman modes which indicates there are no quasi-periodic components in the state space dynamics.

The discrete Koopman frequencies obtained for periodic, quasi-periodic, and interestingly, the mixed-spectra flow match the lattice structure of the frequencies in the KMD of quasi-periodic flow (\ref{eq:KMDqperiodic}). That is, every frequency is accurately described by the integer combination of one or two basic frequencies  (see table \ref{tab:freq} in Appendix \ref{app:freq}). From the representation theorem mentioned in \S\ref{sec:KoopmanTheory}, we recall that this means the attractor is shaped like a limit cycle or a torus in the state space. For the flow with the mixed spectra, however, no such theorem exists but we can speculate that it consists of both a quasi-periodic factor and a chaotic factor. This type of attractor is called \emph{skew-periodic} in the literature of dynamical systems theory  \cite{broer1993mixed}.

\begin{figure*}
	\centerline{\includegraphics[width=1.2\textwidth]{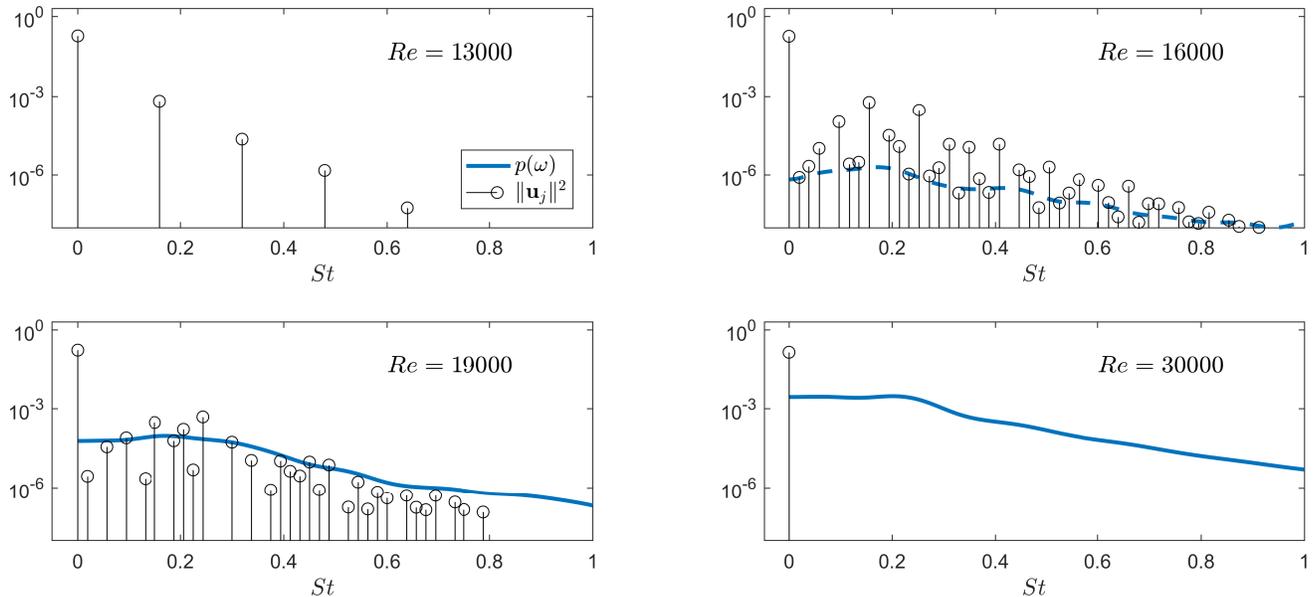}}
	\caption{Distribution of energy in the discrete (black bar) and continuous (blue curve) spectrum of the Koopman operator for cavity flow. The evolution of spectrum indicates the transition from periodic to chaotic flow. The ratio of average chaotic fluctuation energy to total kinetic energy of unsteady motion is $1.5\times10^{-3},~0.11$ and $1.00$ for $Re=16000,19000$ and $30000$, respectively.  
}
	\label{fig:Spectrum}
\end{figure*}

The evolution of the Koopman spectrum in FIG. \ref{fig:Spectrum} offers a picture of transition to chaos that is consistent with the theory of \citet{ruelle1971nature}.
According to this theory, the chaotic state of the flow can be reached after one or two Hopf bifurcations from an initially stable steady flow.
Physical evidence for this theory appeared in the experiments on rotating Couette flow and natural convection by \citet{swinney1978transition}. In particular, they detected the flow bifurcations using the power spectrum of velocity measurements at a single point in the flow domain. The transition to chaos was marked by the sudden growth of ``background noise" in the power spectrum of the quasi-periodic flow.
The above results show that the Koopman spectrum can be used as a generalized spectral tool for study of bifurcations; it offers a clear quantification of the energy in terms of true periodicity and contribution of continuous spectra for deterministic flows, and furthermore it connects the discrete spectrum to the state space geometry and flow domain.

\subsection{ Koopman eigenfunctions} \label{sec:Kefun}

In this section, we discuss the relationship between the Koopman spectrum and the state space dynamics. This relationship is realized through the notion of Koopman eigenfunctions, which are associated with the Koopman eigenvalues. For post-transient flows, the eigenfunctions provide an intrinsic coordinate on the state space along which the time evolution is linear oscillation. First, we construct the Koopman eigenfunctions for the quasi-periodic cavity flow using the theory presented in \cite{mezic2017koopman,mezic2004comparison}, and then we discuss its application to the flow with mixed spectra. Note that for flows with ultimately high Reynolds (e.g. $Re=30000$), the only Koopman frequency is zero and there are no oscillatory Koopman eigenfunctions and modes.

 Consider the quasi-periodic flow at $Re=16000$. This flow possesses a torus-shaped attractor and the state space trajectory evolves on this torus with two basic frequencies $\omega_1$ and $\omega_2$ (whose non-dimensional values are given in Appendix \ref{app:freq}).
We can parameterize the torus using two \emph{time-linear coordinates}, that is $(\theta,\gamma) \in [0,2\pi)^2$ with the linear evolution equation
\begin{eqnarray}  
  \dot{\theta} &=& \omega_1, \nonumber\\
  \dot{\gamma} &=& \omega_2. \nonumber
\end{eqnarray}
The evolution of the trajectories on the actual torus in the state space is nonlinear, but the tuple $(\theta,\gamma)$ are angular coordinates along a torus with uniform flow which is dynamically equivalent to the actual torus in the state space \cite{mezic2017koopman}. On this time-linear coordinates the Koopman eigenfunctions are the same as the Fourier functions, i.e., the Koopman eigenfunction $\phi_{k,l}$ associated with frequency $\omega_{k,l}=k\omega_1+l\omega_2$ is
\begin{equation}
\phi_{k,l}(\theta,\gamma)=e^{ik\theta+il\gamma}. \label{eq:EF}
\end{equation}
Although there is no analytical formula for transformation from the actual attractor to the time-linear coordinates defined above, we can construct the eigenfunctions in the state space using information on a trajectory, given that the trajectory is ergodic and sufficiently long. We can normalize an the eigenfunctions such that $\phi^{t=0}(\mathbf{u})=1$, and hence the value of the eigenfunction along the trajectory is given by setting $\theta=\omega_1t$ and $\gamma=\omega_2 t$ in \eqref{eq:EF}.

Figure \ref{fig:EFs} shows the construction of the eigenfunctions in the state space of the flow. The state space is realized by delay embedding of some typical observables \cite{takens1981detecting} - in this case the stream function at random points in the flow domain. The attractor of the periodic flow is a limit cycle (top row in the figure) and the Koopman eigenfunctions (shown as color field) correspond to the one dimensional linear time coordinate ($\theta\in[0,2\pi)$). For the quasi-periodic flow, the attractor is a 2-torus and the Koopman eigenfunctions show the directions on the torus where the evolution is linear and periodic, e.g., the eigenfunction $\phi_{0,1}$, shown in the rightmost panel of the second row shows the coordinate $\gamma$ along which the trajectories oscillate with frequency $\omega_2$.

For skew-periodic attractors (i.e. flow with mixed spectra) the eigenfunctions are even more interesting because they provide coordinates on an attractor which is not exactly a torus, but possesses directions with periodic motion.
The embedded attractor of the flow at $Re=19000$, for example, is similar to a torus which is related to the fact that this flow possesses a strong discrete spectrum (in the energy sense) with two basic frequencies and a relatively weak continuous spectrum. In fact, using the Koopman eigenfunctions, we can  compute the \emph{factors} (i.e. geometric slice) of such an attractor, on which, the motion is purely quasi-periodic. The existence of such factorization for systems with discrete Koopman eigenvalues was shown in Ref. \onlinecite{mezic2004comparison}. Here, we use this idea to reconstruct the quasi-periodic component of the attractor at $Re=19000$. Let $E$ be the observable whose embedding is used to construct the attractor. According to the \eqref{eq:KMDgeneral}, this observable can be split to two components:
\begin{equation}\label{eq:functionsplit}
  E=E_{qp}+E_{c},
\end{equation}
where $E_{qp}$ denotes the component of $E$ that lies in the span of Koopman eigenfunctions (including the eigenfunction at zero frequency), and $E_c$ is the chaotic component that belongs to subspace associated with the continuous spectrum. By doing KMD on $E$, we can extract its Koopman modes and reconstruct the evolution quasi-periodic component $E_{qp}$ over the trajectory which is given by the first two terms in \eqref{eq:KMDgeneral}. The embedding of $E_{qp}$ constructs the torus which corresponds to the quasi-periodic part of the motion.
As such, the general motion on the skew-periodic attractor (third row of the figure) can be decomposed into rotational motion along its quasi-periodic component (bottom row) superposed with chaotic motion in an unknown direction.
We stress that the above constructions are valid for any type of state space realization as long as the data on an ergodic trajectory is available.

\begin{figure*}
	\centerline{\includegraphics[width=1\textwidth]{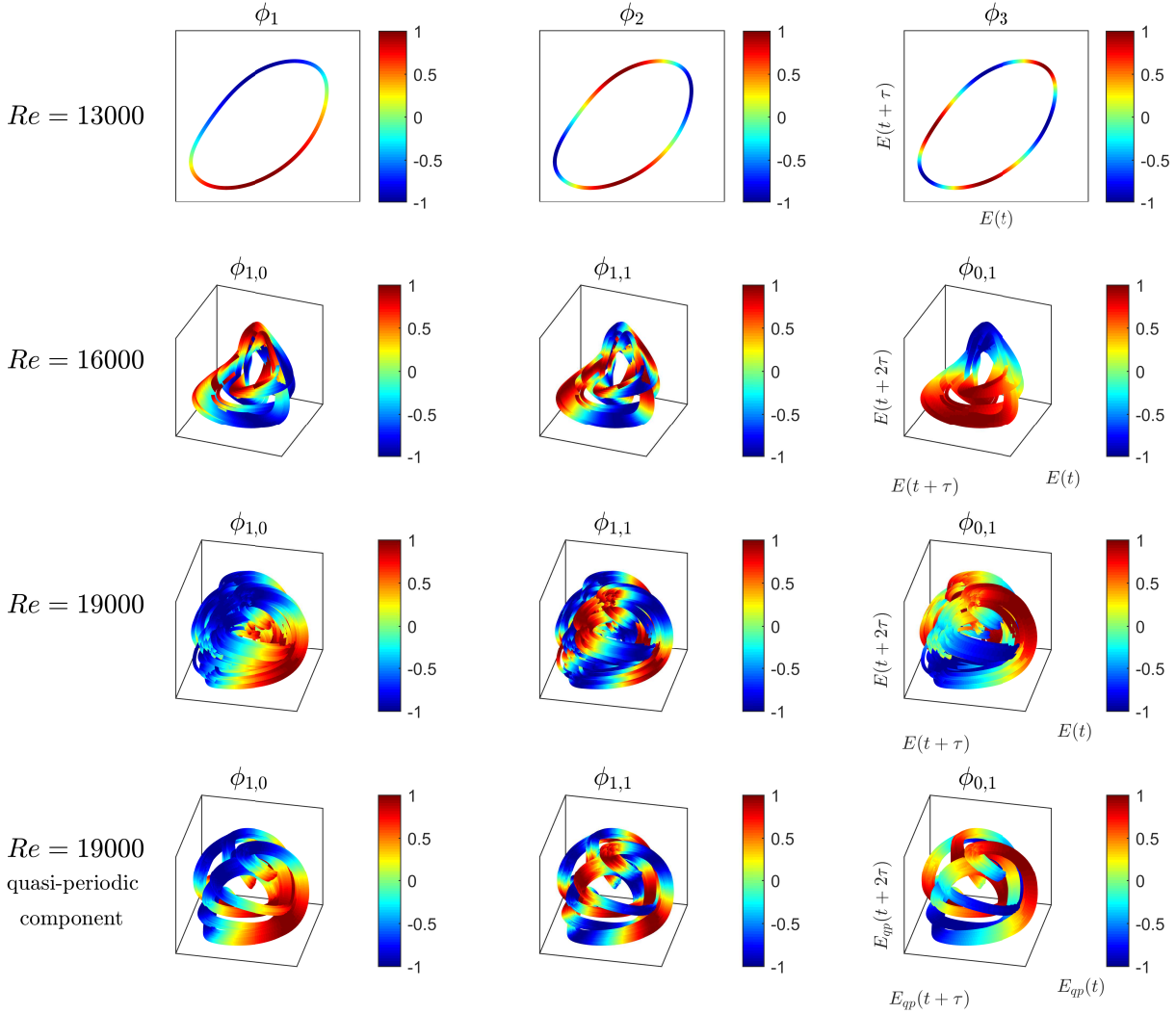}}
	\caption{ The (real part of) Koopman eigenfunctions shown as color field on the periodic, quasi-periodic and skew-periodic attractors (row 1-3). The last row is the quasi-periodic component of the skew-periodic attractor extracted using Koopman eigenfunctions.  The attractors are reconstructed using delay embedding of stream function values at random points in the flow domain (E) with the time delay of $1.0~sec$ for the periodic flow, and $2.4~sec$ for the rest.}
	\label{fig:EFs}
\end{figure*}

\subsection{ Koopman modes} \label{sec:modes}

The Koopman modes of the vorticity field associated with the basic frequencies of each flow are shown in FIG. \ref{fig:Modes}.
Each mode can be interpreted as the component of the vorticity field along the eigenfunction coordinates in the state space (the color field in FIG. \ref{fig:EFs}). For the eigenfunction at zero frequency, this component is the mean flow (Koopman mode associated with zero frequency) and does not change in time. The oscillatory modes however are components of the vorticity field that linearly oscillate along the eigenfunction directions.

\begin{figure*}
	\centerline{\includegraphics[width=1 \textwidth]{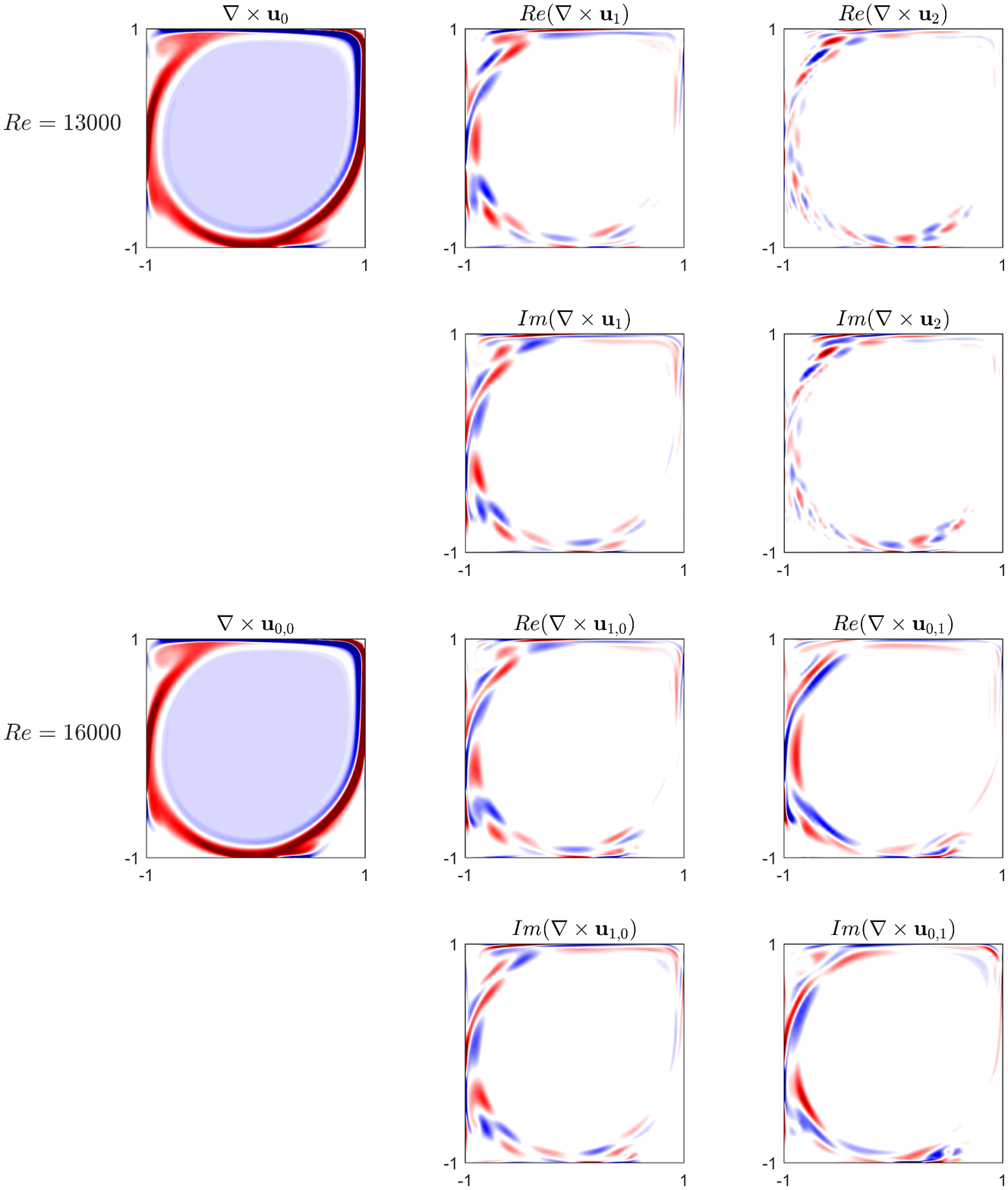}}
	\caption{The Koopman modes of vorticity in cavity flow (see the caption in next page).}
	\label{}
\end{figure*}

\begin{figure*}
\ContinuedFloat
	\centerline{\includegraphics[width=1 \textwidth]{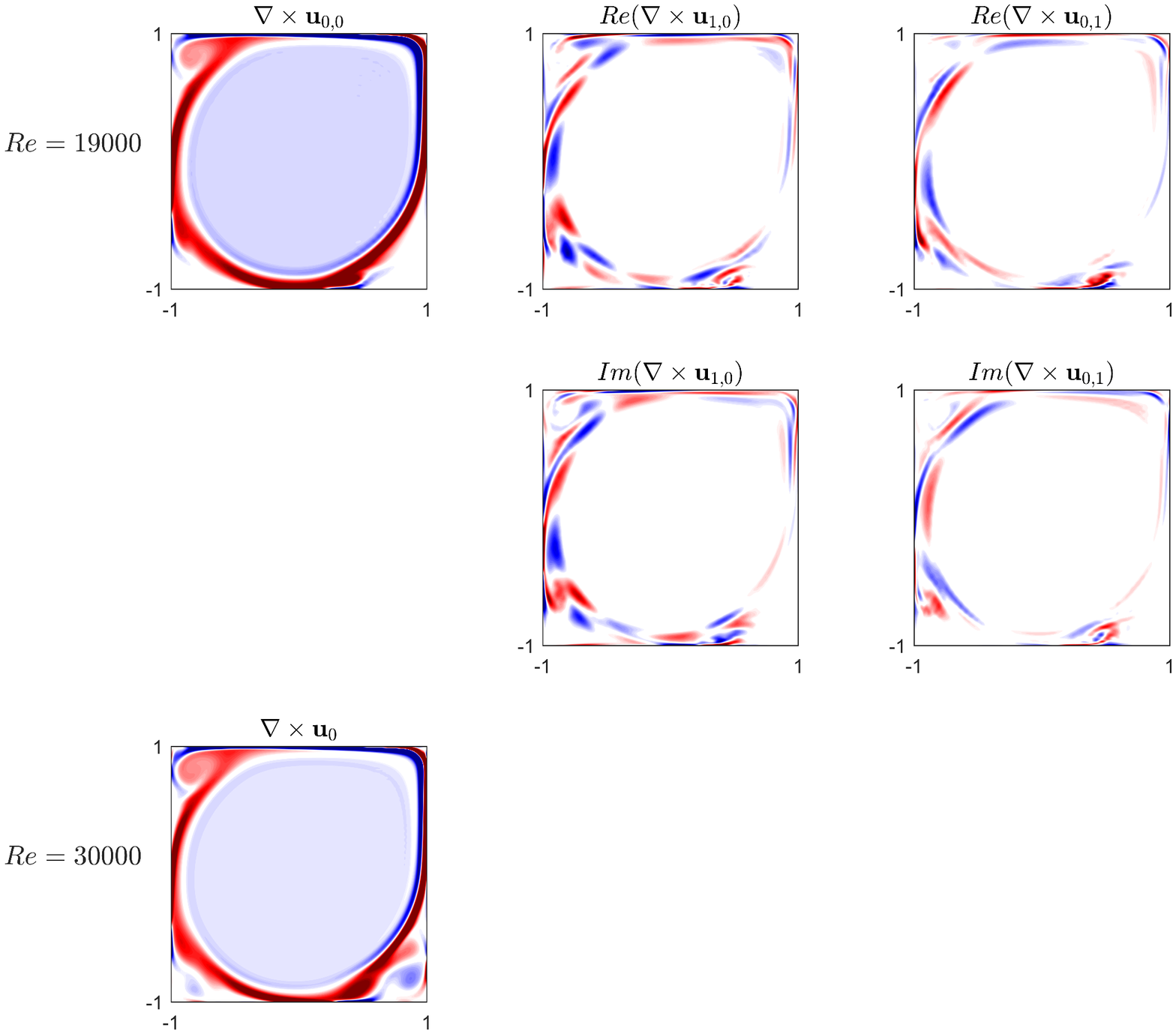}}
	\caption{The Koopman modes of the cavity flow associated with basic Koopman frequencies: The color field shows the real part of vorticity, with clockwise rotation shown in red, and counterclockwise in blue. The general structure of Koopman modes associated with same frequency trace remains unchanged as the Reynolds number is varied.}
	\label{fig:Modes}
\end{figure*}

The major share of kinetic energy in all the unsteady regimes is contained within the mean flow.
This mode is essentially composed of the central vortex in the flow and the corner eddies in the left corners.
For fully chaotic flows, the mean flow is the only Koopman mode, and its structure is similar to the mean flow of periodic and quasi-periodic flow except for the intensification of the downstream eddy in the bottom right corner.

The oscillatory Koopman modes, on the other hand, describe the flow oscillations around the edge of the central vortex in the mean flow.
To be more precise, the evolution of the mode $\mathbf{u}_k$ can be written as
\begin{eqnarray}\label{eq:modeevolution}
  \mathbf{u}(t) &=& \mathbf{u}_k e^{i\omega_kt}+\overline{\mathbf{u}}_ke^{i\omega_kt},\nonumber\\
                &=& 2Re(\mathbf{u}_k) \cos \omega_k t- 2Im(\mathbf{u}_k) \sin\omega_k t.
\end{eqnarray}
 A careful examination of the figure shows that $Re(\mathbf{u}_k)$ and $Im(\mathbf{u}_k)$  are similar for each $k$ but appear to be shifted in the direction along the shear layer of the mean flow. This observation is related to the fact that the unsteady motion in periodic and quasi-periodic regimes corresponds to wave(s) that travel along the downstream edge of the central vortex. This traveling wave structure is observed in the previous studies of cavity flow (see e.g. the sequence of flow snapshots in \cite{poliashenko1995direct,auteri2002numerical}) but never characterized.

The Koopman modes provide a straightforward framework to characterize the traveling waves from the data.
 Let's consider a simple example first: let $[0,1]$ be a periodic domain, over which, the general form of traveling wave is
\begin{equation}\label{eq:waveexample}
  f(\omega t-2\pi kx),~x\in[0,1]
\end{equation}
with $f$ being $2\pi$-periodic. Using the Fourier series expansion, we have
\begin{eqnarray}\label{eq:waveKMD}
  f(\omega t-2\pi kx)&=&\sum_{j=-\infty}^{\infty} e^{ij(\omega t- 2\pi k x)},\nonumber\\
                     &=&\sum_{j=-\infty}^{\infty} e^{ij\omega t}e^{- 2\pi jk  x}.
\end{eqnarray}
Clearly, the last expression is the KMD of $f$ with the Koopman modes given as $f_j(x)=\exp{-2\pi j k x}$. Having this example in mind, we can compute the wave numbers (and phase velocity) of traveling waves in the cavity flow through the following steps:
First, we compute the phase of each Koopman mode given by
\begin{equation}\label{eq:phase}
  \theta_k:=\angle\mathbf{u}_k=\tan^{-1}\bigg(\frac{Re(\mathbf{u}_k)}{Im(\mathbf{u}_k)}\bigg).
\end{equation}
Then, we sample the values of $\theta_k$ along the direction of travel, denoted by $\hat{x}$, and compute the average local slope of $\theta_k(\hat{x})$ to get the wave number of the mode. This process is summarized in FIG. \ref{fig:wave} (a) for the Koopman mode $\mathbf{u}_1$ of the periodic flow at $Re=13000$.
The results of this computation for different Koopman modes (shown in FIG. \ref{fig:wave}(b)) indicates that the Koopman modes associated with higher frequencies have proportionally higher wave numbers. This is expected from the KMD expansion of traveling wave \eqref{eq:waveKMD}.
Moreover, it suggests that the wave numbers rarely change with the Reynolds number, however the phase velocity (slope of the lines in FIG. \ref{fig:wave}(b)) slightly decreases due to the decrease in the Koopman frequencies.
We note that the above methodology based on Koopman modes is already used, in a concealed form, to extract the wave numbers from experimental data on nonlinear waves in thermally-driven flows  \cite{mukolobwiez1998supercritical,garnier2003nonlinear}, and internal waves in stratified flows  \cite{mercier2008reflection}.

\begin{figure*}
	\centerline{\includegraphics[width=1 \textwidth]{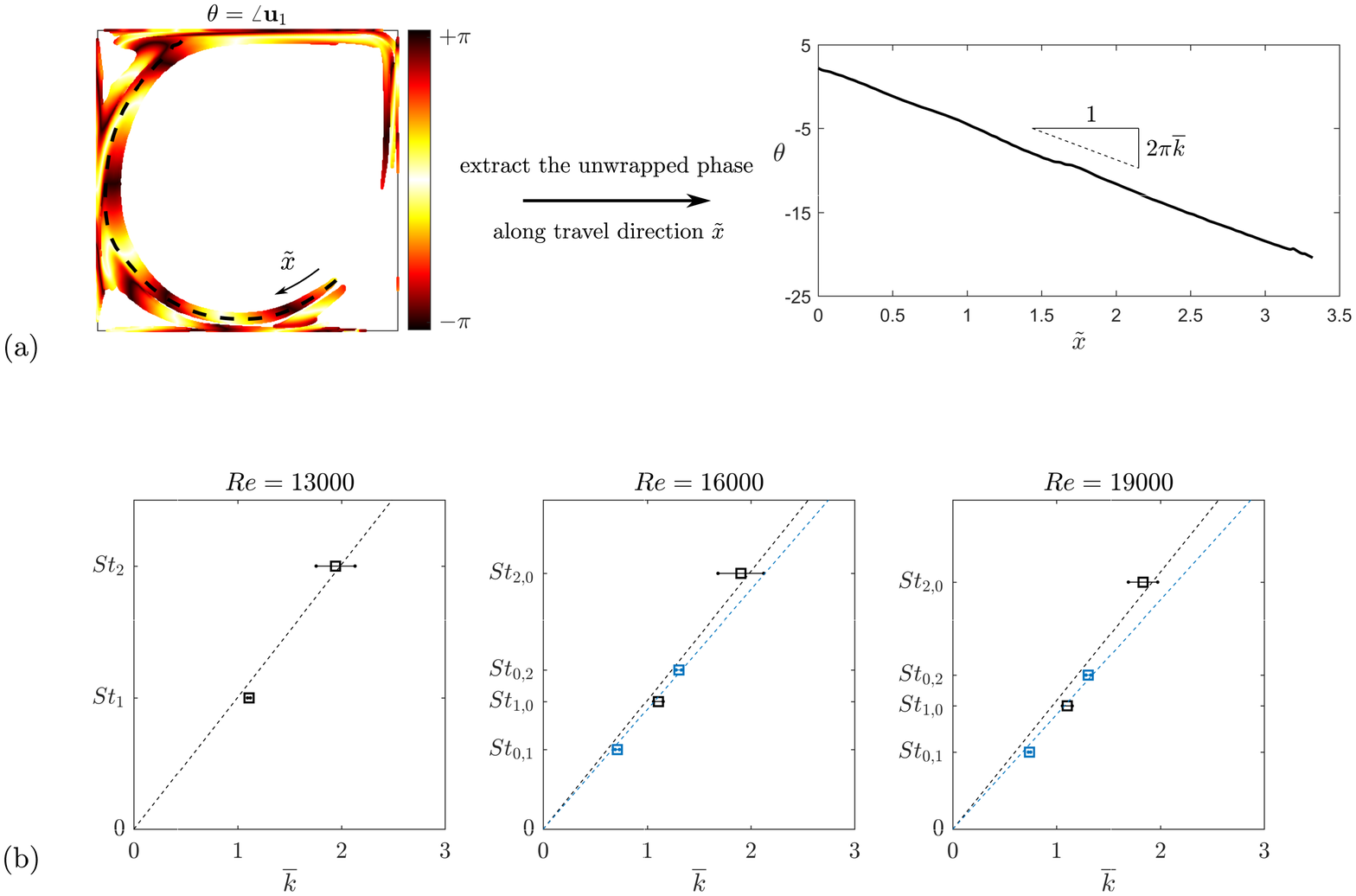}}
	\caption{(a) The process of extracting the (spatial) wave number from each Koopman mode and (b) relationship between the Koopman frequency and wave number of Koopman modes.}
	\label{fig:wave}
\end{figure*}

An interesting observation is that the oscillatory Koopman modes show a remarkable structural robustness in the range $Re=11000-19000$, despite the flow undergoing bifurcations from periodic to quasi-periodic and then skew-periodic. An examples of this robustness are the modes in panel 2, 7 and 12 in FIG. \ref{fig:Modes} (counting from top left to right) which correspond to the frequency trace of $St_{1,0}$ in periodic, quasi-periodic and mixed-spectra regimes, and panel 8 and 13 associated with the frequency $St_{0,2}$ in quasi-periodic and mixed-spectra regimes. This observation suggests that Koopman modes may provide a suitable basis for reduced modeling of flows (e.g. \cite{noack2003hierarchy}) over wide range of Reynolds number.

\subsection{Spectral Projections and Proper Orthogonal Decomposition (POD)} \label{sec:ProjectedModels}
We study the the efficiency of Koopman modes in representing the flow,  by computing the error of the spectral-projection models. An $n-$dimensional spectral projection model, is an $n$-term truncation of the KMD where the modes are sorted based on their kinetic energy.
The error defined as
\begin{equation}\label{eq:error}
\tilde e(n) = \frac{1}{T}\int_{0}^T \left\| \mathbf{u}(\mathbf{x},t)-\sum_{k=1}^{n}\mathbf{u}_k(\mathbf{x})e^{i\omega_k t}\right\|dt
\end{equation}
gives the kinetic-energy norm of the difference between the spectral projection model (the sum in the above expression) and the actual flow field. Given the finite-dimensional nature of these models, they are essentially  quasi-periodic approximations of the flow.
The time-averaged kinetic energy of the error for the spectral projections of order 1-10 is shown in FIG. \ref{fig:EulerianEerror}. In the periodic and quasi-periodic flows, the bulk of the motion is readily captured by a few  Koopman modes and the low-order projections approximate the flow with great accuracy. As the flow becomes less periodic with the increase of Reynolds number, the approximation error increases as well. For fully chaotic flows, the only Koopman mode is the mean flow and therefore there are no low-dimensional spectral projections except the steady one-dimensional model which is the mean flow itself. For purpose of comparison, however, we have plotted the error of spectral projections using Fourier modes (computed via FFT) at $Re$=30000. The kinetic energy of unsteady motion in this flow is spread in the continuous spectrum and any low-dimensional approximation using oscillatory components would involve large errors.

Figure \ref{fig:EulerianEerror} also shows an instructive comparison between Koopman mode decomposition and the Proper Orthogonal Decomposition (POD). POD is a decomposition of the flow field into spatially-orthogonal modes such that the POD-truncated models have the minimum energy error among all choices of orthogonal decompositions \citep{berkooz1993proper}. Due to its optimality and advantageous numerical properties, POD has been the keystone of many studies on coherent structures and low-order modeling of complex flows, inclduing the lid-driven cavity flow \cite{cazemier1998,balajewicz2013low}.
\begin{figure*}[ht]
	\centerline{\includegraphics[width=1 \textwidth]{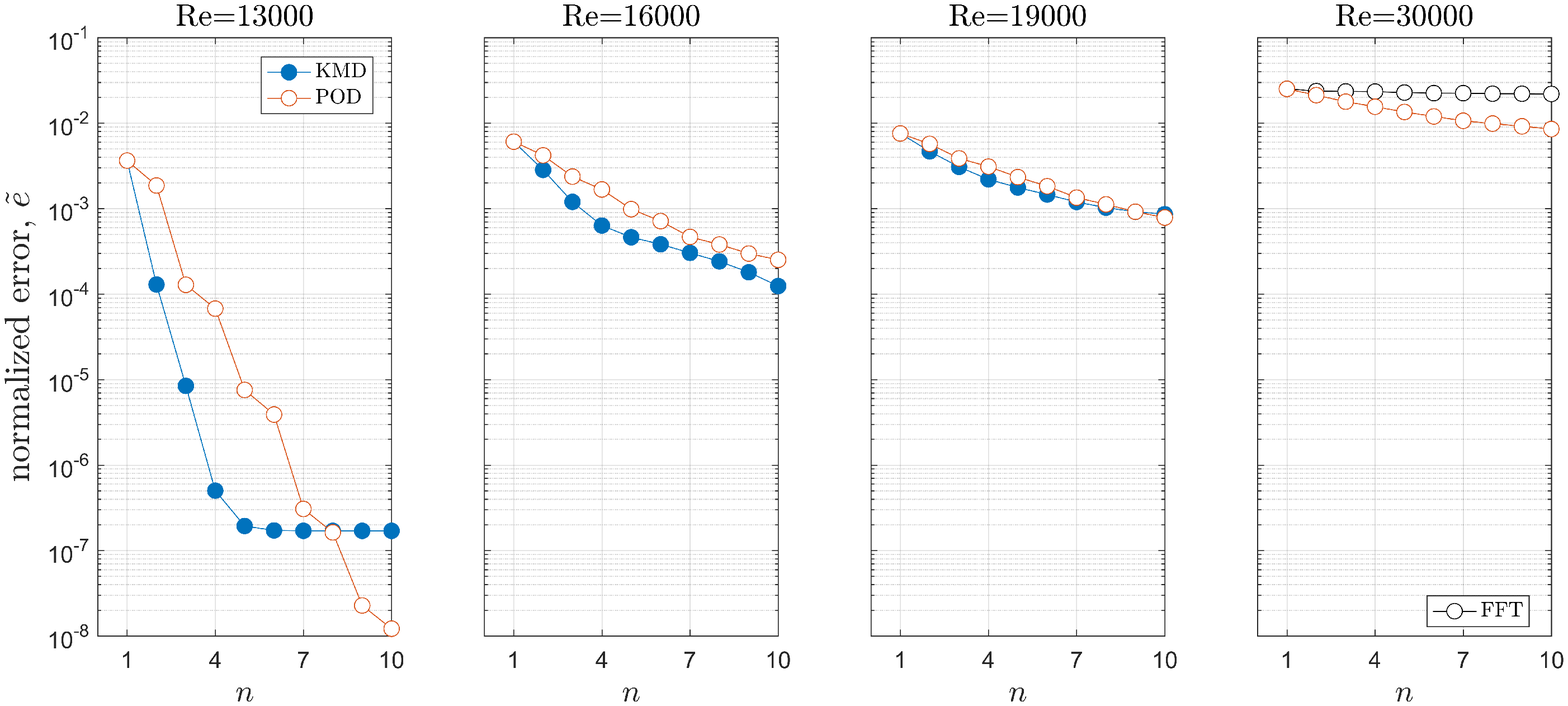}}
	\caption{The normalized kinetic energy of the error for approximation in Koopman mode decomposition (KMD) and Proper Orthogonal Decomposition (POD) as function of number of modes used in the approximation ($n$). For $Re=30000$ FFT modes are used in lieu of Koopman modes. }
	\label{fig:EulerianEerror}
\end{figure*}
In case of periodic and quasi-periodic flow, the low-dimensional spectral projection model gives a better approximation than POD-truncated models (first two panels in FIG. \ref{fig:EulerianEerror}). This observation is not contradictory to the optimality of POD-truncated models, but due to the fact that Koopman modes are complex-valued for oscillating systems which allows for better representation of dynamics evolving on limit cycles or torus as previously suggested in Ref. \onlinecite{mezic2005} (see example 4 therein). We speculate that this observations would be valid for other flows that may exhibit strong quasi-periodic behavior over some parameter range, for example, periodically-driven flows (e.g. \cite{lopez2000dynamics}), rotating Couette flow \cite{takeda1999quasi} and Rayleigh-Bernard convection \cite{swinney1978transition}, flow in converging-diverging channel \cite{guzman1994transition} and behind bluff objects \cite{karniadakis1992three}.

For flows with mixed spectra, the error of approximation with spectral projection and POD is generally comparable but as the Reynolds number increases and the flow becomes more chaotic, POD-truncated models perform better.
These observations support the alternative decomposition for model reduction proposed in \cite{mezic2005}, in which, the KMD is used to extract the periodic components of the flow, and then POD is used for representation of the remaining chaotic component.
\section{Conclusion} \label{sec:discussion}
In this work, we used the spectral properties of the Koopman operator for analysis of a post-transient flow.
Understanding of the asymptotic dynamics in the state space of the flow is achieved by inspecting the Koopman spectrum and eigenfunctions obtained by KMD of numerical data. In case of periodic or quasi-periodic attractors, the dominant Koopman frequencies possess a lattice-type structure whereas chaotic flows are associated with the continuous part of the Koopman spectrum. In between these two regimes, we observed flows that exhibit a combination of chaotic and quasi-periodic behavior (i.e. skew-periodic attractor) which is associated with mixed Koopman spectrum.
The Koopman eigenfunctions, on the other hand, determine the directions of linear evolution on quasi-periodic and skew-periodic attractors.

We employed a new approach for computation of Koopman modes and frequencies which comprises of a high-resolution algorithm to detect the eigenvalues, and an averaging estimation of the continuous spectrum. This approach, in contrast to DMD-type algorithms, is capable of computing the continuous spectrum of the Koopman operator and detecting genuine Koopman frequencies for flows with mixed spectra.

 Study of unsteady lid-driven cavity flow suggests that in case of periodic and quasi-periodic flow, a handful of Koopman modes are sufficient to represent the spatio-temporal features of the flow in a low-dimensional model. In fact, the Koopman modes offer a more efficient representation of such flows than POD modes. This is due to the fact that complex-valued Koopman modes are more suitable for describing the evolution on limit cycles and tori.

\section*{Acknowledgement and supplemental material}
This research was partially supported by the ONR grant N00014-14-1-0633. H.A. is grateful to Ati Sharma and Babak Mamandipoor for valuable discussions on spectral estimation and POD analysis. The authors also thank Peter Schmid for suggesting the computation of wave numbers using Koopman modes.

The time series data and MATLAB codes for KMD can be found at \url{https://mgroup.me.ucsb.edu/resources}.

\appendix
\section{Relationship between Koopman modes of different observables}\label{app:KMD}

We present two propositions regarding the Koopman mode decomposition of observables which are related through a linear operator or spatial differentiation.

\begin{prop}
		Let $g,h:\Omega \times M \rar \mathbb{R}$ be two fields of observables defined on the flow domain $\Omega$ and the state space of the flow $M$.
Assume $g(.,z),~h(.,z)\in\mathcal{F}(\Omega)$ for every state $z\in M$ where $\mathcal{F}(\Omega)$ denotes the space of fields on $\Omega$. Moreover, the observables are related through a linear operator $\mathcal{D}:\mathcal{F}(\Omega)\rar \mathcal{F}(\Omega)$,  such that $\mathcal Dg=h$ and $\mathcal{D}$ is bounded in $\mathcal{F}$.
Let $g_i$ and $h_i$ denote the Koopman modes of observable $g$ and $h$, respectively, associated with the Koopman eigenvalue $\lambda_j$. Then those modes are related to each other via  $\mathcal Dg_i=h_i$. In words, the Koopman mode decomposition commutes with the linear bounded operator $\mathcal{D}$.
	\label{thm:kmd3}	
\end{prop}

\begin{proof}
	Let us assume for now that the dynamical system is measure-preserving which implies $\lambda_j=i\omega_j,~\omega_j\in\mathbb{R}$. The Koopman modes are computed via the harmonic averaging and it follows that
	\begin{eqnarray}
	h_j&=&\lim_{T\rightarrow \infty}\frac{1}{T}\int_0^T h(t) e^{-i\omega_jt}dt \\
	&=&\lim_{T\rightarrow \infty}\frac{1}{T}\int_0^T \mathcal Dg(t) e^{-i\omega_jt}dt \nonumber \\
	&=&\lim_{T\rightarrow \infty}\frac{1}{T}\mathcal D\int_0^T g(t) e^{-i\omega_jt}dt \nonumber\\
	&=&\mathcal D\lim_{T\rightarrow \infty}\frac{1}{T}\int_0^T g(t) e^{-i\omega_jt}dt \nonumber\\
	&=& \mathcal D~ g_j \nonumber
	\end{eqnarray}
	We have used the linearity and continuity (i.e. boundedness) of $\mathcal D$ in the 3rd and 4th equalities, respectively. If the system is dissipative, the Koopman modes are given by Generalized Laplace Analysis formula \cite{mezic2013analysis}. The above argument could be used, along with induction, to prove the statement for that case.
\end{proof}
For the purposes of this paper, we are mostly interested in linear operators that involve spatial derivatives, such as gradient or curl. The derivative operator, however, is not bounded and therefore needs a special treatment which is given by the following proposition.

\begin{prop}	\label{thm:kmd4}
	Let $g,h:\Omega \times M \rar \mathbb{R}$ be two fields of observables defined on the flow domain $\Omega$ and the state space of the flow denoted by $M$. Assume the observables are bounded over their domain and related to each other through
	\begin{equation}
	\mathcal{D}g=h,
	\end{equation}
	where $D$ denotes the partial derivative with respect to a spatial coordinate in the flow domain $\Omega$.
	Let $g_k$ and $h_k$ denote the Koopman modes of observables associated with the Koopman frequencies $\omega_k,~k=1,2,\ldots$.
	Then
	\begin{equation}
	\mathcal{D}g_k=h_k.
	\end{equation}
	for every $x\in \Omega$. In words, the Koopman mode decomposition commutes with the spatial differentiation in the flow domain.
\end{prop}
\begin{proof}
Define the finite-time harmonic average of $g$ at the frequency $\omega_k$ by
\begin{equation}
g_k^T(x)=\frac{1}{T}\int_0^T g(x,t)e^{-i\omega_kt}dt. \label{eq:finInt}
\end{equation}
The Koopman mode of $g$ associated with $\omega_k$ is then given by \cite{mezic2005},
\begin{equation}
g_k=\lim_{T\rar\infty} g_k^T.
\end{equation}
The Leibniz rule implies that the spatial derivative commutes with the finite-time averaging in \eqref{eq:finInt}, i.e.,
\begin{equation}
h^T_k=\mathcal{D} g_k^T.
\end{equation}
for any finite $T$. Fix $x$.
If both $h(x)$ and $g(x)$ are integrable functions, then the existence of $g_k(x)$ implies the existence of $h_k(x)$ \citep{wiener1941harmonic}, and since
the scalar functions $h^T_k(x)$ and $\mathcal{D} g_k^T(x)$ are equal up to any finite $T$, their limits as $T\rar \infty$ must be equal, i.e., $h_k(x)=\mathcal{D} g_k(x)$.
\end{proof}

\section{Discrete spectra of Koopman operator}\label{app:freq}
Table \ref{tab:freq} shows the Koopman frequencies of cavity flow, computed by algorithm \ref{alg:Laskar} and ranked in order of decreasing kinetic energy.
The basic frequency vector is $St = 0.159826$ for the flow at $Re$=13000, $St=[0.155375,0.096911]$ for $Re$=16000 and  $St=[0.150046,0.093778]$ for $Re=19000$. The rest of the frequencies can be described (to the sixth digit of accuracy) as a linear combination of the basic frequencies. The coefficients of the combination are given by $\mathbf{k}$ in the table.

\begin{table}
 \begin{ruledtabular}
  \begin{tabular}{cc @{\hskip 0.35in} cc @{\hskip 0.35in} cc}
  \multicolumn{2}{l}{\qquad$Re$=13000 } & \multicolumn{2}{c}{$Re$=16000 } & \multicolumn{2}{ c}{$Re$=19000  }\\
    \multicolumn{2}{ l}{ \qquad(periodic)} & \multicolumn{2}{c}{ (quasi-periodic) } & \multicolumn{2}{ c}{ (mixed spectra) }\\[2pt]
    \hline
      $St$                  &   $k$     &$St$ & $\mathbf{k}$ &  $St$   &   $\mathbf{k}$ \\
    \multicolumn{1}{c}{0.000000} &   0    & 0.000000 & (0,0)      & 0.000000 & (0,0)\\
    \multicolumn{1}{c}{0.159826} &    1   & 0.155375 &  (1,0)     & 0.243822 & (1,1)\\
    \multicolumn{1}{c}{0.319653} &   2   & 0.252287 &   (1,1)    & 0.150046  & (1,0)\\
    \multicolumn{1}{c}{0.479479} &   3    & 0.096911 &   (0,1)    & 0.206310 &  (2,-1)\\
    \multicolumn{1}{c}{0.639303} &   4   & 0.193823 &    (0,2)   & 0.093778 &    (0,1)\\
    \multicolumn{1}{c}{0.799130} &   5   & 0.310751 &    (2,0)   & 0.300089 & (2,0)\\
    \multicolumn{1}{c}{0.958954} &   6   & 0.213839 &     (2,-1)  & 0.187556 & (0,2)\\
                                &       & 0.349197 &  (1,2)     & 0.056266 & (1,-1)\\
                                &       & 0.407661 &   (2,1)    & 0.450133 & (8,-8)\\
                                &       & 0.058463 &   (1,-1)    & 0.393868 & (7,-7)\\
                                &       & 0.116927 &   (2,-2)    & 0.337602 & (6,-6)\\
                                &       & 0.135359 &    (-1,3)   & 0.225068 &  (4,-4)\\
                                &       & 0.446109 &     (1,3)  & 0.487643 & (2,2)\\
                                &       & 0.290733 &   (0,3)    & 0.018756 & (2,-3)\\
                                &       & 0.038447 &    (-1,2)   & 0.431379 & (6,-5)\\
                                &       & 0.504573 &     (2,2)  & 0.543910 & (8,-7)\\
                                &       & 0.020017 &     (2,-3)  & 0.375111 & (0,4)\\
                                 &       & 0.232270 &    (-1,4)   & 0.468886 & (0,5) \\
                                &       & 0.466125 &     (3,0)  & 0.693955 & (9,-7)\\
                                &       & 0.272303 &     (3,-2)  & 0.581423 & (2,3)\\
    \end{tabular}
    \caption{Koopman frequencies with highest kinetic energy in the lid-driven cavity flow.}{\label{tab:freq}}
  \end{ruledtabular}
\end{table}


%


\begin{thebibliography}{110}%
\makeatletter
\providecommand \@ifxundefined [1]{%
 \@ifx{#1\undefined}
}%
\providecommand \@ifnum [1]{%
 \ifnum #1\expandafter \@firstoftwo
 \else \expandafter \@secondoftwo
 \fi
}%
\providecommand \@ifx [1]{%
 \ifx #1\expandafter \@firstoftwo
 \else \expandafter \@secondoftwo
 \fi
}%
\providecommand \natexlab [1]{#1}%
\providecommand \enquote  [1]{``#1''}%
\providecommand \bibnamefont  [1]{#1}%
\providecommand \bibfnamefont [1]{#1}%
\providecommand \citenamefont [1]{#1}%
\providecommand \href@noop [0]{\@secondoftwo}%
\providecommand \href [0]{\begingroup \@sanitize@url \@href}%
\providecommand \@href[1]{\@@startlink{#1}\@@href}%
\providecommand \@@href[1]{\endgroup#1\@@endlink}%
\providecommand \@sanitize@url [0]{\catcode `\\12\catcode `\$12\catcode
  `\&12\catcode `\#12\catcode `\^12\catcode `\_12\catcode `\%12\relax}%
\providecommand \@@startlink[1]{}%
\providecommand \@@endlink[0]{}%
\providecommand \url  [0]{\begingroup\@sanitize@url \@url }%
\providecommand \@url [1]{\endgroup\@href {#1}{\urlprefix }}%
\providecommand \urlprefix  [0]{URL }%
\providecommand \Eprint [0]{\href }%
\providecommand \doibase [0]{http://dx.doi.org/}%
\providecommand \selectlanguage [0]{\@gobble}%
\providecommand \bibinfo  [0]{\@secondoftwo}%
\providecommand \bibfield  [0]{\@secondoftwo}%
\providecommand \translation [1]{[#1]}%
\providecommand \BibitemOpen [0]{}%
\providecommand \bibitemStop [0]{}%
\providecommand \bibitemNoStop [0]{.\EOS\space}%
\providecommand \EOS [0]{\spacefactor3000\relax}%
\providecommand \BibitemShut  [1]{\csname bibitem#1\endcsname}%
\let\auto@bib@innerbib\@empty
\bibitem [{\citenamefont {Koopman}(1931)}]{koopman1931}%
  \BibitemOpen
  \bibfield  {author} {\bibinfo {author} {\bibfnamefont {Bernard~O}\
  \bibnamefont {Koopman}},\ }\bibfield  {title} {\enquote {\bibinfo {title}
  {Hamiltonian systems and transformation in hilbert space},}\ }\href@noop {}
  {\bibfield  {journal} {\bibinfo  {journal} {Proceedings of the National
  Academy of Sciences}\ }\textbf {\bibinfo {volume} {17}},\ \bibinfo {pages}
  {315--318} (\bibinfo {year} {1931})}\BibitemShut {NoStop}%
\bibitem [{\citenamefont {Mezi{\'c}}\ and\ \citenamefont
  {Banaszuk}(2004)}]{mezic2004comparison}%
  \BibitemOpen
  \bibfield  {author} {\bibinfo {author} {\bibfnamefont {Igor}\ \bibnamefont
  {Mezi{\'c}}}\ and\ \bibinfo {author} {\bibfnamefont {Andrzej}\ \bibnamefont
  {Banaszuk}},\ }\bibfield  {title} {\enquote {\bibinfo {title} {Comparison of
  systems with complex behavior},}\ }\href@noop {} {\bibfield  {journal}
  {\bibinfo  {journal} {Physica D: Nonlinear Phenomena}\ }\textbf {\bibinfo
  {volume} {197}},\ \bibinfo {pages} {101--133} (\bibinfo {year}
  {2004})}\BibitemShut {NoStop}%
\bibitem [{\citenamefont {Mezi{\'c}}(2005)}]{mezic2005}%
  \BibitemOpen
  \bibfield  {author} {\bibinfo {author} {\bibfnamefont {Igor}\ \bibnamefont
  {Mezi{\'c}}},\ }\bibfield  {title} {\enquote {\bibinfo {title} {Spectral
  properties of dynamical systems, model reduction and decompositions},}\
  }\href@noop {} {\bibfield  {journal} {\bibinfo  {journal} {Nonlinear
  Dynamics}\ }\textbf {\bibinfo {volume} {41}},\ \bibinfo {pages} {309--325}
  (\bibinfo {year} {2005})}\BibitemShut {NoStop}%
\bibitem [{\citenamefont {Susuki}\ and\ \citenamefont
  {Mezi\'c}(2014)}]{susuki2014}%
  \BibitemOpen
  \bibfield  {author} {\bibinfo {author} {\bibfnamefont {Yoshihiko}\
  \bibnamefont {Susuki}}\ and\ \bibinfo {author} {\bibfnamefont {Igor}\
  \bibnamefont {Mezi\'c}},\ }\bibfield  {title} {\enquote {\bibinfo {title}
  {Nonlinear koopman modes and power system stability assessment without
  models},}\ }\href@noop {} {\bibfield  {journal} {\bibinfo  {journal} {Power
  Systems, IEEE Transactions on}\ }\textbf {\bibinfo {volume} {29}},\ \bibinfo
  {pages} {899--907} (\bibinfo {year} {2014})}\BibitemShut {NoStop}%
\bibitem [{\citenamefont {Giannakis}\ \emph {et~al.}(2015)\citenamefont
  {Giannakis}, \citenamefont {Slawinska},\ and\ \citenamefont
  {Zhao}}]{giannakis2015spatiotemporal}%
  \BibitemOpen
  \bibfield  {author} {\bibinfo {author} {\bibfnamefont {Dimitrios}\
  \bibnamefont {Giannakis}}, \bibinfo {author} {\bibfnamefont {Joanna}\
  \bibnamefont {Slawinska}}, \ and\ \bibinfo {author} {\bibfnamefont {Zhizhen}\
  \bibnamefont {Zhao}},\ }\bibfield  {title} {\enquote {\bibinfo {title}
  {Spatiotemporal feature extraction with data-driven koopman operators},}\
  }in\ \href@noop {} {\emph {\bibinfo {booktitle} {Journal of Machine Learning
  Research, Proceedings of the 1st International Workshop on’Feature
  Extraction: Modern Questions and Challenges’ and NIPS Conference}}},\
  Vol.~\bibinfo {volume} {44}\ (\bibinfo {year} {2015})\ pp.\ \bibinfo {pages}
  {103--115}\BibitemShut {NoStop}%
\bibitem [{\citenamefont {Georgescu}\ and\ \citenamefont
  {Mezi{\'c}}(2015)}]{georgescu2015building}%
  \BibitemOpen
  \bibfield  {author} {\bibinfo {author} {\bibfnamefont {Michael}\ \bibnamefont
  {Georgescu}}\ and\ \bibinfo {author} {\bibfnamefont {Igor}\ \bibnamefont
  {Mezi{\'c}}},\ }\bibfield  {title} {\enquote {\bibinfo {title} {Building
  energy modeling: A systematic approach to zoning and model reduction using
  koopman mode analysis},}\ }\href@noop {} {\bibfield  {journal} {\bibinfo
  {journal} {Energy and buildings}\ }\textbf {\bibinfo {volume} {86}},\
  \bibinfo {pages} {794--802} (\bibinfo {year} {2015})}\BibitemShut {NoStop}%
\bibitem [{\citenamefont {Erichson}\ \emph {et~al.}(2015)\citenamefont
  {Erichson}, \citenamefont {Brunton},\ and\ \citenamefont
  {Kutz}}]{erichson2015compressed}%
  \BibitemOpen
  \bibfield  {author} {\bibinfo {author} {\bibfnamefont {N~Benjamin}\
  \bibnamefont {Erichson}}, \bibinfo {author} {\bibfnamefont {Steven~L}\
  \bibnamefont {Brunton}}, \ and\ \bibinfo {author} {\bibfnamefont {J~Nathan}\
  \bibnamefont {Kutz}},\ }\bibfield  {title} {\enquote {\bibinfo {title}
  {Compressed dynamic mode decomposition for real-time object detection},}\
  }\href@noop {} {\bibfield  {journal} {\bibinfo  {journal} {arXiv preprint
  arXiv:1512.04205}\ } (\bibinfo {year} {2015})}\BibitemShut {NoStop}%
\bibitem [{\citenamefont {Brunton}\ \emph {et~al.}(2016)\citenamefont
  {Brunton}, \citenamefont {Johnson}, \citenamefont {Ojemann},\ and\
  \citenamefont {Kutz}}]{brunton2016extracting}%
  \BibitemOpen
  \bibfield  {author} {\bibinfo {author} {\bibfnamefont {Bingni~W}\
  \bibnamefont {Brunton}}, \bibinfo {author} {\bibfnamefont {Lise~A}\
  \bibnamefont {Johnson}}, \bibinfo {author} {\bibfnamefont {Jeffrey~G}\
  \bibnamefont {Ojemann}}, \ and\ \bibinfo {author} {\bibfnamefont {J~Nathan}\
  \bibnamefont {Kutz}},\ }\bibfield  {title} {\enquote {\bibinfo {title}
  {Extracting spatial--temporal coherent patterns in large-scale neural
  recordings using dynamic mode decomposition},}\ }\href@noop {} {\bibfield
  {journal} {\bibinfo  {journal} {Journal of neuroscience methods}\ }\textbf
  {\bibinfo {volume} {258}},\ \bibinfo {pages} {1--15} (\bibinfo {year}
  {2016})}\BibitemShut {NoStop}%
\bibitem [{\citenamefont {Mann}\ and\ \citenamefont
  {Kutz}(2016)}]{mann2016dynamic}%
  \BibitemOpen
  \bibfield  {author} {\bibinfo {author} {\bibfnamefont {Jordan}\ \bibnamefont
  {Mann}}\ and\ \bibinfo {author} {\bibfnamefont {J~Nathan}\ \bibnamefont
  {Kutz}},\ }\bibfield  {title} {\enquote {\bibinfo {title} {Dynamic mode
  decomposition for financial trading strategies},}\ }\href@noop {} {\bibfield
  {journal} {\bibinfo  {journal} {Quantitative Finance}\ ,\ \bibinfo {pages}
  {1--13}} (\bibinfo {year} {2016})}\BibitemShut {NoStop}%
\bibitem [{\citenamefont {Rowley}\ \emph {et~al.}(2009)\citenamefont {Rowley},
  \citenamefont {Mezi{\'c}}, \citenamefont {Bagheri}, \citenamefont
  {Schlatter},\ and\ \citenamefont {Henningson}}]{rowley2009spectral}%
  \BibitemOpen
  \bibfield  {author} {\bibinfo {author} {\bibfnamefont {Clarence~W}\
  \bibnamefont {Rowley}}, \bibinfo {author} {\bibfnamefont {Igor}\ \bibnamefont
  {Mezi{\'c}}}, \bibinfo {author} {\bibfnamefont {Shervin}\ \bibnamefont
  {Bagheri}}, \bibinfo {author} {\bibfnamefont {Philipp}\ \bibnamefont
  {Schlatter}}, \ and\ \bibinfo {author} {\bibfnamefont {Dan~S}\ \bibnamefont
  {Henningson}},\ }\bibfield  {title} {\enquote {\bibinfo {title} {Spectral
  analysis of nonlinear flows},}\ }\href@noop {} {\bibfield  {journal}
  {\bibinfo  {journal} {Journal of Fluid Mechanics}\ }\textbf {\bibinfo
  {volume} {641}},\ \bibinfo {pages} {115--127} (\bibinfo {year}
  {2009})}\BibitemShut {NoStop}%
\bibitem [{\citenamefont {Schmid}\ and\ \citenamefont
  {Sesterhenn}(2008)}]{schmid2008}%
  \BibitemOpen
  \bibfield  {author} {\bibinfo {author} {\bibfnamefont {P.}~\bibnamefont
  {Schmid}}\ and\ \bibinfo {author} {\bibfnamefont {J.}~\bibnamefont
  {Sesterhenn}},\ }\bibfield  {title} {\enquote {\bibinfo {title} {Dynamic mode
  decomposition of numerical and experimental data},}\ }in\ \href@noop {}
  {\emph {\bibinfo {booktitle} {In Sixty-First Annual Meeting of the APS
  Division of Fluid Dynamics}}}\ (\bibinfo {year} {2008})\BibitemShut {NoStop}%
\bibitem [{\citenamefont {Schmid}(2010)}]{schmid2010}%
  \BibitemOpen
  \bibfield  {author} {\bibinfo {author} {\bibfnamefont {Peter~J}\ \bibnamefont
  {Schmid}},\ }\bibfield  {title} {\enquote {\bibinfo {title} {Dynamic mode
  decomposition of numerical and experimental data},}\ }\href@noop {}
  {\bibfield  {journal} {\bibinfo  {journal} {Journal of Fluid Mechanics}\
  }\textbf {\bibinfo {volume} {656}},\ \bibinfo {pages} {5--28} (\bibinfo
  {year} {2010})}\BibitemShut {NoStop}%
\bibitem [{\citenamefont {Schmid}\ \emph {et~al.}(2011)\citenamefont {Schmid},
  \citenamefont {Li}, \citenamefont {Juniper},\ and\ \citenamefont
  {Pust}}]{schmid2011applications}%
  \BibitemOpen
  \bibfield  {author} {\bibinfo {author} {\bibfnamefont {Peter~J}\ \bibnamefont
  {Schmid}}, \bibinfo {author} {\bibfnamefont {L}~\bibnamefont {Li}}, \bibinfo
  {author} {\bibfnamefont {MP}~\bibnamefont {Juniper}}, \ and\ \bibinfo
  {author} {\bibfnamefont {O}~\bibnamefont {Pust}},\ }\bibfield  {title}
  {\enquote {\bibinfo {title} {Applications of the dynamic mode
  decomposition},}\ }\href@noop {} {\bibfield  {journal} {\bibinfo  {journal}
  {Theoretical and Computational Fluid Dynamics}\ }\textbf {\bibinfo {volume}
  {25}},\ \bibinfo {pages} {249--259} (\bibinfo {year} {2011})}\BibitemShut
  {NoStop}%
\bibitem [{\citenamefont {Pan}\ \emph {et~al.}(2011)\citenamefont {Pan},
  \citenamefont {Yu},\ and\ \citenamefont {Wang}}]{pan2011dynamical}%
  \BibitemOpen
  \bibfield  {author} {\bibinfo {author} {\bibfnamefont {Chong}\ \bibnamefont
  {Pan}}, \bibinfo {author} {\bibfnamefont {Dongsheng}\ \bibnamefont {Yu}}, \
  and\ \bibinfo {author} {\bibfnamefont {Jinjun}\ \bibnamefont {Wang}},\
  }\bibfield  {title} {\enquote {\bibinfo {title} {Dynamical mode decomposition
  of gurney flap wake flow},}\ }\href@noop {} {\bibfield  {journal} {\bibinfo
  {journal} {Theoretical and Applied Mechanics Letters}\ }\textbf {\bibinfo
  {volume} {1}} (\bibinfo {year} {2011})}\BibitemShut {NoStop}%
\bibitem [{\citenamefont {Seena}\ and\ \citenamefont
  {Sung}(2011)}]{seena2011dynamic}%
  \BibitemOpen
  \bibfield  {author} {\bibinfo {author} {\bibfnamefont {Abu}\ \bibnamefont
  {Seena}}\ and\ \bibinfo {author} {\bibfnamefont {Hyung~Jin}\ \bibnamefont
  {Sung}},\ }\bibfield  {title} {\enquote {\bibinfo {title} {Dynamic mode
  decomposition of turbulent cavity flows for self-sustained oscillations},}\
  }\href@noop {} {\bibfield  {journal} {\bibinfo  {journal} {International
  Journal of Heat and Fluid Flow}\ }\textbf {\bibinfo {volume} {32}},\ \bibinfo
  {pages} {1098--1110} (\bibinfo {year} {2011})}\BibitemShut {NoStop}%
\bibitem [{\citenamefont {Muld}\ \emph {et~al.}(2012)\citenamefont {Muld},
  \citenamefont {Efraimsson},\ and\ \citenamefont {Henningson}}]{muld2012flow}%
  \BibitemOpen
  \bibfield  {author} {\bibinfo {author} {\bibfnamefont {Tomas~W}\ \bibnamefont
  {Muld}}, \bibinfo {author} {\bibfnamefont {Gunilla}\ \bibnamefont
  {Efraimsson}}, \ and\ \bibinfo {author} {\bibfnamefont {Dan~S}\ \bibnamefont
  {Henningson}},\ }\bibfield  {title} {\enquote {\bibinfo {title} {Flow
  structures around a high-speed train extracted using proper orthogonal
  decomposition and dynamic mode decomposition},}\ }\href@noop {} {\bibfield
  {journal} {\bibinfo  {journal} {Computers \& Fluids}\ }\textbf {\bibinfo
  {volume} {57}},\ \bibinfo {pages} {87--97} (\bibinfo {year}
  {2012})}\BibitemShut {NoStop}%
\bibitem [{\citenamefont {Hua}\ \emph {et~al.}(2016)\citenamefont {Hua},
  \citenamefont {Gunaratne}, \citenamefont {Talley}, \citenamefont {Gord},\
  and\ \citenamefont {Roy}}]{hua2016dynamic}%
  \BibitemOpen
  \bibfield  {author} {\bibinfo {author} {\bibfnamefont {Jia-Chen}\
  \bibnamefont {Hua}}, \bibinfo {author} {\bibfnamefont {Gemunu~H}\
  \bibnamefont {Gunaratne}}, \bibinfo {author} {\bibfnamefont {Douglas~G}\
  \bibnamefont {Talley}}, \bibinfo {author} {\bibfnamefont {James~R}\
  \bibnamefont {Gord}}, \ and\ \bibinfo {author} {\bibfnamefont {Sukesh}\
  \bibnamefont {Roy}},\ }\bibfield  {title} {\enquote {\bibinfo {title}
  {Dynamic-mode decomposition based analysis of shear coaxial jets with and
  without transverse acoustic driving},}\ }\href@noop {} {\bibfield  {journal}
  {\bibinfo  {journal} {Journal of Fluid Mechanics}\ }\textbf {\bibinfo
  {volume} {790}},\ \bibinfo {pages} {5--32} (\bibinfo {year}
  {2016})}\BibitemShut {NoStop}%
\bibitem [{\citenamefont {Chen}\ \emph {et~al.}(2012)\citenamefont {Chen},
  \citenamefont {Tu},\ and\ \citenamefont {Rowley}}]{chen2012variants}%
  \BibitemOpen
  \bibfield  {author} {\bibinfo {author} {\bibfnamefont {Kevin~K}\ \bibnamefont
  {Chen}}, \bibinfo {author} {\bibfnamefont {Jonathan~H}\ \bibnamefont {Tu}}, \
  and\ \bibinfo {author} {\bibfnamefont {Clarence~W}\ \bibnamefont {Rowley}},\
  }\bibfield  {title} {\enquote {\bibinfo {title} {Variants of dynamic mode
  decomposition: boundary condition, koopman, and fourier analyses},}\
  }\href@noop {} {\bibfield  {journal} {\bibinfo  {journal} {Journal of
  Nonlinear Science}\ }\textbf {\bibinfo {volume} {22}},\ \bibinfo {pages}
  {887--915} (\bibinfo {year} {2012})}\BibitemShut {NoStop}%
\bibitem [{\citenamefont {Bagheri}(2013)}]{bagheri2013koopman}%
  \BibitemOpen
  \bibfield  {author} {\bibinfo {author} {\bibfnamefont {Shervin}\ \bibnamefont
  {Bagheri}},\ }\bibfield  {title} {\enquote {\bibinfo {title} {Koopman-mode
  decomposition of the cylinder wake},}\ }\href@noop {} {\bibfield  {journal}
  {\bibinfo  {journal} {J. Fluid Mech}\ }\textbf {\bibinfo {volume} {726}},\
  \bibinfo {pages} {596--623} (\bibinfo {year} {2013})}\BibitemShut {NoStop}%
\bibitem [{\citenamefont {Sayadi}\ \emph {et~al.}(2014)\citenamefont {Sayadi},
  \citenamefont {Schmid}, \citenamefont {Nichols},\ and\ \citenamefont
  {Moin}}]{sayadi2014reduced}%
  \BibitemOpen
  \bibfield  {author} {\bibinfo {author} {\bibfnamefont {Taraneh}\ \bibnamefont
  {Sayadi}}, \bibinfo {author} {\bibfnamefont {Peter~J}\ \bibnamefont
  {Schmid}}, \bibinfo {author} {\bibfnamefont {Joseph~W}\ \bibnamefont
  {Nichols}}, \ and\ \bibinfo {author} {\bibfnamefont {Parviz}\ \bibnamefont
  {Moin}},\ }\bibfield  {title} {\enquote {\bibinfo {title} {Reduced-order
  representation of near-wall structures in the late transitional boundary
  layer},}\ }\href@noop {} {\bibfield  {journal} {\bibinfo  {journal} {Journal
  of Fluid Mechanics}\ }\textbf {\bibinfo {volume} {748}},\ \bibinfo {pages}
  {278--301} (\bibinfo {year} {2014})}\BibitemShut {NoStop}%
\bibitem [{\citenamefont {Subbareddy}\ \emph {et~al.}(2014)\citenamefont
  {Subbareddy}, \citenamefont {Bartkowicz},\ and\ \citenamefont
  {Candler}}]{subbareddy2014direct}%
  \BibitemOpen
  \bibfield  {author} {\bibinfo {author} {\bibfnamefont {Pramod~K}\
  \bibnamefont {Subbareddy}}, \bibinfo {author} {\bibfnamefont {Matthew~D}\
  \bibnamefont {Bartkowicz}}, \ and\ \bibinfo {author} {\bibfnamefont
  {Graham~V}\ \bibnamefont {Candler}},\ }\bibfield  {title} {\enquote {\bibinfo
  {title} {Direct numerical simulation of high-speed transition due to an
  isolated roughness element},}\ }\href@noop {} {\bibfield  {journal} {\bibinfo
   {journal} {Journal of Fluid Mechanics}\ }\textbf {\bibinfo {volume} {748}},\
  \bibinfo {pages} {848--878} (\bibinfo {year} {2014})}\BibitemShut {NoStop}%
\bibitem [{\citenamefont {Sayadi}\ \emph {et~al.}(2015)\citenamefont {Sayadi},
  \citenamefont {Schmid}, \citenamefont {Richecoeur},\ and\ \citenamefont
  {Durox}}]{sayadi2015parametrized}%
  \BibitemOpen
  \bibfield  {author} {\bibinfo {author} {\bibfnamefont {Taraneh}\ \bibnamefont
  {Sayadi}}, \bibinfo {author} {\bibfnamefont {Peter~J}\ \bibnamefont
  {Schmid}}, \bibinfo {author} {\bibfnamefont {Franck}\ \bibnamefont
  {Richecoeur}}, \ and\ \bibinfo {author} {\bibfnamefont {Daniel}\ \bibnamefont
  {Durox}},\ }\bibfield  {title} {\enquote {\bibinfo {title} {Parametrized
  data-driven decomposition for bifurcation analysis, with application to
  thermo-acoustically unstable systems},}\ }\href@noop {} {\bibfield  {journal}
  {\bibinfo  {journal} {Physics of Fluids}\ }\textbf {\bibinfo {volume} {27}},\
  \bibinfo {pages} {037102} (\bibinfo {year} {2015})}\BibitemShut {NoStop}%
\bibitem [{\citenamefont {Kramer}\ \emph {et~al.}(2015)\citenamefont {Kramer},
  \citenamefont {Grover}, \citenamefont {Boufounos}, \citenamefont {Benosman},\
  and\ \citenamefont {Nabi}}]{kramer2015sparse}%
  \BibitemOpen
  \bibfield  {author} {\bibinfo {author} {\bibfnamefont {Boris}\ \bibnamefont
  {Kramer}}, \bibinfo {author} {\bibfnamefont {Piyush}\ \bibnamefont {Grover}},
  \bibinfo {author} {\bibfnamefont {Petros}\ \bibnamefont {Boufounos}},
  \bibinfo {author} {\bibfnamefont {Mouhacine}\ \bibnamefont {Benosman}}, \
  and\ \bibinfo {author} {\bibfnamefont {Saleh}\ \bibnamefont {Nabi}},\
  }\bibfield  {title} {\enquote {\bibinfo {title} {Sparse sensing and dmd based
  identification of flow regimes and bifurcations in complex flows},}\
  }\href@noop {} {\bibfield  {journal} {\bibinfo  {journal} {arXiv preprint
  arXiv:1510.02831}\ } (\bibinfo {year} {2015})}\BibitemShut {NoStop}%
\bibitem [{\citenamefont {Swinney}\ and\ \citenamefont
  {Gollub}(1978)}]{swinney1978transition}%
  \BibitemOpen
  \bibfield  {author} {\bibinfo {author} {\bibfnamefont {Harry~L}\ \bibnamefont
  {Swinney}}\ and\ \bibinfo {author} {\bibfnamefont {Jerry~P}\ \bibnamefont
  {Gollub}},\ }\bibfield  {title} {\enquote {\bibinfo {title} {The transition
  to turbulence},}\ }\href@noop {} {\bibfield  {journal} {\bibinfo  {journal}
  {Phys. Today}\ }\textbf {\bibinfo {volume} {31}} (\bibinfo {year}
  {1978})}\BibitemShut {NoStop}%
\bibitem [{\citenamefont {Farmer}\ \emph {et~al.}(1980)\citenamefont {Farmer},
  \citenamefont {Crutchfield}, \citenamefont {Froehling}, \citenamefont
  {Packard},\ and\ \citenamefont {Shaw}}]{farmer1980power}%
  \BibitemOpen
  \bibfield  {author} {\bibinfo {author} {\bibfnamefont {Doyne}\ \bibnamefont
  {Farmer}}, \bibinfo {author} {\bibfnamefont {James}\ \bibnamefont
  {Crutchfield}}, \bibinfo {author} {\bibfnamefont {Harold}\ \bibnamefont
  {Froehling}}, \bibinfo {author} {\bibfnamefont {Norman}\ \bibnamefont
  {Packard}}, \ and\ \bibinfo {author} {\bibfnamefont {Robert}\ \bibnamefont
  {Shaw}},\ }\bibfield  {title} {\enquote {\bibinfo {title} {Power spectra and
  mixing properties of strange attractors},}\ }\href@noop {} {\bibfield
  {journal} {\bibinfo  {journal} {Annals of the New York Academy of Sciences}\
  }\textbf {\bibinfo {volume} {357}},\ \bibinfo {pages} {453--471} (\bibinfo
  {year} {1980})}\BibitemShut {NoStop}%
\bibitem [{\citenamefont {Karniadakis}\ and\ \citenamefont
  {Triantafyllou}(1992)}]{karniadakis1992three}%
  \BibitemOpen
  \bibfield  {author} {\bibinfo {author} {\bibfnamefont {George~Em}\
  \bibnamefont {Karniadakis}}\ and\ \bibinfo {author} {\bibfnamefont
  {George~S}\ \bibnamefont {Triantafyllou}},\ }\bibfield  {title} {\enquote
  {\bibinfo {title} {Three-dimensional dynamics and transition to turbulence in
  the wake of bluff objects},}\ }\href@noop {} {\bibfield  {journal} {\bibinfo
  {journal} {Journal of Fluid Mechanics}\ }\textbf {\bibinfo {volume} {238}},\
  \bibinfo {pages} {1--30} (\bibinfo {year} {1992})}\BibitemShut {NoStop}%
\bibitem [{\citenamefont {Tomboulides}\ and\ \citenamefont
  {Orszag}(2000)}]{tomboulides2000numerical}%
  \BibitemOpen
  \bibfield  {author} {\bibinfo {author} {\bibfnamefont {Ananias~G}\
  \bibnamefont {Tomboulides}}\ and\ \bibinfo {author} {\bibfnamefont
  {Steven~A}\ \bibnamefont {Orszag}},\ }\bibfield  {title} {\enquote {\bibinfo
  {title} {Numerical investigation of transitional and weak turbulent flow past
  a sphere},}\ }\href@noop {} {\bibfield  {journal} {\bibinfo  {journal}
  {Journal of Fluid Mechanics}\ }\textbf {\bibinfo {volume} {416}},\ \bibinfo
  {pages} {45--73} (\bibinfo {year} {2000})}\BibitemShut {NoStop}%
\bibitem [{\citenamefont {Peng}\ \emph {et~al.}(2003)\citenamefont {Peng},
  \citenamefont {Shiau},\ and\ \citenamefont {Hwang}}]{peng2003transition}%
  \BibitemOpen
  \bibfield  {author} {\bibinfo {author} {\bibfnamefont {Yih-Ferng}\
  \bibnamefont {Peng}}, \bibinfo {author} {\bibfnamefont {Yuo-Hsien}\
  \bibnamefont {Shiau}}, \ and\ \bibinfo {author} {\bibfnamefont {Robert~R}\
  \bibnamefont {Hwang}},\ }\bibfield  {title} {\enquote {\bibinfo {title}
  {Transition in a 2-d lid-driven cavity flow},}\ }\href@noop {} {\bibfield
  {journal} {\bibinfo  {journal} {Computers \& Fluids}\ }\textbf {\bibinfo
  {volume} {32}},\ \bibinfo {pages} {337--352} (\bibinfo {year}
  {2003})}\BibitemShut {NoStop}%
\bibitem [{\citenamefont {Basley}\ \emph {et~al.}(2011)\citenamefont {Basley},
  \citenamefont {Pastur}, \citenamefont {Lusseyran}, \citenamefont {Faure},\
  and\ \citenamefont {Delprat}}]{basley2011experimental}%
  \BibitemOpen
  \bibfield  {author} {\bibinfo {author} {\bibfnamefont {J}~\bibnamefont
  {Basley}}, \bibinfo {author} {\bibfnamefont {LR}~\bibnamefont {Pastur}},
  \bibinfo {author} {\bibfnamefont {F}~\bibnamefont {Lusseyran}}, \bibinfo
  {author} {\bibfnamefont {Th~M}\ \bibnamefont {Faure}}, \ and\ \bibinfo
  {author} {\bibfnamefont {N}~\bibnamefont {Delprat}},\ }\bibfield  {title}
  {\enquote {\bibinfo {title} {Experimental investigation of global structures
  in an incompressible cavity flow using time-resolved piv},}\ }\href@noop {}
  {\bibfield  {journal} {\bibinfo  {journal} {Experiments in Fluids}\ }\textbf
  {\bibinfo {volume} {50}},\ \bibinfo {pages} {905--918} (\bibinfo {year}
  {2011})}\BibitemShut {NoStop}%
\bibitem [{\citenamefont {Mezi\'c}(2017)}]{mezic2017koopman}%
  \BibitemOpen
  \bibfield  {author} {\bibinfo {author} {\bibfnamefont {Igor}\ \bibnamefont
  {Mezi\'c}},\ }\bibfield  {title} {\enquote {\bibinfo {title} {Koopman
  operator spectrum and data analysis},}\ }\href@noop {} {\bibfield  {journal}
  {\bibinfo  {journal} {arXiv preprint arXiv:1702.07597}\ } (\bibinfo {year}
  {2017})}\BibitemShut {NoStop}%
\bibitem [{\citenamefont {Giannakis}(2017)}]{giannakis2017data}%
  \BibitemOpen
  \bibfield  {author} {\bibinfo {author} {\bibfnamefont {Dimitrios}\
  \bibnamefont {Giannakis}},\ }\bibfield  {title} {\enquote {\bibinfo {title}
  {Data-driven spectral decomposition and forecasting of ergodic dynamical
  systems},}\ }\href@noop {} {\bibfield  {journal} {\bibinfo  {journal}
  {Applied and Computational Harmonic Analysis}\ } (\bibinfo {year}
  {2017})}\BibitemShut {NoStop}%
\bibitem [{\citenamefont {Korda}\ and\ \citenamefont
  {Mezi{\'c}}(2016)}]{korda2016linear}%
  \BibitemOpen
  \bibfield  {author} {\bibinfo {author} {\bibfnamefont {Milan}\ \bibnamefont
  {Korda}}\ and\ \bibinfo {author} {\bibfnamefont {Igor}\ \bibnamefont
  {Mezi{\'c}}},\ }\bibfield  {title} {\enquote {\bibinfo {title} {Linear
  predictors for nonlinear dynamical systems: Koopman operator meets model
  predictive control},}\ }\href@noop {} {\bibfield  {journal} {\bibinfo
  {journal} {arXiv preprint arXiv:1611.03537}\ } (\bibinfo {year}
  {2016})}\BibitemShut {NoStop}%
\bibitem [{\citenamefont {Surana}\ and\ \citenamefont
  {Banaszuk}(2016)}]{surana2016linear}%
  \BibitemOpen
  \bibfield  {author} {\bibinfo {author} {\bibfnamefont {Amit}\ \bibnamefont
  {Surana}}\ and\ \bibinfo {author} {\bibfnamefont {Andrzej}\ \bibnamefont
  {Banaszuk}},\ }\bibfield  {title} {\enquote {\bibinfo {title} {Linear
  observer synthesis for nonlinear systems using koopman operator framework},}\
  }\href@noop {} {\bibfield  {journal} {\bibinfo  {journal}
  {IFAC-PapersOnLine}\ }\textbf {\bibinfo {volume} {49}},\ \bibinfo {pages}
  {716--723} (\bibinfo {year} {2016})}\BibitemShut {NoStop}%
\bibitem [{\citenamefont {Surana}(2016)}]{surana2016koopman}%
  \BibitemOpen
  \bibfield  {author} {\bibinfo {author} {\bibfnamefont {Amit}\ \bibnamefont
  {Surana}},\ }\bibfield  {title} {\enquote {\bibinfo {title} {Koopman operator
  based observer synthesis for control-affine nonlinear systems},}\ }in\
  \href@noop {} {\emph {\bibinfo {booktitle} {Decision and Control (CDC), 2016
  IEEE 55th Conference on}}}\ (\bibinfo {organization} {IEEE},\ \bibinfo {year}
  {2016})\ pp.\ \bibinfo {pages} {6492--6499}\BibitemShut {NoStop}%
\bibitem [{\citenamefont {Kaiser}\ \emph {et~al.}(2017)\citenamefont {Kaiser},
  \citenamefont {Kutz},\ and\ \citenamefont {Brunton}}]{kaiser2017data}%
  \BibitemOpen
  \bibfield  {author} {\bibinfo {author} {\bibfnamefont {Eurika}\ \bibnamefont
  {Kaiser}}, \bibinfo {author} {\bibfnamefont {J~Nathan}\ \bibnamefont {Kutz}},
  \ and\ \bibinfo {author} {\bibfnamefont {Steven~L}\ \bibnamefont {Brunton}},\
  }\bibfield  {title} {\enquote {\bibinfo {title} {Data-driven discovery of
  koopman eigenfunctions for control},}\ }\href@noop {} {\bibfield  {journal}
  {\bibinfo  {journal} {arXiv preprint arXiv:1707.01146}\ } (\bibinfo {year}
  {2017})}\BibitemShut {NoStop}%
\bibitem [{\citenamefont {Taira}\ \emph {et~al.}(2017)\citenamefont {Taira},
  \citenamefont {Brunton}, \citenamefont {Dawson}, \citenamefont {Rowley},
  \citenamefont {Colonius}, \citenamefont {McKeon}, \citenamefont {Schmidt},
  \citenamefont {Gordeyev}, \citenamefont {Theofilis},\ and\ \citenamefont
  {Ukeiley}}]{taira2017modal}%
  \BibitemOpen
  \bibfield  {author} {\bibinfo {author} {\bibfnamefont {Kunihiko}\
  \bibnamefont {Taira}}, \bibinfo {author} {\bibfnamefont {Steven~L}\
  \bibnamefont {Brunton}}, \bibinfo {author} {\bibfnamefont {Scott}\
  \bibnamefont {Dawson}}, \bibinfo {author} {\bibfnamefont {Clarence~W}\
  \bibnamefont {Rowley}}, \bibinfo {author} {\bibfnamefont {Tim}\ \bibnamefont
  {Colonius}}, \bibinfo {author} {\bibfnamefont {Beverley~J}\ \bibnamefont
  {McKeon}}, \bibinfo {author} {\bibfnamefont {Oliver~T}\ \bibnamefont
  {Schmidt}}, \bibinfo {author} {\bibfnamefont {Stanislav}\ \bibnamefont
  {Gordeyev}}, \bibinfo {author} {\bibfnamefont {Vassilios}\ \bibnamefont
  {Theofilis}}, \ and\ \bibinfo {author} {\bibfnamefont {Lawrence~S}\
  \bibnamefont {Ukeiley}},\ }\bibfield  {title} {\enquote {\bibinfo {title}
  {Modal analysis of fluid flows: An overview},}\ }\href@noop {} {\bibfield
  {journal} {\bibinfo  {journal} {arXiv preprint arXiv:1702.01453}\ } (\bibinfo
  {year} {2017})}\BibitemShut {NoStop}%
\bibitem [{\citenamefont {Williams}\ \emph {et~al.}(2015)\citenamefont
  {Williams}, \citenamefont {Kevrekidis},\ and\ \citenamefont
  {Rowley}}]{williams2015data}%
  \BibitemOpen
  \bibfield  {author} {\bibinfo {author} {\bibfnamefont {Matthew~O}\
  \bibnamefont {Williams}}, \bibinfo {author} {\bibfnamefont {Ioannis~G}\
  \bibnamefont {Kevrekidis}}, \ and\ \bibinfo {author} {\bibfnamefont
  {Clarence~W}\ \bibnamefont {Rowley}},\ }\bibfield  {title} {\enquote
  {\bibinfo {title} {A data--driven approximation of the koopman operator:
  Extending dynamic mode decomposition},}\ }\href@noop {} {\bibfield  {journal}
  {\bibinfo  {journal} {Journal of Nonlinear Science}\ }\textbf {\bibinfo
  {volume} {25}},\ \bibinfo {pages} {1307--1346} (\bibinfo {year}
  {2015})}\BibitemShut {NoStop}%
\bibitem [{\citenamefont {Tu}\ \emph {et~al.}(2014)\citenamefont {Tu},
  \citenamefont {Rowley}, \citenamefont {Luchtenburg}, \citenamefont
  {Brunton},\ and\ \citenamefont {Kutz}}]{tu2014dynamic}%
  \BibitemOpen
  \bibfield  {author} {\bibinfo {author} {\bibfnamefont {Jonathan~H}\
  \bibnamefont {Tu}}, \bibinfo {author} {\bibfnamefont {Clarence~W}\
  \bibnamefont {Rowley}}, \bibinfo {author} {\bibfnamefont {Dirk~M}\
  \bibnamefont {Luchtenburg}}, \bibinfo {author} {\bibfnamefont {Steven~L}\
  \bibnamefont {Brunton}}, \ and\ \bibinfo {author} {\bibfnamefont {J~Nathan}\
  \bibnamefont {Kutz}},\ }\bibfield  {title} {\enquote {\bibinfo {title} {On
  dynamic mode decomposition: theory and applications},}\ }\href@noop {}
  {\bibfield  {journal} {\bibinfo  {journal} {Journal of Computational
  Dynamics}\ } (\bibinfo {year} {2014})}\BibitemShut {NoStop}%
\bibitem [{\citenamefont {Kutz}\ \emph {et~al.}(2016)\citenamefont {Kutz},
  \citenamefont {Fu},\ and\ \citenamefont {Brunton}}]{kutz2016multiresolution}%
  \BibitemOpen
  \bibfield  {author} {\bibinfo {author} {\bibfnamefont {J~Nathan}\
  \bibnamefont {Kutz}}, \bibinfo {author} {\bibfnamefont {Xing}\ \bibnamefont
  {Fu}}, \ and\ \bibinfo {author} {\bibfnamefont {Steven~L}\ \bibnamefont
  {Brunton}},\ }\bibfield  {title} {\enquote {\bibinfo {title} {Multiresolution
  dynamic mode decomposition},}\ }\href@noop {} {\bibfield  {journal} {\bibinfo
   {journal} {SIAM Journal on Applied Dynamical Systems}\ }\textbf {\bibinfo
  {volume} {15}},\ \bibinfo {pages} {713--735} (\bibinfo {year}
  {2016})}\BibitemShut {NoStop}%
\bibitem [{\citenamefont {Proctor}\ \emph {et~al.}(2016)\citenamefont
  {Proctor}, \citenamefont {Brunton},\ and\ \citenamefont
  {Kutz}}]{proctor2016dynamic}%
  \BibitemOpen
  \bibfield  {author} {\bibinfo {author} {\bibfnamefont {Joshua~L}\
  \bibnamefont {Proctor}}, \bibinfo {author} {\bibfnamefont {Steven~L}\
  \bibnamefont {Brunton}}, \ and\ \bibinfo {author} {\bibfnamefont {J~Nathan}\
  \bibnamefont {Kutz}},\ }\bibfield  {title} {\enquote {\bibinfo {title}
  {Dynamic mode decomposition with control},}\ }\href@noop {} {\bibfield
  {journal} {\bibinfo  {journal} {SIAM Journal on Applied Dynamical Systems}\
  }\textbf {\bibinfo {volume} {15}},\ \bibinfo {pages} {142--161} (\bibinfo
  {year} {2016})}\BibitemShut {NoStop}%
\bibitem [{\citenamefont {Hemati}\ \emph {et~al.}(2014)\citenamefont {Hemati},
  \citenamefont {Williams},\ and\ \citenamefont {Rowley}}]{hemati2014dynamic}%
  \BibitemOpen
  \bibfield  {author} {\bibinfo {author} {\bibfnamefont {Maziar~S}\
  \bibnamefont {Hemati}}, \bibinfo {author} {\bibfnamefont {Matthew~O}\
  \bibnamefont {Williams}}, \ and\ \bibinfo {author} {\bibfnamefont
  {Clarence~W}\ \bibnamefont {Rowley}},\ }\bibfield  {title} {\enquote
  {\bibinfo {title} {Dynamic mode decomposition for large and streaming
  datasets},}\ }\href@noop {} {\bibfield  {journal} {\bibinfo  {journal}
  {Physics of Fluids}\ }\textbf {\bibinfo {volume} {26}},\ \bibinfo {pages}
  {111701} (\bibinfo {year} {2014})}\BibitemShut {NoStop}%
\bibitem [{\citenamefont {Gu{\'e}niat}\ \emph {et~al.}(2015)\citenamefont
  {Gu{\'e}niat}, \citenamefont {Mathelin},\ and\ \citenamefont
  {Pastur}}]{gueniat2015dynamic}%
  \BibitemOpen
  \bibfield  {author} {\bibinfo {author} {\bibfnamefont {Florimond}\
  \bibnamefont {Gu{\'e}niat}}, \bibinfo {author} {\bibfnamefont {Lionel}\
  \bibnamefont {Mathelin}}, \ and\ \bibinfo {author} {\bibfnamefont {Luc~R}\
  \bibnamefont {Pastur}},\ }\bibfield  {title} {\enquote {\bibinfo {title} {A
  dynamic mode decomposition approach for large and arbitrarily sampled
  systems},}\ }\href@noop {} {\bibfield  {journal} {\bibinfo  {journal}
  {Physics of Fluids}\ }\textbf {\bibinfo {volume} {27}},\ \bibinfo {pages}
  {025113} (\bibinfo {year} {2015})}\BibitemShut {NoStop}%
\bibitem [{\citenamefont {Brunton}\ \emph {et~al.}(2013)\citenamefont
  {Brunton}, \citenamefont {Proctor},\ and\ \citenamefont
  {Kutz}}]{brunton2013compressive}%
  \BibitemOpen
  \bibfield  {author} {\bibinfo {author} {\bibfnamefont {Steven~L}\
  \bibnamefont {Brunton}}, \bibinfo {author} {\bibfnamefont {Joshua~L}\
  \bibnamefont {Proctor}}, \ and\ \bibinfo {author} {\bibfnamefont {J~Nathan}\
  \bibnamefont {Kutz}},\ }\bibfield  {title} {\enquote {\bibinfo {title}
  {Compressive sampling and dynamic mode decomposition},}\ }\href@noop {}
  {\bibfield  {journal} {\bibinfo  {journal} {arXiv preprint arXiv:1312.5186}\
  } (\bibinfo {year} {2013})}\BibitemShut {NoStop}%
\bibitem [{\citenamefont {Dawson}\ \emph {et~al.}(2016)\citenamefont {Dawson},
  \citenamefont {Hemati}, \citenamefont {Williams},\ and\ \citenamefont
  {Rowley}}]{dawson2016characterizing}%
  \BibitemOpen
  \bibfield  {author} {\bibinfo {author} {\bibfnamefont {Scott~TM}\
  \bibnamefont {Dawson}}, \bibinfo {author} {\bibfnamefont {Maziar~S}\
  \bibnamefont {Hemati}}, \bibinfo {author} {\bibfnamefont {Matthew~O}\
  \bibnamefont {Williams}}, \ and\ \bibinfo {author} {\bibfnamefont
  {Clarence~W}\ \bibnamefont {Rowley}},\ }\bibfield  {title} {\enquote
  {\bibinfo {title} {Characterizing and correcting for the effect of sensor
  noise in the dynamic mode decomposition},}\ }\href@noop {} {\bibfield
  {journal} {\bibinfo  {journal} {Experiments in Fluids}\ }\textbf {\bibinfo
  {volume} {57}},\ \bibinfo {pages} {1--19} (\bibinfo {year}
  {2016})}\BibitemShut {NoStop}%
\bibitem [{\citenamefont {Hemati}\ \emph {et~al.}(2017)\citenamefont {Hemati},
  \citenamefont {Rowley}, \citenamefont {Deem},\ and\ \citenamefont
  {Cattafesta}}]{hemati2017biasing}%
  \BibitemOpen
  \bibfield  {author} {\bibinfo {author} {\bibfnamefont {Maziar~S}\
  \bibnamefont {Hemati}}, \bibinfo {author} {\bibfnamefont {Clarence~W}\
  \bibnamefont {Rowley}}, \bibinfo {author} {\bibfnamefont {Eric~A}\
  \bibnamefont {Deem}}, \ and\ \bibinfo {author} {\bibfnamefont {Louis~N}\
  \bibnamefont {Cattafesta}},\ }\bibfield  {title} {\enquote {\bibinfo {title}
  {De-biasing the dynamic mode decomposition for applied koopman spectral
  analysis of noisy datasets},}\ }\href@noop {} {\bibfield  {journal} {\bibinfo
   {journal} {Theoretical and Computational Fluid Dynamics}\ ,\ \bibinfo
  {pages} {1--20}} (\bibinfo {year} {2017})}\BibitemShut {NoStop}%
\bibitem [{\citenamefont {Laskar}(1990)}]{laskar1990chaotic}%
  \BibitemOpen
  \bibfield  {author} {\bibinfo {author} {\bibfnamefont {Jacques}\ \bibnamefont
  {Laskar}},\ }\bibfield  {title} {\enquote {\bibinfo {title} {The chaotic
  motion of the solar system: a numerical estimate of the size of the chaotic
  zones},}\ }\href@noop {} {\bibfield  {journal} {\bibinfo  {journal} {Icarus}\
  }\textbf {\bibinfo {volume} {88}},\ \bibinfo {pages} {266--291} (\bibinfo
  {year} {1990})}\BibitemShut {NoStop}%
\bibitem [{\citenamefont {Laskar}\ \emph {et~al.}(1992)\citenamefont {Laskar},
  \citenamefont {Froeschl{\'e}},\ and\ \citenamefont
  {Celletti}}]{laskar1992measure}%
  \BibitemOpen
  \bibfield  {author} {\bibinfo {author} {\bibfnamefont {Jacques}\ \bibnamefont
  {Laskar}}, \bibinfo {author} {\bibfnamefont {Claude}\ \bibnamefont
  {Froeschl{\'e}}}, \ and\ \bibinfo {author} {\bibfnamefont {Alessandra}\
  \bibnamefont {Celletti}},\ }\bibfield  {title} {\enquote {\bibinfo {title}
  {The measure of chaos by the numerical analysis of the fundamental
  frequencies. application to the standard mapping},}\ }\href@noop {}
  {\bibfield  {journal} {\bibinfo  {journal} {Physica D: Nonlinear Phenomena}\
  }\textbf {\bibinfo {volume} {56}},\ \bibinfo {pages} {253--269} (\bibinfo
  {year} {1992})}\BibitemShut {NoStop}%
\bibitem [{\citenamefont {Welch}(1967)}]{welch1967use}%
  \BibitemOpen
  \bibfield  {author} {\bibinfo {author} {\bibfnamefont {Peter}\ \bibnamefont
  {Welch}},\ }\bibfield  {title} {\enquote {\bibinfo {title} {The use of fast
  fourier transform for the estimation of power spectra: a method based on time
  averaging over short, modified periodograms},}\ }\href@noop {} {\bibfield
  {journal} {\bibinfo  {journal} {IEEE Transactions on audio and
  electroacoustics}\ }\textbf {\bibinfo {volume} {15}},\ \bibinfo {pages}
  {70--73} (\bibinfo {year} {1967})}\BibitemShut {NoStop}%
\bibitem [{\citenamefont {Wynn}\ \emph {et~al.}(2013)\citenamefont {Wynn},
  \citenamefont {Pearson}, \citenamefont {Ganapathisubramani},\ and\
  \citenamefont {Goulart}}]{wynn2013optimal}%
  \BibitemOpen
  \bibfield  {author} {\bibinfo {author} {\bibfnamefont {A}~\bibnamefont
  {Wynn}}, \bibinfo {author} {\bibfnamefont {DS}~\bibnamefont {Pearson}},
  \bibinfo {author} {\bibfnamefont {B}~\bibnamefont {Ganapathisubramani}}, \
  and\ \bibinfo {author} {\bibfnamefont {PJ}~\bibnamefont {Goulart}},\
  }\bibfield  {title} {\enquote {\bibinfo {title} {Optimal mode decomposition
  for unsteady flows},}\ }\href@noop {} {\bibfield  {journal} {\bibinfo
  {journal} {Journal of Fluid Mechanics}\ }\textbf {\bibinfo {volume} {733}},\
  \bibinfo {pages} {473} (\bibinfo {year} {2013})}\BibitemShut {NoStop}%
\bibitem [{\citenamefont {Jovanovi{\'c}}\ \emph {et~al.}(2014)\citenamefont
  {Jovanovi{\'c}}, \citenamefont {Schmid},\ and\ \citenamefont
  {Nichols}}]{jovanovic2014sparsity}%
  \BibitemOpen
  \bibfield  {author} {\bibinfo {author} {\bibfnamefont {Mihailo~R}\
  \bibnamefont {Jovanovi{\'c}}}, \bibinfo {author} {\bibfnamefont {Peter~J}\
  \bibnamefont {Schmid}}, \ and\ \bibinfo {author} {\bibfnamefont {Joseph~W}\
  \bibnamefont {Nichols}},\ }\bibfield  {title} {\enquote {\bibinfo {title}
  {Sparsity-promoting dynamic mode decomposition},}\ }\href@noop {} {\bibfield
  {journal} {\bibinfo  {journal} {Physics of Fluids (1994-present)}\ }\textbf
  {\bibinfo {volume} {26}},\ \bibinfo {pages} {024103} (\bibinfo {year}
  {2014})}\BibitemShut {NoStop}%
\bibitem [{\citenamefont {Tissot}\ \emph {et~al.}(2014)\citenamefont {Tissot},
  \citenamefont {Cordier}, \citenamefont {Benard},\ and\ \citenamefont
  {Noack}}]{tissot2014model}%
  \BibitemOpen
  \bibfield  {author} {\bibinfo {author} {\bibfnamefont {Gilles}\ \bibnamefont
  {Tissot}}, \bibinfo {author} {\bibfnamefont {Laurent}\ \bibnamefont
  {Cordier}}, \bibinfo {author} {\bibfnamefont {Nicolas}\ \bibnamefont
  {Benard}}, \ and\ \bibinfo {author} {\bibfnamefont {Bernd~R}\ \bibnamefont
  {Noack}},\ }\bibfield  {title} {\enquote {\bibinfo {title} {Model reduction
  using dynamic mode decomposition},}\ }\href@noop {} {\bibfield  {journal}
  {\bibinfo  {journal} {Comptes Rendus M{\'e}canique}\ }\textbf {\bibinfo
  {volume} {342}},\ \bibinfo {pages} {410--416} (\bibinfo {year}
  {2014})}\BibitemShut {NoStop}%
\bibitem [{\citenamefont {Mezi\'c}(2013)}]{mezic2013analysis}%
  \BibitemOpen
  \bibfield  {author} {\bibinfo {author} {\bibfnamefont {Igor}\ \bibnamefont
  {Mezi\'c}},\ }\bibfield  {title} {\enquote {\bibinfo {title} {Analysis of
  fluid flows via spectral properties of the koopman operator},}\ }\href@noop
  {} {\bibfield  {journal} {\bibinfo  {journal} {Annual Review of Fluid
  Mechanics}\ }\textbf {\bibinfo {volume} {45}},\ \bibinfo {pages} {357--378}
  (\bibinfo {year} {2013})}\BibitemShut {NoStop}%
\bibitem [{\citenamefont {Wiggins}(2003)}]{wiggins2003introduction}%
  \BibitemOpen
  \bibfield  {author} {\bibinfo {author} {\bibfnamefont {Stephen}\ \bibnamefont
  {Wiggins}},\ }\href@noop {} {\emph {\bibinfo {title} {Introduction to applied
  nonlinear dynamical systems and chaos}}},\ Vol.~\bibinfo {volume} {2}\
  (\bibinfo  {publisher} {Springer Science \& Business Media},\ \bibinfo {year}
  {2003})\BibitemShut {NoStop}%
\bibitem [{\citenamefont {Neumann}(1932)}]{neumann1932operatorenmethode}%
  \BibitemOpen
  \bibfield  {author} {\bibinfo {author} {\bibfnamefont {J~v}\ \bibnamefont
  {Neumann}},\ }\bibfield  {title} {\enquote {\bibinfo {title} {Zur
  operatorenmethode in der klassischen mechanik},}\ }\href@noop {} {\bibfield
  {journal} {\bibinfo  {journal} {Annals of Mathematics}\ ,\ \bibinfo {pages}
  {587--642}} (\bibinfo {year} {1932})}\BibitemShut {NoStop}%
\bibitem [{\citenamefont {MacCluer}(2008)}]{maccluer2008elementary}%
  \BibitemOpen
  \bibfield  {author} {\bibinfo {author} {\bibfnamefont {Barbara}\ \bibnamefont
  {MacCluer}},\ }\href@noop {} {\emph {\bibinfo {title} {Elementary functional
  analysis}}},\ Vol.\ \bibinfo {volume} {253}\ (\bibinfo  {publisher} {Springer
  Science \& Business Media},\ \bibinfo {year} {2008})\BibitemShut {NoStop}%
\bibitem [{\citenamefont {Luzzatto}\ \emph {et~al.}(2005)\citenamefont
  {Luzzatto}, \citenamefont {Melbourne},\ and\ \citenamefont
  {Paccaut}}]{luzzatto2005lorenz}%
  \BibitemOpen
  \bibfield  {author} {\bibinfo {author} {\bibfnamefont {Stefano}\ \bibnamefont
  {Luzzatto}}, \bibinfo {author} {\bibfnamefont {Ian}\ \bibnamefont
  {Melbourne}}, \ and\ \bibinfo {author} {\bibfnamefont {Frederic}\
  \bibnamefont {Paccaut}},\ }\bibfield  {title} {\enquote {\bibinfo {title}
  {The lorenz attractor is mixing},}\ }\href@noop {} {\bibfield  {journal}
  {\bibinfo  {journal} {Communications in Mathematical Physics}\ }\textbf
  {\bibinfo {volume} {260}},\ \bibinfo {pages} {393--401} (\bibinfo {year}
  {2005})}\BibitemShut {NoStop}%
\bibitem [{\citenamefont {Petersen}(1989)}]{petersen1989ergodic}%
  \BibitemOpen
  \bibfield  {author} {\bibinfo {author} {\bibfnamefont {Karl~E}\ \bibnamefont
  {Petersen}},\ }\href@noop {} {\emph {\bibinfo {title} {Ergodic theory}}},\
  Vol.~\bibinfo {volume} {2}\ (\bibinfo  {publisher} {Cambridge University
  Press},\ \bibinfo {year} {1989})\BibitemShut {NoStop}%
\bibitem [{\citenamefont {Doob}\ and\ \citenamefont
  {Doob}(1953)}]{doob1953stochastic}%
  \BibitemOpen
  \bibfield  {author} {\bibinfo {author} {\bibfnamefont {Joseph~L}\
  \bibnamefont {Doob}}\ and\ \bibinfo {author} {\bibfnamefont {Joseph~L}\
  \bibnamefont {Doob}},\ }\href@noop {} {\emph {\bibinfo {title} {Stochastic
  processes}}},\ Vol.~\bibinfo {volume} {7}\ (\bibinfo  {publisher} {Wiley New
  York},\ \bibinfo {year} {1953})\BibitemShut {NoStop}%
\bibitem [{\citenamefont {Peebles}()}]{peebles1980probability}%
  \BibitemOpen
  \bibfield  {author} {\bibinfo {author} {\bibfnamefont {Peyton~Z}\
  \bibnamefont {Peebles}},\ }\href@noop {} {\emph {\bibinfo {title}
  {Probability, random variables, and random signal principles}}}\BibitemShut
  {NoStop}%
\bibitem [{\citenamefont {Tseng}\ and\ \citenamefont
  {Ferziger}(2001)}]{tseng2001mixing}%
  \BibitemOpen
  \bibfield  {author} {\bibinfo {author} {\bibfnamefont {Yu-heng}\ \bibnamefont
  {Tseng}}\ and\ \bibinfo {author} {\bibfnamefont {Joel~H}\ \bibnamefont
  {Ferziger}},\ }\bibfield  {title} {\enquote {\bibinfo {title} {Mixing and
  available potential energy in stratified flows},}\ }\href@noop {} {\bibfield
  {journal} {\bibinfo  {journal} {Physics of Fluids}\ }\textbf {\bibinfo
  {volume} {13}},\ \bibinfo {pages} {1281--1293} (\bibinfo {year}
  {2001})}\BibitemShut {NoStop}%
\bibitem [{\citenamefont {Gildor}\ \emph {et~al.}(2010)\citenamefont {Gildor},
  \citenamefont {Fredj},\ and\ \citenamefont {Kostinski}}]{gildor2010gulf}%
  \BibitemOpen
  \bibfield  {author} {\bibinfo {author} {\bibfnamefont {Hezi}\ \bibnamefont
  {Gildor}}, \bibinfo {author} {\bibfnamefont {Erick}\ \bibnamefont {Fredj}}, \
  and\ \bibinfo {author} {\bibfnamefont {Alex}\ \bibnamefont {Kostinski}},\
  }\bibfield  {title} {\enquote {\bibinfo {title} {The gulf of eilat/aqaba: a
  natural driven cavity?}}\ }\href@noop {} {\bibfield  {journal} {\bibinfo
  {journal} {Geophysical and Astrophysical Fluid Dynamics}\ }\textbf {\bibinfo
  {volume} {104}},\ \bibinfo {pages} {301--308} (\bibinfo {year}
  {2010})}\BibitemShut {NoStop}%
\bibitem [{\citenamefont {Chella}\ and\ \citenamefont
  {Ottino}(1985)}]{chella1985fluid}%
  \BibitemOpen
  \bibfield  {author} {\bibinfo {author} {\bibfnamefont {Ravindran}\
  \bibnamefont {Chella}}\ and\ \bibinfo {author} {\bibfnamefont {Julio~M}\
  \bibnamefont {Ottino}},\ }\bibfield  {title} {\enquote {\bibinfo {title}
  {Fluid mechanics of mixing in a single-screw extruder},}\ }\href@noop {}
  {\bibfield  {journal} {\bibinfo  {journal} {Industrial \& engineering
  chemistry fundamentals}\ }\textbf {\bibinfo {volume} {24}},\ \bibinfo {pages}
  {170--180} (\bibinfo {year} {1985})}\BibitemShut {NoStop}%
\bibitem [{\citenamefont {Gharib}\ and\ \citenamefont
  {Derango}(1989)}]{gharib1989liquid}%
  \BibitemOpen
  \bibfield  {author} {\bibinfo {author} {\bibfnamefont {Morteza}\ \bibnamefont
  {Gharib}}\ and\ \bibinfo {author} {\bibfnamefont {Philip}\ \bibnamefont
  {Derango}},\ }\bibfield  {title} {\enquote {\bibinfo {title} {A liquid film
  (soap film) tunnel to study two-dimensional laminar and turbulent shear
  flows},}\ }\href@noop {} {\bibfield  {journal} {\bibinfo  {journal} {Physica
  D: Nonlinear Phenomena}\ }\textbf {\bibinfo {volume} {37}},\ \bibinfo {pages}
  {406--416} (\bibinfo {year} {1989})}\BibitemShut {NoStop}%
\bibitem [{\citenamefont {Koseff}\ and\ \citenamefont
  {Street}(1984)}]{koseff1984}%
  \BibitemOpen
  \bibfield  {author} {\bibinfo {author} {\bibfnamefont {JR}~\bibnamefont
  {Koseff}}\ and\ \bibinfo {author} {\bibfnamefont {RL}~\bibnamefont
  {Street}},\ }\bibfield  {title} {\enquote {\bibinfo {title} {The lid-driven
  cavity flow: a synthesis of qualitative and quantitative observations},}\
  }\href@noop {} {\bibfield  {journal} {\bibinfo  {journal} {Journal of Fluids
  Engineering}\ }\textbf {\bibinfo {volume} {106}},\ \bibinfo {pages}
  {390--398} (\bibinfo {year} {1984})}\BibitemShut {NoStop}%
\bibitem [{\citenamefont {Shen}(1991)}]{Shen1991}%
  \BibitemOpen
  \bibfield  {author} {\bibinfo {author} {\bibfnamefont {J}~\bibnamefont
  {Shen}},\ }\bibfield  {title} {\enquote {\bibinfo {title} {{Hopf bifurcation
  of the unsteady regularized driven cavity flow}},}\ }\href@noop {} {\bibfield
   {journal} {\bibinfo  {journal} {Journal of Computational Physics}\ }\textbf
  {\bibinfo {volume} {245}},\ \bibinfo {pages} {228--245} (\bibinfo {year}
  {1991})}\BibitemShut {NoStop}%
\bibitem [{\citenamefont {Poliashenko}\ and\ \citenamefont
  {Aidun}(1995)}]{poliashenko1995direct}%
  \BibitemOpen
  \bibfield  {author} {\bibinfo {author} {\bibfnamefont {Maxim}\ \bibnamefont
  {Poliashenko}}\ and\ \bibinfo {author} {\bibfnamefont {Cyrus~K}\ \bibnamefont
  {Aidun}},\ }\bibfield  {title} {\enquote {\bibinfo {title} {A direct method
  for computation of simple bifurcations},}\ }\href@noop {} {\bibfield
  {journal} {\bibinfo  {journal} {Journal of Computational Physics}\ }\textbf
  {\bibinfo {volume} {121}},\ \bibinfo {pages} {246--260} (\bibinfo {year}
  {1995})}\BibitemShut {NoStop}%
\bibitem [{\citenamefont {Cazemier}\ \emph {et~al.}(1998)\citenamefont
  {Cazemier}, \citenamefont {Verstappen},\ and\ \citenamefont
  {Veldman}}]{cazemier1998}%
  \BibitemOpen
  \bibfield  {author} {\bibinfo {author} {\bibfnamefont {W}~\bibnamefont
  {Cazemier}}, \bibinfo {author} {\bibfnamefont {RWCP}\ \bibnamefont
  {Verstappen}}, \ and\ \bibinfo {author} {\bibfnamefont {AEP}\ \bibnamefont
  {Veldman}},\ }\bibfield  {title} {\enquote {\bibinfo {title} {Proper
  orthogonal decomposition and low-dimensional models for driven cavity
  flows},}\ }\href@noop {} {\bibfield  {journal} {\bibinfo  {journal} {Physics
  of Fluids (1994-present)}\ }\textbf {\bibinfo {volume} {10}},\ \bibinfo
  {pages} {1685--1699} (\bibinfo {year} {1998})}\BibitemShut {NoStop}%
\bibitem [{\citenamefont {Auteri}\ \emph {et~al.}(2002)\citenamefont {Auteri},
  \citenamefont {Parolini},\ and\ \citenamefont
  {Quartapelle}}]{auteri2002numerical}%
  \BibitemOpen
  \bibfield  {author} {\bibinfo {author} {\bibfnamefont {F}~\bibnamefont
  {Auteri}}, \bibinfo {author} {\bibfnamefont {N}~\bibnamefont {Parolini}}, \
  and\ \bibinfo {author} {\bibfnamefont {L}~\bibnamefont {Quartapelle}},\
  }\bibfield  {title} {\enquote {\bibinfo {title} {Numerical investigation on
  the stability of singular driven cavity flow},}\ }\href@noop {} {\bibfield
  {journal} {\bibinfo  {journal} {Journal of Computational Physics}\ }\textbf
  {\bibinfo {volume} {183}},\ \bibinfo {pages} {1--25} (\bibinfo {year}
  {2002})}\BibitemShut {NoStop}%
\bibitem [{\citenamefont {Balajewicz}\ \emph {et~al.}(2013)\citenamefont
  {Balajewicz}, \citenamefont {Dowell},\ and\ \citenamefont
  {Noack}}]{balajewicz2013low}%
  \BibitemOpen
  \bibfield  {author} {\bibinfo {author} {\bibfnamefont {Maciej~J}\
  \bibnamefont {Balajewicz}}, \bibinfo {author} {\bibfnamefont {Earl~H}\
  \bibnamefont {Dowell}}, \ and\ \bibinfo {author} {\bibfnamefont {Bernd~R}\
  \bibnamefont {Noack}},\ }\bibfield  {title} {\enquote {\bibinfo {title}
  {Low-dimensional modelling of high-reynolds-number shear flows incorporating
  constraints from the navier--stokes equation},}\ }\href@noop {} {\bibfield
  {journal} {\bibinfo  {journal} {Journal of Fluid Mechanics}\ }\textbf
  {\bibinfo {volume} {729}},\ \bibinfo {pages} {285--308} (\bibinfo {year}
  {2013})}\BibitemShut {NoStop}%
\bibitem [{\citenamefont {Botella}(1997)}]{botella1997solution}%
  \BibitemOpen
  \bibfield  {author} {\bibinfo {author} {\bibfnamefont {Olivier}\ \bibnamefont
  {Botella}},\ }\bibfield  {title} {\enquote {\bibinfo {title} {On the solution
  of the navier-stokes equations using chebyshev projection schemes with
  third-order accuracy in time},}\ }\href@noop {} {\bibfield  {journal}
  {\bibinfo  {journal} {Computers \& Fluids}\ }\textbf {\bibinfo {volume}
  {26}},\ \bibinfo {pages} {107--116} (\bibinfo {year} {1997})}\BibitemShut
  {NoStop}%
\bibitem [{\citenamefont {Temam}(1988)}]{temam1988infinite}%
  \BibitemOpen
  \bibfield  {author} {\bibinfo {author} {\bibfnamefont {Roger}\ \bibnamefont
  {Temam}},\ }\href@noop {} {\emph {\bibinfo {title} {Infinite-dimensional
  dynamical systems in mechanics and physics}}},\ Vol.~\bibinfo {volume} {68}\
  (\bibinfo  {publisher} {Springer-Verlag},\ \bibinfo {year}
  {1988})\BibitemShut {NoStop}%
\bibitem [{\citenamefont {Trefethen}(2000)}]{trefethen2000spectral}%
  \BibitemOpen
  \bibfield  {author} {\bibinfo {author} {\bibfnamefont {Lloyd~N}\ \bibnamefont
  {Trefethen}},\ }\href@noop {} {\emph {\bibinfo {title} {Spectral methods in
  MATLAB}}},\ Vol.~\bibinfo {volume} {10}\ (\bibinfo  {publisher} {Siam},\
  \bibinfo {year} {2000})\BibitemShut {NoStop}%
\bibitem [{\citenamefont {Stoica}\ \emph {et~al.}(2005)\citenamefont {Stoica},
  \citenamefont {Moses} \emph {et~al.}}]{stoica2005spectral}%
  \BibitemOpen
  \bibfield  {author} {\bibinfo {author} {\bibfnamefont {Petre}\ \bibnamefont
  {Stoica}}, \bibinfo {author} {\bibfnamefont {Randolph~L}\ \bibnamefont
  {Moses}},  \emph {et~al.},\ }\href@noop {} {\emph {\bibinfo {title} {Spectral
  analysis of signals}}},\ Vol.\ \bibinfo {volume} {452}\ (\bibinfo
  {publisher} {Pearson Prentice Hall Upper Saddle River, NJ},\ \bibinfo {year}
  {2005})\BibitemShut {NoStop}%
\bibitem [{\citenamefont {Ghil}\ \emph {et~al.}(2002)\citenamefont {Ghil},
  \citenamefont {Allen}, \citenamefont {Dettinger}, \citenamefont {Ide},
  \citenamefont {Kondrashov}, \citenamefont {Mann}, \citenamefont {Robertson},
  \citenamefont {Saunders}, \citenamefont {Tian}, \citenamefont {Varadi} \emph
  {et~al.}}]{ghil2002advanced}%
  \BibitemOpen
  \bibfield  {author} {\bibinfo {author} {\bibfnamefont {Michael}\ \bibnamefont
  {Ghil}}, \bibinfo {author} {\bibfnamefont {MR}~\bibnamefont {Allen}},
  \bibinfo {author} {\bibfnamefont {MD}~\bibnamefont {Dettinger}}, \bibinfo
  {author} {\bibfnamefont {K}~\bibnamefont {Ide}}, \bibinfo {author}
  {\bibfnamefont {D}~\bibnamefont {Kondrashov}}, \bibinfo {author}
  {\bibfnamefont {ME}~\bibnamefont {Mann}}, \bibinfo {author} {\bibfnamefont
  {Andrew~W}\ \bibnamefont {Robertson}}, \bibinfo {author} {\bibfnamefont
  {A}~\bibnamefont {Saunders}}, \bibinfo {author} {\bibfnamefont
  {Y}~\bibnamefont {Tian}}, \bibinfo {author} {\bibfnamefont {F}~\bibnamefont
  {Varadi}},  \emph {et~al.},\ }\bibfield  {title} {\enquote {\bibinfo {title}
  {Advanced spectral methods for climatic time series},}\ }\href@noop {}
  {\bibfield  {journal} {\bibinfo  {journal} {Reviews of geophysics}\ }\textbf
  {\bibinfo {volume} {40}} (\bibinfo {year} {2002})}\BibitemShut {NoStop}%
\bibitem [{\citenamefont {Li}\ and\ \citenamefont
  {Stoica}(1996)}]{li1996efficient}%
  \BibitemOpen
  \bibfield  {author} {\bibinfo {author} {\bibfnamefont {Jian}\ \bibnamefont
  {Li}}\ and\ \bibinfo {author} {\bibfnamefont {Petre}\ \bibnamefont
  {Stoica}},\ }\bibfield  {title} {\enquote {\bibinfo {title} {Efficient
  mixed-spectrum estimation with applications to target feature extraction},}\
  }\href@noop {} {\bibfield  {journal} {\bibinfo  {journal} {IEEE transactions
  on signal processing}\ }\textbf {\bibinfo {volume} {44}},\ \bibinfo {pages}
  {281--295} (\bibinfo {year} {1996})}\BibitemShut {NoStop}%
\bibitem [{\citenamefont {Stoica}\ \emph {et~al.}(1997)\citenamefont {Stoica},
  \citenamefont {Jakobsson},\ and\ \citenamefont {Li}}]{stoica1997cisoid}%
  \BibitemOpen
  \bibfield  {author} {\bibinfo {author} {\bibfnamefont {Petre}\ \bibnamefont
  {Stoica}}, \bibinfo {author} {\bibfnamefont {Andreas}\ \bibnamefont
  {Jakobsson}}, \ and\ \bibinfo {author} {\bibfnamefont {Jian}\ \bibnamefont
  {Li}},\ }\bibfield  {title} {\enquote {\bibinfo {title} {Cisoid parameter
  estimation in the colored noise case: asymptotic cramer-rao bound, maximum
  likelihood, and nonlinear least-squares},}\ }\href@noop {} {\bibfield
  {journal} {\bibinfo  {journal} {IEEE Transactions on Signal Processing}\
  }\textbf {\bibinfo {volume} {45}},\ \bibinfo {pages} {2048--2059} (\bibinfo
  {year} {1997})}\BibitemShut {NoStop}%
\bibitem [{\citenamefont {Wiener}\ and\ \citenamefont
  {Wintner}(1941)}]{wiener1941harmonic}%
  \BibitemOpen
  \bibfield  {author} {\bibinfo {author} {\bibfnamefont {Norbert}\ \bibnamefont
  {Wiener}}\ and\ \bibinfo {author} {\bibfnamefont {Aurel}\ \bibnamefont
  {Wintner}},\ }\bibfield  {title} {\enquote {\bibinfo {title} {Harmonic
  analysis and ergodic theory},}\ }\href@noop {} {\bibfield  {journal}
  {\bibinfo  {journal} {American Journal of Mathematics}\ }\textbf {\bibinfo
  {volume} {63}},\ \bibinfo {pages} {415--426} (\bibinfo {year}
  {1941})}\BibitemShut {NoStop}%
\bibitem [{\citenamefont {Mezi{\'c}}\ and\ \citenamefont
  {Sotiropoulos}(2002)}]{mezic2002ergodic}%
  \BibitemOpen
  \bibfield  {author} {\bibinfo {author} {\bibfnamefont {Igor}\ \bibnamefont
  {Mezi{\'c}}}\ and\ \bibinfo {author} {\bibfnamefont {Fotis}\ \bibnamefont
  {Sotiropoulos}},\ }\bibfield  {title} {\enquote {\bibinfo {title} {Ergodic
  theory and experimental visualization of invariant sets in chaotically
  advected flows},}\ }\href@noop {} {\bibfield  {journal} {\bibinfo  {journal}
  {Physics of Fluids (1994-present)}\ }\textbf {\bibinfo {volume} {14}},\
  \bibinfo {pages} {2235--2243} (\bibinfo {year} {2002})}\BibitemShut {NoStop}%
\bibitem [{\citenamefont {Krengel}\ and\ \citenamefont
  {Brunel}(1985)}]{krengel1985ergodic}%
  \BibitemOpen
  \bibfield  {author} {\bibinfo {author} {\bibfnamefont {Ulrich}\ \bibnamefont
  {Krengel}}\ and\ \bibinfo {author} {\bibfnamefont {Antoine}\ \bibnamefont
  {Brunel}},\ }\href@noop {} {\emph {\bibinfo {title} {Ergodic theorems}}},\
  Vol.~\bibinfo {volume} {59}\ (\bibinfo  {publisher} {Cambridge Univ Press},\
  \bibinfo {year} {1985})\BibitemShut {NoStop}%
\bibitem [{\citenamefont {Oppenheim}\ and\ \citenamefont
  {Buck}(1999)}]{oppenheim1999discrete}%
  \BibitemOpen
  \bibfield  {author} {\bibinfo {author} {\bibfnamefont {Schafer~RW}\
  \bibnamefont {Oppenheim}, \bibfnamefont {AV}}\ and\ \bibinfo {author}
  {\bibfnamefont {JR}~\bibnamefont {Buck}},\ }\href@noop {} {\emph {\bibinfo
  {title} {Discrete-time signal processing}}}\ (\bibinfo  {publisher}
  {Prentice-Hall, NJ},\ \bibinfo {year} {1999})\BibitemShut {NoStop}%
\bibitem [{\citenamefont {Jain}\ \emph {et~al.}(1979)\citenamefont {Jain},
  \citenamefont {Collins},\ and\ \citenamefont {Davis}}]{jain1979high}%
  \BibitemOpen
  \bibfield  {author} {\bibinfo {author} {\bibfnamefont {Vijay~K}\ \bibnamefont
  {Jain}}, \bibinfo {author} {\bibfnamefont {William~L}\ \bibnamefont
  {Collins}}, \ and\ \bibinfo {author} {\bibfnamefont {David~C}\ \bibnamefont
  {Davis}},\ }\bibfield  {title} {\enquote {\bibinfo {title} {High-accuracy
  analog measurements via interpolated fft},}\ }\href@noop {} {\bibfield
  {journal} {\bibinfo  {journal} {IEEE Transactions on Instrumentation and
  Measurement}\ }\textbf {\bibinfo {volume} {28}},\ \bibinfo {pages} {113--122}
  (\bibinfo {year} {1979})}\BibitemShut {NoStop}%
\bibitem [{\citenamefont {Andria}\ \emph {et~al.}(1989)\citenamefont {Andria},
  \citenamefont {Savino},\ and\ \citenamefont {Trotta}}]{andria1989windows}%
  \BibitemOpen
  \bibfield  {author} {\bibinfo {author} {\bibfnamefont {Gregorio}\
  \bibnamefont {Andria}}, \bibinfo {author} {\bibfnamefont {Mario}\
  \bibnamefont {Savino}}, \ and\ \bibinfo {author} {\bibfnamefont {Amerigo}\
  \bibnamefont {Trotta}},\ }\bibfield  {title} {\enquote {\bibinfo {title}
  {Windows and interpolation algorithms to improve electrical measurement
  accuracy},}\ }\href@noop {} {\bibfield  {journal} {\bibinfo  {journal} {IEEE
  Transactions on Instrumentation and Measurement}\ }\textbf {\bibinfo {volume}
  {38}},\ \bibinfo {pages} {856--863} (\bibinfo {year} {1989})}\BibitemShut
  {NoStop}%
\bibitem [{\citenamefont {Agrez}(2002)}]{agrez2002weighted}%
  \BibitemOpen
  \bibfield  {author} {\bibinfo {author} {\bibfnamefont {Dusan}\ \bibnamefont
  {Agrez}},\ }\bibfield  {title} {\enquote {\bibinfo {title} {Weighted
  multipoint interpolated dft to improve amplitude estimation of multifrequency
  signal},}\ }\href@noop {} {\bibfield  {journal} {\bibinfo  {journal} {IEEE
  Transactions on Instrumentation and Measurement}\ }\textbf {\bibinfo {volume}
  {51}},\ \bibinfo {pages} {287--292} (\bibinfo {year} {2002})}\BibitemShut
  {NoStop}%
\bibitem [{\citenamefont {De~Prony}(1795)}]{de1795essai}%
  \BibitemOpen
  \bibfield  {author} {\bibinfo {author} {\bibfnamefont {Baron Gaspard~Riche}\
  \bibnamefont {De~Prony}},\ }\bibfield  {title} {\enquote {\bibinfo {title}
  {Essai {\'e}xperimental et analytique: sur les lois de la dilatabilit{\'e} de
  fluides {\'e}lastique et sur celles de la force expansive de la vapeur de
  l’alkool,a diff{\'e}rentes temp{\'e}ratures},}\ }\href@noop {} {\bibfield
  {journal} {\bibinfo  {journal} {Journal de l’{\'e}cole polytechnique}\
  }\textbf {\bibinfo {volume} {1}},\ \bibinfo {pages} {24--76} (\bibinfo {year}
  {1795})}\BibitemShut {NoStop}%
\bibitem [{\citenamefont {Stoica}\ and\ \citenamefont
  {Nehorai}(1989)}]{stoica1989statistical}%
  \BibitemOpen
  \bibfield  {author} {\bibinfo {author} {\bibfnamefont {Petre}\ \bibnamefont
  {Stoica}}\ and\ \bibinfo {author} {\bibfnamefont {Arye}\ \bibnamefont
  {Nehorai}},\ }\bibfield  {title} {\enquote {\bibinfo {title} {Statistical
  analysis of two nonlinear least-squares estimators of sine-wave parameters in
  the colored-noise case},}\ }\href@noop {} {\bibfield  {journal} {\bibinfo
  {journal} {Circuits, Systems and Signal Processing}\ }\textbf {\bibinfo
  {volume} {8}},\ \bibinfo {pages} {3--15} (\bibinfo {year}
  {1989})}\BibitemShut {NoStop}%
\bibitem [{\citenamefont {Susuki}\ and\ \citenamefont
  {Mezi{\'c}}(2015)}]{susuki2015prony}%
  \BibitemOpen
  \bibfield  {author} {\bibinfo {author} {\bibfnamefont {Yoshihiko}\
  \bibnamefont {Susuki}}\ and\ \bibinfo {author} {\bibfnamefont {Igor}\
  \bibnamefont {Mezi{\'c}}},\ }\bibfield  {title} {\enquote {\bibinfo {title}
  {A prony approximation of koopman mode decomposition},}\ }in\ \href@noop {}
  {\emph {\bibinfo {booktitle} {Decision and Control (CDC), 2015 IEEE 54th
  Annual Conference on}}}\ (\bibinfo {organization} {IEEE},\ \bibinfo {year}
  {2015})\ pp.\ \bibinfo {pages} {7022--7027}\BibitemShut {NoStop}%
\bibitem [{\citenamefont {Schmidt}(1986)}]{schmidt1986multiple}%
  \BibitemOpen
  \bibfield  {author} {\bibinfo {author} {\bibfnamefont {Ralph}\ \bibnamefont
  {Schmidt}},\ }\bibfield  {title} {\enquote {\bibinfo {title} {Multiple
  emitter location and signal parameter estimation},}\ }\href@noop {}
  {\bibfield  {journal} {\bibinfo  {journal} {IEEE transactions on antennas and
  propagation}\ }\textbf {\bibinfo {volume} {34}},\ \bibinfo {pages} {276--280}
  (\bibinfo {year} {1986})}\BibitemShut {NoStop}%
\bibitem [{\citenamefont {Roy}\ and\ \citenamefont
  {Kailath}(1989)}]{roy1989esprit}%
  \BibitemOpen
  \bibfield  {author} {\bibinfo {author} {\bibfnamefont {Richard}\ \bibnamefont
  {Roy}}\ and\ \bibinfo {author} {\bibfnamefont {Thomas}\ \bibnamefont
  {Kailath}},\ }\bibfield  {title} {\enquote {\bibinfo {title}
  {Esprit-estimation of signal parameters via rotational invariance
  techniques},}\ }\href@noop {} {\bibfield  {journal} {\bibinfo  {journal}
  {IEEE Transactions on acoustics, speech, and signal processing}\ }\textbf
  {\bibinfo {volume} {37}},\ \bibinfo {pages} {984--995} (\bibinfo {year}
  {1989})}\BibitemShut {NoStop}%
\bibitem [{\citenamefont {Cand{\`e}s}\ and\ \citenamefont
  {Fernandez-Granda}(2014)}]{candes2014towards}%
  \BibitemOpen
  \bibfield  {author} {\bibinfo {author} {\bibfnamefont {Emmanuel~J}\
  \bibnamefont {Cand{\`e}s}}\ and\ \bibinfo {author} {\bibfnamefont {Carlos}\
  \bibnamefont {Fernandez-Granda}},\ }\bibfield  {title} {\enquote {\bibinfo
  {title} {Towards a mathematical theory of super-resolution},}\ }\href@noop {}
  {\bibfield  {journal} {\bibinfo  {journal} {Communications on Pure and
  Applied Mathematics}\ }\textbf {\bibinfo {volume} {67}},\ \bibinfo {pages}
  {906--956} (\bibinfo {year} {2014})}\BibitemShut {NoStop}%
\bibitem [{\citenamefont {Fernandez-Granda}(2016)}]{fernandez2016super}%
  \BibitemOpen
  \bibfield  {author} {\bibinfo {author} {\bibfnamefont {Carlos}\ \bibnamefont
  {Fernandez-Granda}},\ }\bibfield  {title} {\enquote {\bibinfo {title}
  {Super-resolution of point sources via convex programming},}\ }\href@noop {}
  {\bibfield  {journal} {\bibinfo  {journal} {Information and Inference: A
  Journal of the IMA}\ }\textbf {\bibinfo {volume} {5}},\ \bibinfo {pages}
  {251--303} (\bibinfo {year} {2016})}\BibitemShut {NoStop}%
\bibitem [{\citenamefont {Tropp}\ and\ \citenamefont
  {Gilbert}(2007)}]{tropp2007signal}%
  \BibitemOpen
  \bibfield  {author} {\bibinfo {author} {\bibfnamefont {Joel~A}\ \bibnamefont
  {Tropp}}\ and\ \bibinfo {author} {\bibfnamefont {Anna~C}\ \bibnamefont
  {Gilbert}},\ }\bibfield  {title} {\enquote {\bibinfo {title} {Signal recovery
  from random measurements via orthogonal matching pursuit},}\ }\href@noop {}
  {\bibfield  {journal} {\bibinfo  {journal} {IEEE Transactions on information
  theory}\ }\textbf {\bibinfo {volume} {53}},\ \bibinfo {pages} {4655--4666}
  (\bibinfo {year} {2007})}\BibitemShut {NoStop}%
\bibitem [{\citenamefont {Mallat}\ and\ \citenamefont
  {Zhang}(1993)}]{mallat1993matching}%
  \BibitemOpen
  \bibfield  {author} {\bibinfo {author} {\bibfnamefont {St{\'e}phane~G}\
  \bibnamefont {Mallat}}\ and\ \bibinfo {author} {\bibfnamefont {Zhifeng}\
  \bibnamefont {Zhang}},\ }\bibfield  {title} {\enquote {\bibinfo {title}
  {Matching pursuits with time-frequency dictionaries},}\ }\href@noop {}
  {\bibfield  {journal} {\bibinfo  {journal} {IEEE Transactions on signal
  processing}\ }\textbf {\bibinfo {volume} {41}},\ \bibinfo {pages}
  {3397--3415} (\bibinfo {year} {1993})}\BibitemShut {NoStop}%
\bibitem [{\citenamefont {Wang}(2006)}]{wang2006seismic}%
  \BibitemOpen
  \bibfield  {author} {\bibinfo {author} {\bibfnamefont {Yanghua}\ \bibnamefont
  {Wang}},\ }\bibfield  {title} {\enquote {\bibinfo {title} {Seismic
  time-frequency spectral decomposition by matching pursuit},}\ }\href@noop {}
  {\bibfield  {journal} {\bibinfo  {journal} {Geophysics}\ }\textbf {\bibinfo
  {volume} {72}},\ \bibinfo {pages} {V13--V20} (\bibinfo {year}
  {2006})}\BibitemShut {NoStop}%
\bibitem [{\citenamefont {Fannjiang}\ and\ \citenamefont
  {Liao}(2012)}]{fannjiang2012coherence}%
  \BibitemOpen
  \bibfield  {author} {\bibinfo {author} {\bibfnamefont {Albert}\ \bibnamefont
  {Fannjiang}}\ and\ \bibinfo {author} {\bibfnamefont {Wenjing}\ \bibnamefont
  {Liao}},\ }\bibfield  {title} {\enquote {\bibinfo {title} {Coherence
  pattern--guided compressive sensing with unresolved grids},}\ }\href@noop {}
  {\bibfield  {journal} {\bibinfo  {journal} {SIAM Journal on Imaging
  Sciences}\ }\textbf {\bibinfo {volume} {5}},\ \bibinfo {pages} {179--202}
  (\bibinfo {year} {2012})}\BibitemShut {NoStop}%
\bibitem [{\citenamefont {Mamandipoor}\ \emph {et~al.}(2016)\citenamefont
  {Mamandipoor}, \citenamefont {Ramasamy},\ and\ \citenamefont
  {Madhow}}]{mamandipoor2016newtonized}%
  \BibitemOpen
  \bibfield  {author} {\bibinfo {author} {\bibfnamefont {Babak}\ \bibnamefont
  {Mamandipoor}}, \bibinfo {author} {\bibfnamefont {Dinesh}\ \bibnamefont
  {Ramasamy}}, \ and\ \bibinfo {author} {\bibfnamefont {Upamanyu}\ \bibnamefont
  {Madhow}},\ }\bibfield  {title} {\enquote {\bibinfo {title} {Newtonized
  orthogonal matching pursuit: Frequency estimation over the continuum.}}\
  }\href@noop {} {\bibfield  {journal} {\bibinfo  {journal} {IEEE Trans. Signal
  Processing}\ }\textbf {\bibinfo {volume} {64}},\ \bibinfo {pages}
  {5066--5081} (\bibinfo {year} {2016})}\BibitemShut {NoStop}%
\bibitem [{\citenamefont {Bartlett}(1950)}]{bartlett1950periodogram}%
  \BibitemOpen
  \bibfield  {author} {\bibinfo {author} {\bibfnamefont {Maurice~S}\
  \bibnamefont {Bartlett}},\ }\bibfield  {title} {\enquote {\bibinfo {title}
  {Periodogram analysis and continuous spectra},}\ }\href@noop {} {\bibfield
  {journal} {\bibinfo  {journal} {Biometrika}\ }\textbf {\bibinfo {volume}
  {37}},\ \bibinfo {pages} {1--16} (\bibinfo {year} {1950})}\BibitemShut
  {NoStop}%
\bibitem [{\citenamefont {Blackman}\ and\ \citenamefont
  {Tukey}(1958)}]{blackman1958measurement}%
  \BibitemOpen
  \bibfield  {author} {\bibinfo {author} {\bibfnamefont {Ralph~Beebe}\
  \bibnamefont {Blackman}}\ and\ \bibinfo {author} {\bibfnamefont
  {John~Wilder}\ \bibnamefont {Tukey}},\ }\bibfield  {title} {\enquote
  {\bibinfo {title} {The measurement of power spectra from the point of view of
  communications engineering—part i},}\ }\href@noop {} {\bibfield  {journal}
  {\bibinfo  {journal} {Bell Labs Technical Journal}\ }\textbf {\bibinfo
  {volume} {37}},\ \bibinfo {pages} {185--282} (\bibinfo {year}
  {1958})}\BibitemShut {NoStop}%
\bibitem [{\citenamefont {GOVINDARAJAN}\ \emph {et~al.}(2017)\citenamefont
  {GOVINDARAJAN}, \citenamefont {Mohr}, \citenamefont {CHANDRASEKARAN},\ and\
  \citenamefont {Mezi\'c}}]{Nithin2017convergent}%
  \BibitemOpen
  \bibfield  {author} {\bibinfo {author} {\bibfnamefont {Niothin}\ \bibnamefont
  {GOVINDARAJAN}}, \bibinfo {author} {\bibfnamefont {Ryan}\ \bibnamefont
  {Mohr}}, \bibinfo {author} {\bibfnamefont {Shiv}\ \bibnamefont
  {CHANDRASEKARAN}}, \ and\ \bibinfo {author} {\bibfnamefont {Igor}\
  \bibnamefont {Mezi\'c}},\ }\bibfield  {title} {\enquote {\bibinfo {title} {A
  convergent numerical method for computing koopman spectra of
  volume-preserving maps on the d-torus},}\ }\href@noop {} {\bibfield
  {journal} {\bibinfo  {journal} {preprint}\ } (\bibinfo {year}
  {2017})}\BibitemShut {NoStop}%
\bibitem [{\citenamefont {Korda}\ \emph {et~al.}(2017)\citenamefont {Korda},
  \citenamefont {Putinar},\ and\ \citenamefont {Mezi{\'c}}}]{korda2017data}%
  \BibitemOpen
  \bibfield  {author} {\bibinfo {author} {\bibfnamefont {Milan}\ \bibnamefont
  {Korda}}, \bibinfo {author} {\bibfnamefont {Mihai}\ \bibnamefont {Putinar}},
  \ and\ \bibinfo {author} {\bibfnamefont {Igor}\ \bibnamefont {Mezi{\'c}}},\
  }\bibfield  {title} {\enquote {\bibinfo {title} {Data-driven spectral
  analysis of the koopman operator},}\ }\href@noop {} {\bibfield  {journal}
  {\bibinfo  {journal} {arXiv preprint arXiv:1710.06532}\ } (\bibinfo {year}
  {2017})}\BibitemShut {NoStop}%
\bibitem [{\citenamefont {Broer}\ and\ \citenamefont
  {Takens}(1993)}]{broer1993mixed}%
  \BibitemOpen
  \bibfield  {author} {\bibinfo {author} {\bibfnamefont {Henk}\ \bibnamefont
  {Broer}}\ and\ \bibinfo {author} {\bibfnamefont {Floris}\ \bibnamefont
  {Takens}},\ }\bibfield  {title} {\enquote {\bibinfo {title} {Mixed spectra
  and rotational symmetry},}\ }\href@noop {} {\bibfield  {journal} {\bibinfo
  {journal} {Archive for rational mechanics and analysis}\ }\textbf {\bibinfo
  {volume} {124}},\ \bibinfo {pages} {13--42} (\bibinfo {year}
  {1993})}\BibitemShut {NoStop}%
\bibitem [{\citenamefont {Ruelle}\ and\ \citenamefont
  {Takens}(1971)}]{ruelle1971nature}%
  \BibitemOpen
  \bibfield  {author} {\bibinfo {author} {\bibfnamefont {David}\ \bibnamefont
  {Ruelle}}\ and\ \bibinfo {author} {\bibfnamefont {Floris}\ \bibnamefont
  {Takens}},\ }\bibfield  {title} {\enquote {\bibinfo {title} {On the nature of
  turbulence},}\ }\href@noop {} {\bibfield  {journal} {\bibinfo  {journal}
  {Commun. math. phys}\ }\textbf {\bibinfo {volume} {20}},\ \bibinfo {pages}
  {167--192} (\bibinfo {year} {1971})}\BibitemShut {NoStop}%
\bibitem [{\citenamefont {Takens}(1981)}]{takens1981detecting}%
  \BibitemOpen
  \bibfield  {author} {\bibinfo {author} {\bibfnamefont {Floris}\ \bibnamefont
  {Takens}},\ }\bibfield  {title} {\enquote {\bibinfo {title} {Detecting
  strange attractors in turbulence},}\ }in\ \href@noop {} {\emph {\bibinfo
  {booktitle} {Dynamical systems and turbulence, Warwick 1980}}}\ (\bibinfo
  {publisher} {Springer},\ \bibinfo {year} {1981})\ pp.\ \bibinfo {pages}
  {366--381}\BibitemShut {NoStop}%
\bibitem [{\citenamefont {Mukolobwiez}\ \emph {et~al.}(1998)\citenamefont
  {Mukolobwiez}, \citenamefont {Chiffaudel},\ and\ \citenamefont
  {Daviaud}}]{mukolobwiez1998supercritical}%
  \BibitemOpen
  \bibfield  {author} {\bibinfo {author} {\bibfnamefont {Nathalie}\
  \bibnamefont {Mukolobwiez}}, \bibinfo {author} {\bibfnamefont {Arnaud}\
  \bibnamefont {Chiffaudel}}, \ and\ \bibinfo {author} {\bibfnamefont
  {Fran{\c{c}}ois}\ \bibnamefont {Daviaud}},\ }\bibfield  {title} {\enquote
  {\bibinfo {title} {Supercritical eckhaus instability for
  surface-tension-driven hydrothermal waves},}\ }\href@noop {} {\bibfield
  {journal} {\bibinfo  {journal} {Physical review letters}\ }\textbf {\bibinfo
  {volume} {80}},\ \bibinfo {pages} {4661} (\bibinfo {year}
  {1998})}\BibitemShut {NoStop}%
\bibitem [{\citenamefont {Garnier}\ \emph {et~al.}(2003)\citenamefont
  {Garnier}, \citenamefont {Chiffaudel}, \citenamefont {Daviaud},\ and\
  \citenamefont {Prigent}}]{garnier2003nonlinear}%
  \BibitemOpen
  \bibfield  {author} {\bibinfo {author} {\bibfnamefont {Nicolas}\ \bibnamefont
  {Garnier}}, \bibinfo {author} {\bibfnamefont {Arnaud}\ \bibnamefont
  {Chiffaudel}}, \bibinfo {author} {\bibfnamefont {Fran{\c{c}}ois}\
  \bibnamefont {Daviaud}}, \ and\ \bibinfo {author} {\bibfnamefont {Arnaud}\
  \bibnamefont {Prigent}},\ }\bibfield  {title} {\enquote {\bibinfo {title}
  {Nonlinear dynamics of waves and modulated waves in 1d thermocapillary flows.
  i. general presentation and periodic solutions},}\ }\href@noop {} {\bibfield
  {journal} {\bibinfo  {journal} {Physica D: Nonlinear Phenomena}\ }\textbf
  {\bibinfo {volume} {174}},\ \bibinfo {pages} {1--29} (\bibinfo {year}
  {2003})}\BibitemShut {NoStop}%
\bibitem [{\citenamefont {Mercier}\ \emph {et~al.}(2008)\citenamefont
  {Mercier}, \citenamefont {Garnier},\ and\ \citenamefont
  {Dauxois}}]{mercier2008reflection}%
  \BibitemOpen
  \bibfield  {author} {\bibinfo {author} {\bibfnamefont {Matthieu~J}\
  \bibnamefont {Mercier}}, \bibinfo {author} {\bibfnamefont {Nicolas~B}\
  \bibnamefont {Garnier}}, \ and\ \bibinfo {author} {\bibfnamefont {Thierry}\
  \bibnamefont {Dauxois}},\ }\bibfield  {title} {\enquote {\bibinfo {title}
  {Reflection and diffraction of internal waves analyzed with the hilbert
  transform},}\ }\href@noop {} {\bibfield  {journal} {\bibinfo  {journal}
  {Physics of Fluids}\ }\textbf {\bibinfo {volume} {20}},\ \bibinfo {pages}
  {086601} (\bibinfo {year} {2008})}\BibitemShut {NoStop}%
\bibitem [{\citenamefont {Noack}\ \emph {et~al.}(2003)\citenamefont {Noack},
  \citenamefont {Afanasiev}, \citenamefont {Morzynski}, \citenamefont
  {Tadmor},\ and\ \citenamefont {Thiele}}]{noack2003hierarchy}%
  \BibitemOpen
  \bibfield  {author} {\bibinfo {author} {\bibfnamefont {Bernd~R}\ \bibnamefont
  {Noack}}, \bibinfo {author} {\bibfnamefont {Konstantin}\ \bibnamefont
  {Afanasiev}}, \bibinfo {author} {\bibfnamefont {Marek}\ \bibnamefont
  {Morzynski}}, \bibinfo {author} {\bibfnamefont {Gilead}\ \bibnamefont
  {Tadmor}}, \ and\ \bibinfo {author} {\bibfnamefont {Frank}\ \bibnamefont
  {Thiele}},\ }\bibfield  {title} {\enquote {\bibinfo {title} {A hierarchy of
  low-dimensional models for the transient and post-transient cylinder wake},}\
  }\href@noop {} {\bibfield  {journal} {\bibinfo  {journal} {Journal of Fluid
  Mechanics}\ }\textbf {\bibinfo {volume} {497}},\ \bibinfo {pages} {335--363}
  (\bibinfo {year} {2003})}\BibitemShut {NoStop}%
\bibitem [{\citenamefont {Berkooz}\ \emph {et~al.}(1993)\citenamefont
  {Berkooz}, \citenamefont {Holmes},\ and\ \citenamefont
  {Lumley}}]{berkooz1993proper}%
  \BibitemOpen
  \bibfield  {author} {\bibinfo {author} {\bibfnamefont {Gal}\ \bibnamefont
  {Berkooz}}, \bibinfo {author} {\bibfnamefont {Philip}\ \bibnamefont
  {Holmes}}, \ and\ \bibinfo {author} {\bibfnamefont {John~L}\ \bibnamefont
  {Lumley}},\ }\bibfield  {title} {\enquote {\bibinfo {title} {The proper
  orthogonal decomposition in the analysis of turbulent flows},}\ }\href@noop
  {} {\bibfield  {journal} {\bibinfo  {journal} {Annual review of fluid
  mechanics}\ }\textbf {\bibinfo {volume} {25}},\ \bibinfo {pages} {539--575}
  (\bibinfo {year} {1993})}\BibitemShut {NoStop}%
\bibitem [{\citenamefont {Lopez}\ and\ \citenamefont
  {Marques}(2000)}]{lopez2000dynamics}%
  \BibitemOpen
  \bibfield  {author} {\bibinfo {author} {\bibfnamefont {JM}~\bibnamefont
  {Lopez}}\ and\ \bibinfo {author} {\bibfnamefont {F}~\bibnamefont {Marques}},\
  }\bibfield  {title} {\enquote {\bibinfo {title} {Dynamics of three-tori in a
  periodically forced navier-stokes flow},}\ }\href@noop {} {\bibfield
  {journal} {\bibinfo  {journal} {Physical review letters}\ }\textbf {\bibinfo
  {volume} {85}},\ \bibinfo {pages} {972} (\bibinfo {year} {2000})}\BibitemShut
  {NoStop}%
\bibitem [{\citenamefont {Takeda}(1999)}]{takeda1999quasi}%
  \BibitemOpen
  \bibfield  {author} {\bibinfo {author} {\bibfnamefont {Yasushi}\ \bibnamefont
  {Takeda}},\ }\bibfield  {title} {\enquote {\bibinfo {title} {Quasi-periodic
  state and transition to turbulence in a rotating couette system},}\
  }\href@noop {} {\bibfield  {journal} {\bibinfo  {journal} {Journal of Fluid
  Mechanics}\ }\textbf {\bibinfo {volume} {389}},\ \bibinfo {pages} {81--99}
  (\bibinfo {year} {1999})}\BibitemShut {NoStop}%
\bibitem [{\citenamefont {Guzm{\'a}n}\ and\ \citenamefont
  {Amon}(1994)}]{guzman1994transition}%
  \BibitemOpen
  \bibfield  {author} {\bibinfo {author} {\bibfnamefont {AM}~\bibnamefont
  {Guzm{\'a}n}}\ and\ \bibinfo {author} {\bibfnamefont {CH}~\bibnamefont
  {Amon}},\ }\bibfield  {title} {\enquote {\bibinfo {title} {Transition to
  chaos in converging--diverging channel flows: Ruelle--takens--newhouse
  scenario},}\ }\href@noop {} {\bibfield  {journal} {\bibinfo  {journal}
  {Physics of Fluids}\ }\textbf {\bibinfo {volume} {6}},\ \bibinfo {pages}
  {1994--2002} (\bibinfo {year} {1994})}\BibitemShut {NoStop}%
\end{thebibliography}
\end{document}